\documentclass[a4paper, 11pt]{article}

\usepackage[english]{babel}
\usepackage[T1]{fontenc}
\usepackage[utf8]{inputenc}
\usepackage{amssymb}
\usepackage{array,amsmath, mleftright}
\usepackage{amsthm}
\usepackage{setspace}
\usepackage{tcolorbox}
\usepackage{blindtext}
\usepackage{mathrsfs}
\newcommand{\setvarphi}{\mathscr{B}}
\newcommand{\objweight}{w}
\usepackage{mathtools}
\usepackage[breaklinks=true, hidelinks]{hyperref}
\usepackage{bbding}
\usepackage{framed}
\usepackage{xcolor,colortbl}
\usepackage{hhline}

\usepackage{dsfont}
\usepackage{cases}
\usepackage{pdfpages}
\usepackage{empheq}

\newcommand*\diff{\mathop{}\!\mathrm{d}}
\newcommand{\indicate}[1]{\mathds{1}\small{[#1]}}

\usepackage{bm}
\usepackage[inkscapearea=page]{svg}

\usepackage{caption}[font=small,skip=0pt]

\newtheorem{theorem}{Theorem}
\newtheorem{lemma}{Lemma}
\newtheorem{assumptions}{Assumption}
\newtheorem{proposition}{Proposition}
\newtheorem{corollary}{Corollary}
\newtheorem{observation}{Observation}

\newtheorem{example}{Example}

\newtheorem*{examplectd}{Example~\ref*{example:recurred}}

\newtheorem{definition}{Definition}

\allowdisplaybreaks

\makeatletter
\newcommand{\support}[1]{\aslv@support#1\@nil}
\def\aslv@support#1\\#2|#3\@nil{%
  \mleft(
  \left( \genfrac..{0pt}{}{#1}{#2}
  \mskip\aslv@supportskip \right)
  \middle|
  \mskip\aslv@supportskip
  #3
  \mright)
}
\newmuskip\aslv@supportskip
\makeatother
\makeatletter
\newcommand{\pushright}[1]{\ifmeasuring@#1\else\omit\hfill$\displaystyle#1$\fi\ignorespaces}
\makeatother
\usepackage{graphicx}
\usepackage{tikz}
\usetikzlibrary{arrows.meta}
\usetikzlibrary{bending}
\usetikzlibrary{decorations.pathreplacing}
\usetikzlibrary{shapes}
\usetikzlibrary{patterns}

\usepackage{pgfplots}
\pgfplotsset{compat=1.15}
\usepackage{adjustbox}
\usepackage{subcaption}

\usepackage{natbib}
 \bibpunct[, ]{(}{)}{,}{a}{}{,}%
\usepackage{multirow}
\usepackage{booktabs}
\usepackage{longtable}

\usepackage{tabu}
\usepackage{pdflscape}
\makeatletter
\newbox\LT@firstfoot
\def\endfirstfoot{\LT@end@hd@ft\LT@firstfoot}
\newdimen\LT@footdiff
\def\LT@start{%
  \let\LT@start\endgraf
  \endgraf\penalty\z@
  \vskip\LTpre
  \LT@footdiff-\ht\LT@foot
  \advance\LT@footdiff\ht\LT@firstfoot
  \dimen@\pagetotal
  \advance\dimen@ \ht\ifvoid\LT@firsthead\LT@head\else\LT@firsthead\fi
  \advance\dimen@ \dp\ifvoid\LT@firsthead\LT@head\else\LT@firsthead\fi
  \advance\dimen@ \ht\ifvoid\LT@firstfoot\LT@foot\else\LT@firstfoot\fi
  \dimen@ii\vfuzz
  \vfuzz\maxdimen
  \setbox\tw@\copy\z@
  \setbox\tw@\vsplit\tw@ to \ht\@arstrutbox
  \setbox\tw@\vbox{\unvbox\tw@}%
  \vfuzz\dimen@ii
  \advance\dimen@ \ht
      \ifdim\ht\@arstrutbox>\ht\tw@\@arstrutbox\else\tw@\fi
  \advance\dimen@\dp
      \ifdim\dp\@arstrutbox>\dp\tw@\@arstrutbox\else\tw@\fi
  \advance\dimen@ -\pagegoal
  \ifdim \dimen@>\z@\vfil\break\fi
  \global\@colroom\@colht
  \ifvoid\LT@firstfoot
    \ifvoid\LT@foot
    \else
      \advance\vsize-\ht\LT@foot
      \global\advance\@colroom-\ht\LT@foot
      \dimen@\pagegoal\advance\dimen@-\ht\LT@foot\pagegoal\dimen@
      \maxdepth\z@
    \fi
  \else
    \advance\vsize-\ht\LT@firstfoot
    \global\advance\@colroom-\ht\LT@firstfoot
    \dimen@\pagegoal\advance\dimen@-\ht\LT@firstfoot\pagegoal\dimen@
    \maxdepth\z@
  \fi
  \ifvoid\LT@firsthead\copy\LT@head\else\box\LT@firsthead\fi\nobreak
  \output{\LT@output}%
}
\def\LT@output{%
  \ifnum\outputpenalty <-\@Mi
    \ifnum\outputpenalty > -\LT@end@pen
      \LT@err{floats and marginpars not allowed in a longtable}\@ehc
    \else
      \setbox\z@\vbox{\unvbox\@cclv}%
      \ifdim \ht\LT@lastfoot>\ht\LT@foot
        \dimen@\pagegoal
        \advance\dimen@-\ht\LT@lastfoot
        \ifdim\dimen@<\ht\z@
          \setbox\@cclv\vbox{\unvbox\z@\copy\LT@foot\vss}%
          \@makecol
          \@outputpage
          \setbox\z@\vbox{\box\LT@head}%
        \fi
      \fi  
      \global\@colroom\@colht
      \global\vsize\@colht   
      \vbox
        {\unvbox\z@\box\ifvoid\LT@lastfoot\LT@foot\else\LT@lastfoot\fi}%
    \fi
  \else
    \ifvoid\LT@firstfoot
      \setbox\@cclv\vbox{\unvbox\@cclv\copy\LT@foot\vss}%
      \@makecol
      \@outputpage
      \global\vsize\@colroom
    \else
      \setbox\@cclv\vbox{\unvbox\@cclv\box\LT@firstfoot\vss}%
      \@makecol
      \@outputpage
      \global\advance\@colroom\LT@footdiff
      \global\vsize\@colroom
    \fi
    \copy\LT@head\nobreak
  \fi
}
\makeatother

\usepackage{lipsum}

\usepackage[shortlabels]{enumitem}

\usepackage{todonotes}

\usepackage[ruled, vlined]{algorithm2e}
\SetKwRepeat{Do}{do}{while}

\usepackage{algpseudocode}
\usepackage{float}

\usepackage[top=25mm, bottom=25mm, left=25mm, right=25mm]{geometry}

\newcommand \Template[2]{
\usepackage{fancyhdr}
\pagestyle{fancy}
\fancyhf{}
\fancyhead[RO]{#1}
\fancyhead[LO]{#2}
\fancyfoot[CO]{\thepage}

\renewcommand{\headrulewidth}{1pt}
\renewcommand{\footrulewidth}{1pt}}

\usepackage{titling}
\usepackage{authblk}
\newcommand \linedabstract[1]{
  \renewcommand{\maketitlehookd}{%
    \mbox{}\medskip\par
    \centering
    \hrule\medskip
    \begin{minipage}{0.9\textwidth} 
    {\small\textbf{Abstract}\\ #1}
    \end{minipage}\medskip\hrule
    }      
}

\newcommand \ACKNOWLEDGMENT[1]{
	\section*{Acknowledgments}
	#1
}

\definecolor{darkred}{rgb}{0.8,0,0}



\usepackage{soul}

\definecolor{airforceblue}{rgb}{0.36, 0.54, 0.66}
\definecolor{darkbrown}{rgb}{0.4, 0.26, 0.13}

\title{\textbf{Differential Privacy via \\ Distributionally Robust Optimization}}
\author[1]{Aras Selvi\thanks{Corresponding author: \href{mailto:a.selvi19@imperial.ac.uk}{a.selvi19@imperial.ac.uk}}} 
\affil[1]{\small Imperial College Business School, Imperial College London, United Kingdom}

\author[2]{Huikang Liu}
\affil[2]{\small \textcolor{black}{Antai College of Economics and Management, Shanghai Jiao Tong University, China}}

\author[1]{Wolfram Wiesemann}

\linedabstract{In recent years, differential privacy has emerged as the \emph{de facto} standard for sharing statistics of datasets while limiting the disclosure of private information about the involved individuals. This is achieved by randomly perturbing the statistics to be published, which in turn leads to a privacy-accuracy trade-off: larger perturbations provide stronger privacy guarantees, but they result in less accurate statistics that offer lower utility to the recipients. Of particular interest are therefore optimal mechanisms that provide the highest accuracy for a pre-selected level of privacy. To date, work in this area has focused on specifying families of perturbations \emph{a priori} and subsequently proving their asymptotic and/or best-in-class optimality. \\
\hphantom{$\quad$} In this paper, we develop a class of mechanisms that enjoy non-asymptotic and unconditional optimality guarantees. To this end, we formulate the mechanism design problem as an infinite-dimensional distributionally robust optimization problem. We show that the problem affords a strong dual, and we exploit this duality to develop converging hierarchies of finite-dimensional upper and lower bounding problems. Our upper (primal) bounds correspond to implementable perturbations whose suboptimality can be bounded by our lower (dual) bounds. Both bounding problems can be solved within seconds via cutting plane techniques that exploit the inherent problem structure. Our numerical experiments demonstrate that our perturbations can outperform the previously best results from the literature on artificial as well as standard benchmark problems. \\

\hphantom{a} $\mspace{82mu}$ Optimization.
}

\Template{Selvi, Liu, and Wiesemann (2024)}{Differential Privacy via DRO}{}
\begin{document}

\maketitle

\newpage

\section{Introduction}\label{sec:intro}
When organizations collect personal data about individuals, it is their responsibility to protect that data when they share information about it with third parties. Data anonymization, which aims to alter the data so that individuals are no longer identifiable, is often insufficient in that regard as it is prone to, among others, reconstruction attacks \citep{dinur2003revealing} and de-identification attacks \citep{sweeney1997weaving,pro2014privacy}. To address this issue, manifold definitions of data privacy have been proposed, including $k$-map \citep[\S 4.3]{sweeney2001computational}, $k$-anonymity \citep{sweeney2002k}, $\ell$-diversity \citep{machanavajjhala2007diversity} and $\delta$-presence \citep{nergiz2007hiding}; see also the review of \citet[\S 2.1]{desfontaines2020lowering}. Among those, \textit{differential privacy} (DP), first proposed by \cite{dwork2006calibrating}, has arguably received the most attention among researchers and practitioners.

DP considers databases $D \in {\color{black}\mathcal{D}}${\color{black}, where $\mathcal{D} := \mathbb{U}^n$ denotes the set of databases} with $n$ individuals (or rows), each of which stems from a data universe ${\color{black}\mathbb{U}}$ that characterizes the admissible attribute vectors (\emph{e.g.}, the possible values of the predictors and the output for a supervised learner). For any $\varepsilon, \delta \geq 0$, a randomized algorithm $\mathcal{A}$ mapping databases $D \in {\color{black}\mathcal{D}}$ to random outputs $\omega \in \Omega$ is \emph{$(\varepsilon, \delta)$-differentially private} if
\begin{equation*}
    \mathbb{P} [\mathcal{A}(D) \in A] \leq e^\varepsilon \cdot \mathbb{P} [\mathcal{A}(D') \in A] + \delta
    \qquad \forall (D,D') \in \mathcal{N}, \; \forall A \in \mathcal{F}, 
\end{equation*}
where $(\Omega, \mathcal{F}, \mathbb{P})$ is a probability space and
\begin{equation*}
    \mathcal{N} = \{(D, D') \in {\color{black}\mathcal{D}} \times {\color{black}\mathcal{D}} \ : \ D' = (D_{-k}, d) \text{ for some } k=1,\ldots,n \text{ and } d \in{\color{black}\mathbb{U}} \}
\end{equation*}
denotes the (symmetric) set of neighboring databases $(D, D')$, where $D'$ emerges from $D$ by replacing its $k$-th element with any $d \in {\color{black}\mathbb{U}}$ \citep{dwork2006our, dwork2006calibrating}. Intuitively, under DP with $\delta = 0$, an adversary cannot confidently estimate any single row of the database $D$ from a single sample $\omega \sim \mathcal{A} (D)$ even if she knows all other rows of $D$ and the implementation of $\mathcal{A}$ \citep{Tu13}. In the general case where $\delta > 0$, $(\varepsilon, \delta)$-DP is a sufficient (but not necessary) condition for the aforementioned $(\varepsilon, 0)$-DP to hold with a probability of at least $1 - \delta$ \citep{dinur2003revealing, meiser2018approximate, canonne2020discrete}. Viewed through a Bayesian lens, $(\varepsilon, \delta)$-DP allows the adversary to update her prior on $D$ by at most an amount that is bounded by a function of $\varepsilon$ and $\delta$ upon seeing a realization of $\mathcal{A} (D)$, see \citet[\S 1.6]{vadhan2017complexity}.

Compared to other notions of privacy, DP enjoys several desirable features. The composition theorem \citep[Theorem 3.16]{dwork2014algorithmic}, for example, implies that sharing $k$ different statistics, where statistic $i$ has been generated by a $(\varepsilon_i, \delta_i)$-DP mechanism, $i = 1, \ldots, k$, satisfies $(\sum_{i=1}^k \varepsilon_i, \sum_{i=1}^k \delta_i)$-DP. Likewise, the post processing property asserts that any analysis derived from the output of a differentially private mechanism remains differentially private with the same privacy guarantees \citep[Proposition 2.1]{dwork2014algorithmic}. Due to these and other features, DP has found manifold recent applications in statistics and machine learning \citep{chaudhuri2009logit, friedman2010data, chaudhuri2011differentially, abadi2016deep, cai2019econometrics}, optimization \citep{mangasarian2011privacy, hsu2014privately, han2016differentially, hsu2016jointly}, mechanism design \citep{mcsherry2007mechanism} and revenue management \citep{chen2021privacy,chen2023differential,lei2024privacy}. DP has also been widely applied in industry, ranging from emoji recommender systems that learn from user behavior \citep{appleprivacy}, databases that publish user interactions on Facebook \citep{metapaper} and COVID-19 vaccination search insights at Google \citep{googlevaccination} to insights from LinkedIn's Economic Graph \citep{rogers2020members}, the U.S.~Broadband Coverage dataset posted by Microsoft \citep{microsoft} and the earnings distribution published by the U.S.~Census Bureau \citep{foote2019releasing}.

In this paper, we study algorithms that perturb the output of a scalar \textit{query function} $f : {\color{black}\mathcal{D}} \mapsto \mathbb{R}$ so as to guarantee $(\varepsilon, \delta)$-DP. The query function $f$ could be a simple statistical query, such as the average, the median or a quantile of a real-valued attribute across all rows, a count of the rows satisfying a user-specified condition, or it could be part of a machine learning model that is trained on the database. {\color{black}We illustrate our setting with the following motivating example.

\begin{example}\label{example:recurred}
    The popular Kaggle \emph{salary dataset}\footnote{URL: \url{https://www.kaggle.com/datasets/mohithsairamreddy/salary-data/data}} 
    reports the monthly salaries (in thousands of Indian Rupees; INR) of $6,704$ individuals whose features include education level and job title. Consider the query function $f$ that computes the average salary of PhD graduates working in research. Across the $194$ PhD researchers in the database, the salary varies between $120k$ INR and $190k$ INR, and the average salary amounts to $165.65k$ INR. Returning this average would violate DP, however, since an adversary with knowledge of the salaries of the first $193$ PhD researchers could readily compute the salary of the $194$th PhD researcher from the query result.
\end{example}
}

The Laplace mechanism \citep{dwork2006calibrating} achieves $(\varepsilon, 0)$-DP, also referred to as \textit{pure} differential privacy, by returning $f(D) + \tilde{X}$, where the random variable $\tilde{X}$ follows a zero-mean Laplace distribution with scale parameter $\Delta f / \varepsilon$ and $\Delta f := \sup_{(D, D') \in \mathcal{N}} |f(D) - f(D')|$ denoting the sensitivity of the query $f$.

{\color{black}
\begin{examplectd}[cont'd]
Assume that the salaries of PhD researchers vary between $120k$ INR and $190k$ INR, which are the minimum and maximum PhD researcher salaries recorded in the \emph{salary dataset}. Thus, the sensitivity of the average PhD researcher salary query is $\Delta f = \frac{190k-120k}{194} \text{ INR} \approxeq 0.36k \text{ INR}$, which is attained by any two databases $(D, D') \in \mathcal{N}$ that differ in the salary of one PhD researcher, with one database recording a salary of $120k$ INR and the other reporting a salary of $190k$ INR. To achieve $(1, 0)$-DP, the Laplace mechanism returns as average salary the sum of $165.65k$ INR and the realization of a zero-mean Laplace distribution with scale parameter $\Delta f / \varepsilon \approxeq 0.36k \text{ INR}$. Assume that this sum amounts to $165.5k$ INR. Equipped with her knowledge of the salaries of the first $193$ PhD researchers, the parameters $\Delta f$ and $\varepsilon$ of the Laplace mechanism as well as the query response of $165.5k$ INR, the adversary could try to deduce the salary of the $194$th PhD researcher via maximum likelihood estimation. The likelihood of the $194$th PhD researcher's salary being $s$, given the above information, is $L (s) \approxeq p (165.5k \, \vert \, [31{,}966.10k + s] / 194, 0.36k)$, where $p (\cdot \, \vert \, \mu, b)$ is the density function of a Laplace distribution with location parameter $\mu$ and scale parameter $b$ and $31{,}966.10k$ INR is the sum of the first $193$ PhD researcher salaries (that are known to the adversary). The maximum likelihood estimator is $s^\star = 140.9k$ INR, which is exactly the value that makes the average across the first $193$ PhD researcher salaries and the unknown last salary equal to the observed query output of $165.5k$ INR. One readily computes that $L (s^\star) \approxeq 1.389$ while $\min \{ L(s) \, : \, s \in [120k, 190k] \} \approxeq 0.688$. The ratio between the maximum and the minimum likelihood is $1.389 / 0.688 \approxeq \exp(0.7)$, and DP guarantees that this ratio is bounded from above by $\exp(\varepsilon)$. In other words, even if the adversary knew all but one of the PhD researcher salaries, she could not confidently estimate the unknown salary from a single observation of the query output.
\end{examplectd}
}

Pure differential privacy bounds the probability ratio of outputs within any measurable set $A \in \mathcal{F}$, no matter how small the involved probabilities are. This significantly restricts the design of admissible algorithms $\mathcal{A}$, particularly in their tail behavior, and it has spurred research into other DP notions that relax the privacy requirement for unlikely events \citep{desfontaines2019sok}. The most prominent notion is $(\varepsilon, \delta)$-DP, also known as \textit{approximate} differential privacy. The Gaussian mechanism \citep[Appendix A]{dwork2014algorithmic}, for example, achieves $(\varepsilon, \delta)$-DP for any $\varepsilon, \delta \in (0,1)$ by returning $f(D) + \tilde{Y}$, where the random variable $\tilde{Y}$ follows a zero-mean Gaussian distribution with variance $2 \ln(1.25 / \delta) (\Delta f / \varepsilon)^2$. {\color{black} In Example~\ref{example:recurred}, the $(1, 0)$-DP Laplace mechanism adds a noise with standard deviation $510.28$ INR to the query result, whereas the $(1, 0.2)$-DP Gaussian mechanism---despite satisfying a relaxed notion of DP---increases the standard deviation of the noise term to $689.21$ INR.}


The Laplace and Gaussian mechanisms are \emph{data independent additive noise mechanisms} as their additive noises $\tilde{X}$ and $\tilde{Y}$ do not depend on the database $D$. In contrast, the noise of a \emph{data dependent mechanism} may depend on the database $D$. One of the earliest data dependent mechanisms is the exponential mechanism~\citep{mcsherry2007mechanism}, which achieves $(\varepsilon, 0)$-DP for query functions with nominal outputs. Instead of adding data independent noise to the query output, the exponential mechanism ensures that the output is {\color{black}in the range of nominal outputs} by randomly selecting one of finitely many outputs according to some score function. For query functions with real-valued outputs, smooth sensitivity mechanisms \citep{nissim2007smooth} achieve $(\varepsilon, \delta)$-DP at a higher accuracy than their data independent counterparts by adding noise whose variance is smaller for databases in neighborhoods with a low variation of the query outputs.

Algorithms with stronger privacy guarantees tend to offer less utility from data \citep{alvim2011differential}. This is clearly seen for the Laplace and Gaussian mechanisms, whose variances increase with smaller values of $\varepsilon$ and $\delta$. It is therefore natural to study whether those mechanisms are optimal, that is, whether they minimize a pre-specified loss function among the respective classes of $(\varepsilon, 0)$-DP and $(\varepsilon, \delta)$-DP algorithms. {\color{black} In particular, since the $(1, 0.2)$-DP Gaussian mechanism results in a larger standard deviation than the $(1, 0)$-DP Laplace mechanism in Example~\ref{example:recurred}, we conclude that Gaussian mechanisms are not optimal for the $\ell_2$-loss function in general.} The early work on optimal mechanisms has focused on specific query functions (such as count queries), and is reviewed, among others, by \citet[\S1]{geng2014optimal} and \citet[\S4.6]{sommer2021fighting}. Among the first papers that investigate optimal mechanisms for generic query functions is the work of \cite{soria2013optimal}, who show that for large classes of $(\varepsilon, 0)$-DP data independent additive noise mechanisms and loss functions, a necessary optimality condition is that no probability mass in the distribution of the random noise can be moved towards zero without violating the privacy guarantee. The Laplace mechanism violates this condition and is thus not optimal, whereas the condition is satisfied by piecewise constant `staircase distributions' that move the probability mass of the Laplace distribution closer to zero. \cite{geng2014optimal} show that for classes of $\ell_1$- and $\ell_2$-loss functions, the Laplace mechanism is asymptotically optimal as $\varepsilon\rightarrow 0$, but that it can be significantly suboptimal for larger values of $\varepsilon$. They also propose an $(\varepsilon, 0)$-DP data independent additive noise mechanism based on piecewise constant staircase distributions that is optimal across all $(\varepsilon, 0)$-DP algorithms for a large class of symmetric and increasing loss functions. While determining the optimal staircase distribution generally requires the tuning of a single parameter, closed-form characterizations \mbox{are provided for the special cases of $\ell_1$- and $\ell_2$-loss functions.}

In contrast, the design of optimal mechanisms for $\delta \neq 0$ is much less well understood. For the special case where $\varepsilon = 0$ (also known as additive differential privacy), \cite{geng2019optimal} show that sampling the additive noise from the product of a uniform and a Bernoulli random variable is optimal for symmetric and increasing loss functions among the class of $(0, \delta)$-DP data independent additive noise mechanisms with decreasing noise distributions. Closed-form characterizations of the optimal distributions are provided for $\ell_p$-loss functions, whereas the general case requires the tuning of a single parameter. \cite{balle2018improving} show that the parameter choice of the aforementioned Gaussian mechanism is suboptimal. They propose the analytic Gaussian mechanism, which is optimal among the family of Gaussian mechanisms, by numerically computing the smallest variance that satisfies $(\varepsilon, \delta)$-DP. For the general class of $(\varepsilon, \delta)$-DP data independent additive noise mechanisms, \cite{geng2015optimal} show that the suboptimality of uniform and discretized Laplace distributions can be bounded by a multiplicative constant for integer-valued queries under $\ell_1$- and $\ell_2$-loss functions when $\varepsilon \rightarrow 0$ and $\delta \rightarrow 0$ simultaneously. Tighter suboptimality bounds have been derived by \cite{geng2020tight} for real-valued queries when the data independent additive noise is governed by a truncated Laplace distribution. In their analysis, the authors decompose the support of the distribution into a `body' that achieves $(\varepsilon, 0)$-DP and a `tail' that breaches privacy but is limited to a probability mass of $\delta$. The resulting mechanisms are $(\varepsilon, \delta)$-DP, and they are asymptotically optimal under $\ell_1$- and $\ell_2$-loss functions when $\varepsilon \rightarrow 0$ and $\delta \rightarrow 0$ simultaneously. The authors show that the truncated Laplace mechanism outperforms the aforementioned analytic Gaussian mechanism under $\ell_1$- and $\ell_2$-loss functions. To our best knowledge, the truncated Laplace mechanism provides the strongest optimality guarantees among the \mbox{currently known $(\varepsilon, \delta)$-DP mechanisms for generic queries.} {\color{black} In Example~\ref{example:recurred}, the noise added by the $(1, 0.2)$-DP analytic Gaussian mechanism has a standard deviation of $300.96$ INR, whereas the truncated Laplace mechanism reduces the standard deviation to $273.48$ INR.}

The definition of optimality is more involved for data dependent mechanisms since their expected losses vary with the database. Minimizing the worst-case expected loss across all potential databases $D \in {\color{black}\mathcal{D}}$ (which the \textit{minimax optimality} criterion attempts to achieve) is overly conservative since it implies that data independent mechanisms remain optimal under mild conditions \citep{geng2014optimal}. Instead, instance specific optimality criteria have been proposed that compare the expected loss for each database with a lower bound that is tailored to the database. Local minimax optimality \citep{asi2020instance, asi2020near}, for example, requires that a mechanism's expected loss for any database $D \in {\color{black}\mathcal{D}}$ is within a constant factor of the expected loss of any other DP mechanism for at least one database in a neighborhood of $D$. Local minimax optimality recovers the earlier notion of minimax optimality when the neighborhood contains all databases $D \in {\color{black}\mathcal{D}}$. \citet{asi2020instance, asi2020near} show that under the expected $\ell_1$-loss, the inverse sensitivity mechanism first proposed by \citet{johnson2013privacy} satisfies local minimax optimality for various query functions, and that it outperforms the Laplace and the smooth Laplace mechanisms under mild conditions.

The aforementioned contributions to the design of optimal $ (\varepsilon, \delta)$-DP mechanisms have in common that they limit their attention \emph{a priori} to specific classes of mechanisms (such as Gaussian; standard, truncated or discretized Laplace; staircase; uniform-Bernoulli product or uniform distributions) and subsequently prove either optimality among the mechanisms in their respective classes or asymptotic optimality among larger families of mechanisms as $\varepsilon \rightarrow 0$ and $\delta \rightarrow 0$ simultaneously. In this paper, we propose to formulate and solve the optimal $(\varepsilon, \delta)$-DP mechanism design problem as an infinite-dimensional distributionally robust optimization (DRO) problem \citep{DY10:distr_rob_opt, WKS14:DRCO, PME19:Watergate}. To this end, we minimize an expected loss function over all noise distributions, subject to the satisfaction of the DP constraints. In contrast to much of the existing literature, our formulation caters for generic loss functions (including asymmetric ones such as the pinball loss), and it can restrict the noise via support constraints. We show that our formulation affords a strong dual, and we develop hierarchies of finite-dimensional conservative approximations to the primal and dual formulations to derive converging upper and lower bounds on the optimal expected loss. Our upper bounds correspond to implementable perturbations whose optimality gaps can be certified by the lower bounds. Our bounding problems can be solved efficiently via cutting plane techniques that leverage the inherent problem structure. Our numerical experiments show that our optimal mechanisms can outperform the previously best results from the literature on artificial as well as two standard machine learning benchmark problems. 
{\color{black}

\begin{examplectd}[cont'd]
To guarantee $(1, 0.2)$-DP, our data independent additive noise algorithm from Section~\ref{sec:additive} adds a noise with standard deviation $257.68$ INR. As Figure~\ref{fig:distributions} demonstrates, our noise distribution does not appear to admit a simple analytical characterization.
\begin{figure}[!t]
    \centering
    \resizebox{1\columnwidth}{!}{
   \includegraphics[trim={0 4.5cm 0 3.75cm},clip,scale = 0.3]{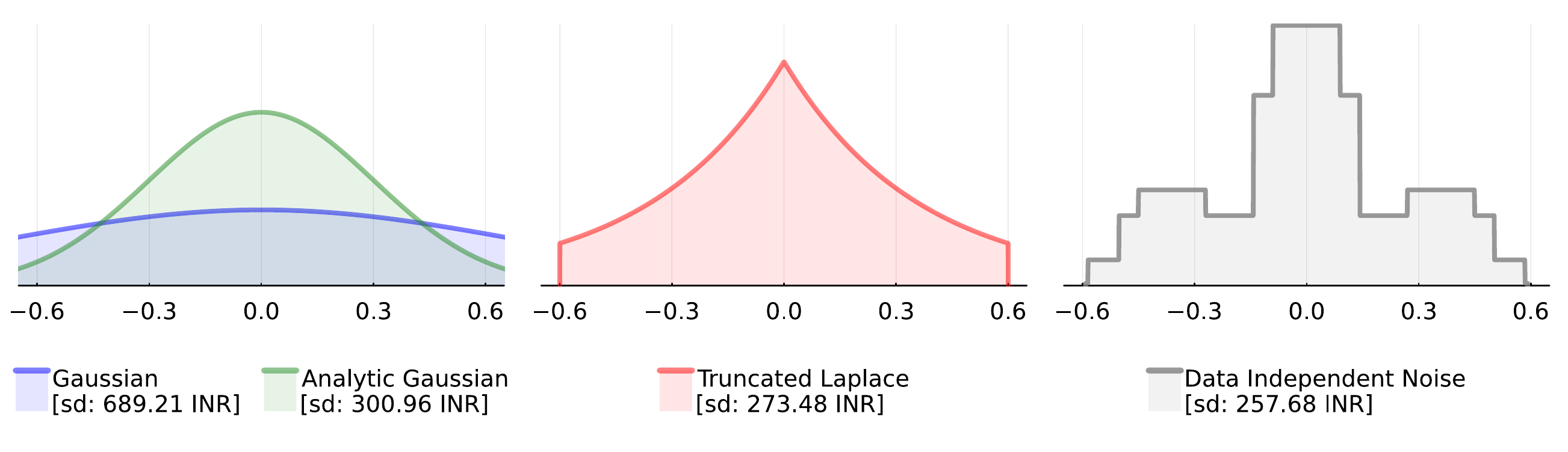}
    }
    \caption{\color{black}\textit{Different noise distributions that guarantee $(1,0.2)$-DP in Example~\ref{example:recurred}.}}
    \label{fig:distributions}
\end{figure}
\end{examplectd}
}

The contributions of this paper may be summarized as follows.
\begin{enumerate}
    \item[\emph{(i)}] We formulate the data independent additive noise problem as a DRO problem. Our formulation is flexible enough to cater for a large range of loss functions, and it extends to various problem variants such as the data dependent and the instance optimal problem.
    \item[\emph{(ii)}] We show that our primal and dual formulations can be bounded from above and below by converging hierarchies of large-scale linear programs that can be solved efficiently via tailored cutting plane techniques.
    \item[\emph{(iii)}] In contrast to the existing optimality results, which are either restricted to specific mechanisms or hold asymptotically, our formulation affords optimality guarantees that are non-asymptotic and that apply to any choice of $\varepsilon$ and $\delta$. Our numerical results showcase the advantages of our approach on a range of artificial as well as benchmark problems.
\end{enumerate}

In our view, the optimization perspective on DP put forward in this paper opens up several opportunities for future research. On one hand, our proposed hierarchy of primal and dual bounds appears to extend to other existing and new definitions of DP, it may allow for the development of approximation schemes for $(\varepsilon, \delta)$-DP with either \emph{a priori} or \emph{a posteriori} optimality guarantees, and it may generalize to multi-dimensional queries. On the other hand, the DP mechanism design problem gives rise to novel classes of DRO problems that have not been studied previously and that may find applications elsewhere. Most importantly, we believe that the design of optimal DP mechanisms should be viewed and addressed as a DRO problem, and the DRO community has developed a rich arsenal of techniques that can be leveraged beneficially to contribute with novel insights and algorithms.

The remainder of the paper unfolds as follows. Section~\ref{sec:additive} formulates the data independent additive noise problem as an infinite-dimensional DRO problem, it shows that the problem affords a strong dual, and it develops a hierarchy of converging finite-dimensional upper and lower bounding problems. {\color{black} Building upon this analysis, Section~\ref{sec:nonlinear} studies the data dependent additive noise problem, which optimizes over an uncountable family of noise distributions.} Section~\ref{sec:algo} develops a cutting plane algorithm to solve the bounding problems of Sections~\ref{sec:additive} and~\ref{sec:nonlinear}. We report numerical experiments in Section~\ref{sec:numerical} and offer concluding remarks in Section~\ref{sec:conclusions}, respectively. {\color{black} All proofs as well as some additional numerical results are relegated to the e-companion. Finally, the GitHub repository accompanying this paper contains the sourcecodes of all algorithms that we implemented as part of this work, the datasets used in our numerical experiments as well as additional numerical experiments and extended discussions of the related literature.\footnote{\url{https://github.com/selvi-aras/DP-via-DRO}}}

\textbf{Notation.}
Bold lower case letters denote vectors, while scalars are assigned standard lower case letters. The sets $\{1,\ldots, N\}$ and $\{-N, \ldots,N\}$ are abbreviated by $[N]$ and $[\pm N]$, respectively. For a set $\mathcal{S}$ and a scalar $a \in \mathbb{R}$, we let $\mathcal{S} + a  = \{s+a \ : \ s \in \mathcal{S}\}$ denote the Minkowski sum of the sets $\mathcal{S}$ and $\{a\}$; similarly, $\mathcal{S} - a$ abbreviates $\mathcal{S} + (-a)$. For a measurable interval $I \subseteq \mathbb{R}^n$, $|I| \in \mathbb{R} \cup \{+\infty \}$ denotes the Lebesgue measure of $I$. Unless otherwise stated, measures of real sets are defined on the corresponding Borel $\sigma$-algebras. The sets of sign-unrestricted and non-negative measures that are defined on the power-set $\sigma$-algebra of some set $A$ are denoted by $\mathcal{M}(A)$ and $\mathcal{M}_{+}(A)$, respectively. Finally, $\indicate{\mathcal{E}}$ is the indicator function taking value $1$ if the condition $\mathcal{E}$ is satisfied and $0$ otherwise.

\section{Data Independent Noise Optimization}\label{sec:additive}

We start with data independent additive noise mechanisms that perturb the query output $f(D)$ of a database $D \in {\color{black}\mathcal{D}}$ by adding a random noise $\tilde{X}$ independent of $D$ so as to minimize an expected loss while satisfying $(\varepsilon, \delta)$-DP. We formalize this problem in Section~\ref{sec:additive:formulation}, and Section~\ref{sec:additive:duality} derives a converging hierarchy of finite-dimensional upper and lower bounding problems.

\subsection{The Data Independent Optimization Problem}\label{sec:additive:formulation}

We study the problem
\begin{align}\label{problem:integral_main:preliminary_formulation}
    \begin{array}{cl}
    \underset{\gamma}{\text{minimize}}  & \displaystyle \int_{x \in \mathbb{R}} c(x) \diff \gamma(x) \\
     \text{subject to}  &   \gamma \in \mathcal{P}_0  \\
    &  \displaystyle \displaystyle \int_{x \in \mathbb{R}}\indicate{f(D) + x \in A} \diff \gamma(x)  \leq e^\varepsilon \cdot \displaystyle \int_{x \in \mathbb{R}} \indicate{f(D') + x \in A} \diff \gamma(x) + \delta \\
    & \pushright{\displaystyle \forall (D, D') \in \mathcal{N}, \ \forall A \in \mathcal{F},}
    \end{array}
\end{align}
where $(\mathbb{R}, \mathcal{F})$ is a measurable space with the Borel $\sigma$-algebra $\mathcal{F}$ on $\mathbb{R}$ and the set $\mathcal{P}_0$ of probability measures supported on $\mathbb{R}$. Problem~\eqref{problem:integral_main:preliminary_formulation} selects a probability measure $\gamma$ governing the random noise $\Tilde{X}$ so as to minimize the expected value of the Borel loss function $c: \mathbb{R}\mapsto \mathbb{R}_+$, subject to $(\varepsilon, \delta)$-DP for $\varepsilon, \delta > 0$. Note in particular that the integrals on both sides of the DP constraint evaluate the probabilities $\mathbb{P}[\mathcal{A}(D) \in A]$ and $\mathbb{P}[\mathcal{A}(D') \in A]$ for the randomized query outputs $\mathcal{A}(D) = f(D) + \tilde{X}$ and $\mathcal{A}(D') = f(D') + \tilde{X}$ with $\tilde{X} \sim \gamma$ as per our definition of $(\varepsilon, \delta)$-DP from the previous section. We assume that $c$ satisfies the following regularity conditions.

\begin{assumptions}[Loss Function]\label{assumptions_c}
The loss function $c : \mathbb{R} \mapsto \mathbb{R}_+$ satisfies the following properties: 
\begin{enumerate}[(a)]
    \item{\emph{Continuity.}} $c$ is continuous on $\mathbb{R}$.
    \item{\emph{Unboundedness.}} For any $r \in \mathbb{R}$ we have $c(x) \geq r$ for $|x|$ sufficiently large.
\end{enumerate}
\end{assumptions}

Assumption~\emph{(a)} enables us to construct discrete approximations to problem~\eqref{problem:integral_main:preliminary_formulation} and its dual that converge as we refine their granularity. 
{\color{black} Assumption~\emph{(b)} allows us to restrict these approximations to a bounded support of the involved measures without incurring an unbounded loss.} Loss functions typically used in the literature, such as the noise amplitude ($\ell_1$-loss with $c(x) = |x|$) and the noise power ($\ell_2$-loss with $c(x) = x^2$), satisfy Assumption~\ref{assumptions_c}.

Recall that $\Delta f:= \sup \{f(D') - f(D) \ : \ (D, D') \in \mathcal{N} \}$ is the global sensitivity of the query $f$ over the set of neighboring databases $\mathcal{N}$. We assume that $f$ is surjective in the following sense.

\begin{assumptions}[Query Function]\label{assumption_varphi}
For each $\varphi \in [-\Delta f, \Delta f]$, we have $f(D') - f(D) = \varphi$ for some $(D, D') \in \mathcal{N}$. 
\end{assumptions}

Assumption~\ref{assumption_varphi} is standard \citep{geng2014optimal,geng2020tight}, and it is satisfied by common descriptive statistics including the mean, median, minimum/maximum and standard deviation, as well as several popular machine learning algorithms (\textit{cf.}~Section~\ref{sec:numerical}), if the data is numeric. Our theory continues to apply if Assumption~\ref{assumption_varphi} is violated, but our reformulation of problem~\eqref{problem:integral_main:preliminary_formulation} will be conservative as it guarantees DP over the entire range $\varphi \in [-\Delta f, \Delta f]$, as opposed to the subset of query output differences that can actually be observed over $\mathcal{N}$.

{\color{black}
\begin{observation}\label{observation_primal}
Under Assumption~\ref{assumption_varphi}, the \emph{data independent noise optimization problem} is
\begin{align}\label{problem:integral_main}\tag{$\mathrm{P}$}
    \begin{array}{cl}
    \underset{\gamma}{\text{\emph{minimize}}}  & \displaystyle \int_{x \in \mathbb{R}} c(x) \diff \gamma(x) \\
     \text{\emph{subject to}}  &   \gamma \in \mathcal{P}_0  \\
    &  \displaystyle \int_{x \in \mathbb{R}}\indicate{x \in A} \diff \gamma(x)  \leq e^\varepsilon \cdot  \int_{x \in \mathbb{R}} \indicate{x + \varphi \in A} \diff \gamma(x) + \delta \quad \forall (\varphi, A) \in \mathcal{E},
    \end{array}
\end{align}
where $\mathcal{E} := [-\Delta f, \Delta f] \times \mathcal{F}$.
\end{observation}
}
Problem~\ref{problem:integral_main} has uncountably many decision variables and constraints and thus appears to be challenging to solve. {\color{black} Despite the linearity of the problem, standard tools from finite-dimensional linear programming, such as strong duality, are typically not available without further assumptions \citep{anderson1987linear}. Problem~\ref{problem:integral_main}} is feasible since its constraints are satisfied, for example, by Laplace \citep{dwork2006calibrating} and Gaussian measures \citep[Theorem A.1]{dwork2014algorithmic}. \mbox{Convexity of the feasible region implies that mixtures of such measures are also feasible.}

Problem~\ref{problem:integral_main} can be interpreted as an uncertainty quantification problem from the distributionally robust optimization literature \citep{OSSMO13:OUQ, HTTOM15:COUQ, HRKW15:DRUQCC}. Under this view, {\color{black} the constraints of~\ref{problem:integral_main} correspond to an uncountable number of moment conditions that define an ambiguity set from which nature selects a distribution $\gamma$ that minimizes the expected profit of the decision maker's action.} Problem~\ref{problem:integral_main} differs from the uncertainty quantification problems typically studied in the literature in both the number and the structure of these moment constraints. {\color{black} The constraints of~\ref{problem:integral_main} can also be interpreted as robust constraints that have to be satisfied for all realizations $(\varphi, A) \in \mathcal{E}$ of the `uncertain parameters' $\varphi$ and $A$ \citep{ben2009robust,BdH22:rob_opt}. In contrast to the standard robust optimization literature, however, the uncertain parameter $A$ in our problem is infinite-dimensional.}

Problem~\ref{problem:integral_main} is also reminiscent of continuous linear programs \citep{anderson1987linear}, which comprise uncountably many decision variables and constraints as well. Owing to their continuous-time control heritage, however, the constraints in continuous linear programs are indexed by a single bounded real scalar $x \in [0, T]$, whereas our constraint indices additionally involve the set of all Borel sets $\mathcal{F}$. Moreover, the decision variable of a continuous linear program has a bounded support $[0, T]$ and is assumed to admit a density, whereas our decision variable $\gamma$ has an unbounded support $\mathbb{R}$ and may not admit a density. Both of these additional complications imply that we cannot directly use the theory of continuous linear programming and instead have to derive bounding problems and prove their convergence from first principles.

\subsection{A Hierarchy of Converging Bounding Problems}\label{sec:additive:duality}

To obtain a tractable upper bound on problem~\ref{problem:integral_main}, we first introduce a restriction of~\ref{problem:integral_main} that replaces the generic measure $\gamma$ with the piecewise constant function
\begin{align}\label{restriction}
    \gamma(A) = \sum_{i \in \mathbb{Z}} p(i) \cdot \frac{|A\cap I_i(\beta)|}{\beta} \quad \forall A \in \mathcal{F},
\end{align}
where $p : \mathbb{Z} \mapsto \mathbb{R}_+$ satisfies $\sum_{i \in \mathbb{Z}} p(i) = 1$, and $\{ I_i(\beta) \}_{i \in \mathbb{Z}}$ partitions $\mathbb{R}$ into disjoint intervals $I_i(\beta) := [i\cdot \beta, (i+1)\cdot \beta)$, $i \in \mathbb{Z}$, of some pre-selected length $\beta > 0$. To simplify the exposition, we assume that $\Delta f$ is divisible by $\beta$. Under restriction~\eqref{restriction}, most constraints in~\ref{problem:integral_main} become redundant.

\begin{lemma}\label{prop:finite}
Under restriction~\eqref{restriction}, \ref{problem:integral_main} has the same optimal value as
\begin{align}\label{problem:restricted_main_finite}\tag{$\mathrm{P}(\beta)$}
    \begin{array}{cll}
        \underset{p}{\mathrm{minimize}} & \displaystyle \sum\limits_{i \in \mathbb{Z}} c_i(\beta) \cdot p(i) & \\
        \mathrm{subject\; to} & p : \mathbb{Z} \mapsto \mathbb{R}_+, \ \displaystyle   \sum_{i \in \mathbb{Z}} p(i) = 1\\
        &\displaystyle \sum\limits_{i \in \mathbb{Z}} \indicate{I_i(\beta) \subseteq A} \cdot p(i) \leq e^\varepsilon \cdot \sum\limits_{i \in \mathbb{Z}} \indicate{I_i(\beta) + \varphi \subseteq A} \cdot p(i) + \delta & \forall (\varphi, A) \in \mathcal{E}(\beta),
    \end{array}\raisetag{3\baselineskip}
\end{align}
where $c_i(\beta):= \beta^{-1} \cdot \int_{x \in I_i(\beta)} c(x) \mathrm{d}x$ and $\mathcal{E}(\beta) := \setvarphi(\beta) \times \mathcal{F}(\beta)$ with $\setvarphi(\beta) := \{-\Delta f, -\Delta f+ \beta, \ldots, \Delta f \}$ and $ \mathcal{F}(\beta) := \big\{ \bigcup_{i \in \mathcal{I}}I_i(\beta) \ : \ \mathcal{I} \subseteq \mathbb{Z} \big\}$.
\end{lemma}

In contrast to~\ref{problem:integral_main}, problem~\ref{problem:restricted_main_finite} has countably many decision variables. By Cantor's theorem, however, it still comprises uncountably many constraints since $\mathcal{F}(\beta)$ is indexed by the power set of infinitely many intervals $\{I_i(\beta)\}_{i\in\mathbb{Z}}$. To bound~\ref{problem:restricted_main_finite} from above by a finite-dimensional linear optimization problem, we constrain~\ref{problem:restricted_main_finite} further by restricting the discrete probability measure $p$ to a bounded support. Formally, we impose that there is $L\in \mathbb{N}$ such that
\begin{align}\label{eq:bound_sup}
    p(i) = 0 \quad \forall i \in \mathbb{Z}\setminus [\pm L],
\end{align}
that is, we bound the overall support to $2L + 1$ intervals $I_i(\beta)$ centered at $0$. Restriction~\eqref{eq:bound_sup} allows us to remove from \ref{problem:restricted_main_finite} any privacy constraint that relates to intervals $I_i(\beta)$ and $I_j(\beta)$ whose indices $i$ and $j$ both lie outside the support $[\pm L]$.

\begin{proposition}\label{prop:bounded}
With the additional constraint~\eqref{eq:bound_sup}, \ref{problem:restricted_main_finite} has the same optimal value as
\begin{align}\label{L-UB}\tag{$\mathrm{P}(L,\beta)$}
    \begin{array}{cll}
        \underset{p}{\mathrm{minimize}} & \displaystyle \sum\limits_{i \in [\pm L]} c_i(\beta) \cdot p(i) & \\
        \mathrm{subject\; to} & p : [\pm L] \mapsto \mathbb{R}_+, \ \displaystyle   \sum_{i \in [\pm L]} p(i) = 1\\
        &\displaystyle \sum\limits_{i \in [\pm L]} \indicate{I_i(\beta) \subseteq A} \cdot p(i) \leq e^\varepsilon \cdot \sum\limits_{i \in [\pm L]} \indicate{I_i(\beta) + \varphi \subseteq A} \cdot p(i) + \delta & \forall (\varphi, A) \in \mathcal{E}(L,\beta),
    \end{array}\raisetag{3\baselineskip}
\end{align}
where $\mathcal{E}(L, \beta):= \setvarphi(\beta) \times \mathcal{F}(L, \beta)$ with $\mathcal{F}(L, \beta) := \big\{ \bigcup_{i \in \mathcal{L}} I_i(\beta) \ : \ \mathcal{L} \subseteq [\pm L] \big\}$.
\end{proposition}

Problem~\ref{L-UB} constitutes a large-scale but finite-dimensional linear optimization problem that will serve as a building block to our cutting plane algorithm in Section~\ref{sec:algo}. {\color{black} We emphasize that the finiteness of the problem is solely due to our measure discretization~\eqref{restriction} as well as the restriction~\eqref{eq:bound_sup} to a bounded support, that is, we did not impose any additional assumptions on the structure of the worst-case events $A$ to arrive at a finite-dimensional model.} In terms of optimal values, we have the relationship $\mathrm{P}(L,\beta) \geq \mathrm{P}(\beta) \geq \mathrm{P}$ for all support sizes $L \in \mathbb{N}$ and interval lengths $\beta$. 

We next derive a lower bound on~\ref{problem:integral_main}. To this end, we employ a strategy widely adopted in distributionally robust optimization and first propose a dual to problem~\ref{problem:integral_main}:
\begin{align}\label{problem:integral_dual}\tag{$\mathrm{D}$}
    \begin{array}{cl}
     \underset{\theta, \psi}{\mathrm{maximize}} & \displaystyle \theta - \delta \int_{(\varphi, A) \in \mathcal{E}} \mathrm{d}\psi(\varphi, A) \\
    \textup{subject to} & \theta \in \mathbb{R}, \ \psi \in \mathcal{M}_{+}(\mathcal{E})\\
    & \displaystyle \theta \leq c(x) + \int_{(\varphi, A) \in \mathcal{E}} \indicate{x \in A}\mathrm{d}\psi(\varphi, A) - e^\varepsilon \cdot \int_{(\varphi, A) \in \mathcal{E}}\indicate{x+\varphi \in A} \mathrm{d}\psi(\varphi, A) \\
    & \pushright{\forall x \in \mathbb{R}.}
    \end{array}
\end{align}
The integrals in this problem are well-defined due to the domain of $\psi$ specified in the first constraint. Problem~\ref{problem:integral_dual} affords a natural interpretation: suppose that $\delta = 0$ and replace in the objective function the epigraphical variable $\theta$ with
\begin{equation*}
    \inf_{x \in \mathbb{R}} \; c(x) + \int_{(\varphi, A) \in \mathcal{E}} \indicate{x \in A}\mathrm{d}\psi(\varphi, A) - e^\varepsilon \cdot \int_{(\varphi, A) \in \mathcal{E}}\indicate{x+\varphi \in A} \mathrm{d}\psi(\varphi, A).
\end{equation*}
Problem~\ref{problem:integral_dual} then determines a conic combination of database-event pairs $(\varphi, A)$ that maximizes the sum of noise-related costs $c (x)$ and cumulative DP shortfall (\emph{i.e.}, the cumulative violation of all DP constraints) under the most benign realization $x$ of the random noise $\tilde{X}$.

We can readily establish weak duality between the problems~\ref{problem:integral_main} and~\ref{problem:integral_dual}.

\begin{proposition}[Weak Duality]\label{prop:weakdual}
For any $\gamma$ feasible in \ref{problem:integral_main} and $(\theta,\psi)$ feasible in \ref{problem:integral_dual}, we have
$$ \int_{x \in \mathbb{R}} c(x) \diff \gamma(x) \geq \displaystyle \theta - \delta \int_{(\varphi, A) \in \mathcal{E}} \mathrm{d}\psi(\varphi, A). $$
\end{proposition}

{\color{black} It is tempting to conclude that strong duality should hold as well between \ref{problem:integral_main} and \ref{problem:integral_dual} due the linearity of both problems. However, strong duality does not typically hold in infinite-dimensional optimization without further assumptions \citep{anderson1987linear}. We therefore defer the discussion of strong duality between \ref{problem:integral_main} and \ref{problem:integral_dual} to the end of this section.}

Similar to problem~\ref{problem:integral_main}, problem~\ref{problem:integral_dual} appears challenging to solve since it comprises uncountably many decision variables and constraints. To construct a tractable lower bound on~\ref{problem:integral_dual}, we first remove all variables $\psi(\varphi, A)$ indexed by $(\varphi, A) \in \mathcal{E} \setminus \mathcal{E}(\beta)$, that is, we impose that
\begin{align}\label{constraints:extra}
    \int_{(\varphi, A) \in \mathcal{E}\setminus\mathcal{E}(\beta)} \diff \psi(\varphi, A) = 0.
\end{align}
Restriction~\eqref{constraints:extra} can be understood as the dual pendant to our discretization~\eqref{restriction}; in fact, the dual variables unaffected by~\eqref{constraints:extra} correspond precisely to the constraints in problem~\ref{problem:restricted_main_finite}. Under restriction~\eqref{constraints:extra}, most constraints of~\ref{problem:integral_dual} become redundant.

\begin{lemma}\label{prop:lb-beta}
With the additional constraint~\eqref{constraints:extra}, \ref{problem:integral_dual} has the same optimal value as
    \begin{align}\label{LB-beta}\tag{$\mathrm{D}(\beta)$}
    \mspace{-30mu}\begin{array}{cl}
    \underset{\theta, \psi}{\mathrm{maximize}} & \theta - \displaystyle \delta  \cdot \int_{(\varphi, A) \in \mathcal{E}(\beta)} \mathrm{d}\psi(\varphi, A)\\
    \textup{subject to} & \theta \in \mathbb{R}, \ \psi \in \mathcal{M}_{+}(\mathcal{E}(\beta))\\
    & \displaystyle  \theta  \leq  \underline{c}_i(\beta) + \int_{(\varphi, A) \in \mathcal{E}(\beta)} \indicate{I_i(\beta) \subseteq A}\mathrm{d}\psi(\varphi, A) - e^\varepsilon \cdot \int_{(\varphi, A) \in \mathcal{E}(\beta)}\indicate{I_i(\beta) + \varphi \subseteq A} \mathrm{d}\psi(\varphi, A) \\
    & \pushright{\forall i \in \mathbb{Z},}
    \end{array}\raisetag{3.5\baselineskip}
\end{align}
where $\underline{c}_i(\beta) := \inf \{ c(x) : x \in I_i(\beta) \}, \ i \in \mathbb{Z}$.
\end{lemma}

In contrast to problem~\ref{problem:integral_dual}, which comprises uncountably many constraints, problem~\ref{LB-beta} has countably many constraints. However, the problem still contains infinitely many constraints as well as uncountably many variables. To bound~\ref{LB-beta} from below by
a finite-dimensional linear optimization problem, we set $\psi(\mathcal{E}(\beta) \setminus \mathcal{E}(L, \beta)) = 0$ for some $L \in \mathbb{N}$, that is, we impose that
\begin{align}\label{constraints:extratwo}
    \displaystyle \int_{(\varphi, A) \in \mathcal{E}(\beta) \setminus \mathcal{E}(L, \beta)} \diff \psi(\varphi, A) = 0.
\end{align}
Restriction~\eqref{constraints:extratwo} is the dual pendant to our support constraint~\eqref{eq:bound_sup}. It removes variables associated with events that contain intervals sufficiently far away from $0$ since $(\varphi, A) \in \mathcal{E}(\beta) \setminus \mathcal{E}(L ,\beta)$ implies that $A \not \subseteq \bigcup_{i \in [\pm L]} I_i(\beta)$.

\begin{proposition}\label{prop:lb-L-beta}
With the additional constraint~\eqref{constraints:extratwo}, \ref{LB-beta} has the same optimal value as
\begin{align}\label{lb-L-beta}\tag{$\mathrm{D}(L, \beta)$}
    \mspace{-25mu}\begin{array}{cl}
    \underset{\theta, \psi}{\mathrm{maximize}} & \displaystyle \theta - \delta \cdot \sum_{(\varphi, A) \in \mathcal{E}(L,\beta)} \psi(\varphi, A) \\
    \textup{subject to} & \theta \in \mathbb{R}, \ \psi: \mathcal{E}(L, \beta) \mapsto \mathbb{R}_{+} \\[1em]
    & \displaystyle \theta \leq \underline{c}_i(\beta) + \sum_{(\varphi, A) \in \mathcal{E}(L,\beta)} \indicate{I_i(\beta) \subseteq A} \cdot \psi(\varphi, A) - e^\varepsilon \cdot \sum_{(\varphi, A) \in \mathcal{E}(L,\beta)}\indicate{I_i(\beta) + \varphi \subseteq A} \cdot \psi(\varphi, A) \\
    & \pushright{\forall i \in [\pm (L + \Delta f / \beta )].}
    \end{array}\raisetag{3.8\baselineskip}
\end{align}
\end{proposition}

Similar to~\ref{L-UB}, problem~\ref{lb-L-beta} constitutes a large-scale but finite-dimensional linear optimization problem that will serve as a building block to our cutting plane algorithm in Section~\ref{sec:algo}. In terms of optimal values, we have the relationship $\mathrm{D}(L,\beta) \leq \mathrm{D}(\beta) \leq \mathrm{D}$ for all support sizes $L \in \mathbb{N}$ and interval lengths $\beta$. In particular, \ref{problem:integral_main} and~\ref{problem:integral_dual} are sandwiched by the finite-dimensional linear optimization problems~\ref{L-UB} and~\ref{lb-L-beta}.

We close this section with an analysis of the convergence of the finite-dimensional linear optimization problems~\ref{L-UB} and~\ref{lb-L-beta}. To this end, recall that by our earlier assumption, $\beta$ divides $\Delta f$, which allows us to equivalently represent $\beta$ as $\Delta f / k$ for some $k \in \mathbb{N}$.

\begin{theorem}\label{thm:strong_duality}
    For any $\xi > 0$, there is $\Lambda' \in \mathbb{N}$ and $k' \in \mathbb{N}$ such that
    \begin{equation*}
        \text{\hyperref[{L-UB}]{$\mathrm{P}(\Lambda \cdot k, \Delta f / k)$}} - \text{\hyperref[{lb-L-beta}]{$\mathrm{D}(\Lambda \cdot k, \Delta f / k)$}} \leq \xi
        \qquad \forall \Lambda \geq \Lambda', \ \forall k \geq k'.
    \end{equation*}
\end{theorem}

Intuitively, Theorem~\ref{thm:strong_duality} states that both the discretization granularity $\Delta f / k$ needs to shrink \emph{and} the support $[- \Lambda \cdot \Delta f, \Lambda \cdot (\Lambda + 1/k)\cdot \Delta f)$ of the noise distribution needs to grow for the primal and dual approximations to converge. In particular, keeping the support fixed (which amounts to fixing $\Lambda$ in Theorem~\ref{thm:strong_duality}) and merely increasing $k$ is \emph{not} sufficient for convergence as the dual approximation provides a lower bound for \emph{all} noise distributions (of potentially unbounded support), as opposed to only the noise distributions that share the same support as the primal approximation. {\color{black} We elaborate further on this in our GitHub supplement, where we also derive conditions on $\Delta$ and $k$ that ensure feasibility of the upper bound~\ref{L-UB}.} Note that Theorem~\ref{thm:strong_duality} also implies strong duality of the two infinite-dimensional problems~\ref{problem:integral_main} and~\ref{problem:integral_dual}.

{\color{black}
We close this section by showing that the conditions of Assumption~\ref{assumptions_c} are in a sense minimal requirements to guarantee the correctness of Theorem~\ref{thm:strong_duality}.

\begin{proposition}\label{prop:impossibility}
    There are loss functions that satisfy one condition of Assumption~\ref{assumptions_c} (but not both) and for which $\text{\hyperref[{L-UB}]{$\mathrm{P}(L, \beta)$}}$ and $\text{\hyperref[{lb-L-beta}]{$\mathrm{D}(L, \beta)$}}$ do not converge for any $L \in \mathbb{N}$ and $\beta > 0$.
\end{proposition}}

Figure~\ref{figure:steps} summarizes the key results of this section.
\begin{figure}[tb]
\begin{center}
\begin{tikzpicture}
\node (primal) at (0,0) [draw,rounded corners, align=center, minimum width=2cm,minimum height=1.5cm] {\footnotesize{\underline{Primal}}\\\ref{problem:integral_main}};
\node (disc) at (6.5,0) [draw,rounded corners, align=center, minimum width=3cm,minimum height=1.5cm] {\footnotesize{\underline{Discretized Primal}}\\\ref{problem:restricted_main_finite}};
\node (finite) at (13,0) [draw,rounded corners, align=center, minimum width=3cm,minimum height=1.5cm] {\footnotesize{\underline{Finite Primal}}\\\ref{L-UB}};
\draw [->, dashed] (primal.east) --  (disc.west) node[above, midway, align= center] {\footnotesize{\textit{{Restriction to a piecewise}}}\\\footnotesize{\textit{{constant measure \eqref{restriction}}}}} node[below, midway] {\textcolor{blue}{\textit{\footnotesize{Lemma~\ref{prop:finite}}}}} ;
\draw [->, dashed] (disc.east) -- (finite.west) node[above, align = center, midway] {\footnotesize{\textit{{Restriction to a}}} \\ \footnotesize{\textit{{bounded support \eqref{eq:bound_sup}}}}} node[below, align = center, midway] {\textcolor{blue}{\footnotesize{\textit{Proposition~\ref{prop:bounded}}}}};
\node (dual) at (0,-3) [draw,rounded corners, align=center, minimum width=2cm,minimum height=1.5cm] {\footnotesize{\underline{Dual}}\\\ref{problem:integral_dual}};
\node (ddisc) at (6.5,-3) [draw,rounded corners, align=center, minimum width=3cm,minimum height=1.5cm] {\footnotesize{\underline{Discretized Dual}}\\\ref{LB-beta}};
\node (dfinite) at (13,-3) [draw,rounded corners, align=center, minimum width=3cm,minimum height=1.5cm] {\footnotesize{\underline{Finite Dual}}\\\ref{lb-L-beta}};
\draw [<-, dashed] (primal.south) -- (dual.north) node[above, midway, align= center] {\footnotesize{\textit{Weak duality}}} node[below, midway, align= center] {\textcolor{blue}{\footnotesize{\textit{Proposition~\ref{prop:weakdual}}}}};
\draw [<-, dashed] (dual.east) -- (ddisc.west) node[below, midway, align=center] {\textcolor{blue}{\textit{\footnotesize{ Lemma~\ref{prop:lb-beta}}  }}} node[above, midway, align=center] {\textit{\footnotesize{Removal of variables}  }\\\textit{\footnotesize{whose indices do not match}} \\ \textit{\footnotesize{the constraints in~\ref{problem:restricted_main_finite}} }};
\draw [<-, dashed] (ddisc.east) -- (dfinite.west) node[below, midway, align=center] {\textcolor{blue}{\textit{\footnotesize{ Proposition~\ref{prop:lb-L-beta} }}}}  node[above, midway, align=center] {\textit{\footnotesize{Removal of variables}} \\ \textit{\footnotesize{with sufficiently}} \\ \textit{\footnotesize{large indices}}};
\draw [<->, thick, dashed, draw=red] (finite.south) -- (dfinite.north) node[midway, align= center]
{\footnotesize{\textit{{Convergence}}}\\ \color{blue} \footnotesize{\textit{{Theorem~\ref{thm:strong_duality}}}}};
\end{tikzpicture}
\end{center}
    \caption{\textit{Summary of the results in Section~\ref{sec:additive}. Directed arrows $x \dashrightarrow y$ indicate upper bound relationships $x \leq y$, whereas the double arrow confirms the convergence of optimal values as $L$ increases and $\beta$ decreases.}}
    \label{figure:steps}
\end{figure} 

\section{Data Dependent Noise Optimization}\label{sec:nonlinear}

We next study data dependent noise mechanisms whose additive perturbation $\tilde{X}(f(D))$ of the query output $f(D)$ of a database $D \in {\color{black}\mathcal{D}}$ may depend on $f(D)$. Section~\ref{sec:general_form} formalizes this problem, and Section~\ref{sec:general_duality} develops finite-dimensional upper and lower bounding problems.

\subsection{The Data Dependent Optimization Problem}\label{sec:general_form}

We study the problem
\begin{align}\label{generalized_problem}
    \begin{array}{cl}
    \underset{\gamma}{\text{minimize}}  & \displaystyle \int_{\phi \in \Phi} \objweight(\phi) \cdot \left[ \int_{x \in \mathbb{R}} c(x) \diff \gamma(x \mid \phi) \right] \diff \phi \\
    \text{subject to}    &    \gamma \in \Gamma \\
    & \displaystyle \int_{x \in \mathbb{R}}\indicate{f(D) + x \in A} \diff \gamma(x\mid f(D))  \leq e^\varepsilon \cdot \displaystyle \int_{x \in \mathbb{R}} \indicate{f(D') + x \in A} \diff \gamma(x\mid f(D') ) + \delta \\
    & \pushright{\forall (D, D') \in \mathcal{N}, \ \forall A \in \mathcal{F},}
    \end{array}\raisetag{3.5\baselineskip}
\end{align}
where $\Phi := \{ f(D) : D \in {\color{black}\mathcal{D}} \}$ is the set of possible query outputs and $(\mathbb{R}, \mathcal{F})$ is again a measurable space with the Borel $\sigma$-algebra $\mathcal{F}$ on $\mathbb{R}$. Problem~\eqref{generalized_problem} selects a family $\{ \gamma (\cdot \mid \phi) \}_{\phi \in \Phi}$ of conditional probability measures governing the random noise $\tilde{X}(\cdot)$ so as to minimize an iterated expectation of the Borel loss function $c : \mathbb{R} \mapsto \mathbb{R}_+$ satisfying Assumption~\ref{assumptions_c}, subject to satisfaction of $(\varepsilon, \delta)$-DP. The inner expectation in the objective function evaluates the expected loss for a specific query output $\phi \in \Phi$, whereas the outer expectation weighs different query outputs according to the continuous probability density function $\objweight: \Phi \mapsto \mathbb{R}_+$. The domain of $\gamma$ is now defined as
\begin{align*}
    \Gamma := \left\{ \gamma \ : \ \begin{bmatrix} 
    \gamma(\cdot \mid \phi) \in \mathcal{P}_0, & \phi \in \Phi \\
    \phi \mapsto \gamma(A \mid \phi) \text{ measurable}, & A \in \mathcal{F} \\
    \end{bmatrix}
    \right\},
\end{align*}
where $\mathcal{P}_0$ is again the set of probability measures supported on $\mathbb{R}$. The domain $\Gamma$ restricts $\gamma$ to the set of Markov kernels with continuous state $\phi \in \Phi$. The first condition ensures that the inner expectation in the objective function is well-defined, while the second condition ensures that this expectation is measurable in the outer expectation. 

To ensure privacy, the weighting $\objweight$ must not reveal anything about the true database (which is ensured, for example, by choosing the uniform distribution over $\Phi$) or it must itself be kept private. In practice, the application domain often guides the choice of $\objweight$; we discuss this further in Section~\ref{sec:numerical}.
Similar to Assumption~\ref{assumption_varphi}, we impose the following surjectivity requirement on $f$.

\begin{assumptions}[Query Function]\label{assumptions_f}
$\Phi$ is a bounded interval, and for each $D \in {\color{black}\mathcal{D}}$ and $\varphi \in [-\Delta f, \Delta f] \cap (\Phi - f(D))$, we have $f(D') - f(D) = \varphi$ for some $(D, D') \in \mathcal{N}$.
\end{assumptions}

Assumption~\ref{assumptions_f} reduces to Assumption~\ref{assumption_varphi} if we set $\Phi = \mathbb{R}$. To simplify the exposition, however, we assume that $\Phi$ is bounded; otherwise, additional steps would have to be taken to restrict the states $\phi \in \Phi$ of $\gamma$ to bounded intervals without incurring a potentially unbounded loss.

{\color{black}
\begin{observation}\label{observation_primal_nonlinear}
Under Assumption~\ref{assumptions_f}, the \emph{data dependent noise optimization problem} is
\begin{align}\label{problem:integral_nonlinear}\tag{$\mathrm{P'}$}
    \begin{array}{cl}
    \underset{\gamma}{\text{\emph{minimize}}}  & \displaystyle \int_{\phi \in \Phi} \objweight(\phi) \cdot \left[ \int_{x \in \mathbb{R}} c(x)  \diff \gamma(x \mid \phi) \right] \diff \phi \\
    \text{\emph{subject to}}    &    \gamma \in \Gamma \\
    & \displaystyle \int_{x \in \mathbb{R}}\indicate{x \in A} \diff \gamma(x\mid\phi) \leq e^\varepsilon \cdot \displaystyle \int_{x \in \mathbb{R}} \indicate{x + \varphi \in A} \diff \gamma(x\mid \phi + \varphi) + \delta \\
    & \pushright{\quad \forall \phi \in \Phi,\ \forall (\varphi, A) \in \mathcal{E}'(\phi),}
    \end{array}
\end{align}
where $\mathcal{E}'(\phi):= [[-\Delta f, \Delta f] \cap (\Phi-\phi)] \times \mathcal{F}$ for $\phi \in \Phi$.
\end{observation}
}

Problem~\ref{problem:integral_nonlinear} is not a generalization of problem~\ref{problem:integral_main} from Section~\ref{sec:additive} \textit{per se}, but \ref{problem:integral_main} would be recovered from \ref{problem:integral_nonlinear} if we introduced the additional requirement that all $\gamma(\cdot | \phi)$, $\phi \in \Phi$, in \ref{problem:integral_nonlinear} must coincide. This argument implies that problem \ref{problem:integral_nonlinear} is guaranteed to be feasible. Note that \ref{problem:integral_nonlinear} does not decompose into separate problems for $\phi \in \Phi$ since the DP constraint couples the conditional measures of neighbouring databases. Similar to~\ref{problem:integral_main}, problem~\ref{problem:integral_nonlinear} contains uncountably many decision variables and constraints and thus appears challenging to solve. {\color{black} The data dependent noise optimization problem~\ref{problem:integral_nonlinear} is more involved than its data independent counterpart~\ref{problem:integral_main}, however, since it optimizes over an uncountable family of noise distributions. In the following, we will re-use the insights of Section~\ref{sec:additive} to reduce the variables and constraints relating to each individual noise distribution, which allows us to focus on the new challenge of uncountably many noise distributions that is unique to the data dependent setting.}

\subsection{A Hierarchy of Converging Bounding Problems}\label{sec:general_duality}

To obtain a tractable upper bound on the data \emph{independent} noise optimization problem, Section~\ref{sec:additive:duality} bounds problem~\ref{problem:integral_main} from above by finite-dimensional approximations that live on a partition $\{ I_i(\beta) \}_{i \in \mathbb{Z}}$ of the possible noise realizations into disjoint intervals $I_i(\beta) = [i\cdot \beta, (i+1)\cdot \beta)$ of length $\beta > 0$. In this section, we retain our earlier assumption that $\Delta f$ is divisible by $\beta$, and we additionally stipulate that $\Phi = \bigcup_{k \in [K]} \Phi_k(\beta)$ with $\Phi_k(\beta) := I_{t+k}(\beta)$ for some $t \in \mathbb{Z}$ and $K \in \mathbb{N}$. This will allow us to partition the set of possible query outputs $\Phi$ in the same way, using a single granularity parameter $\beta$.

To bound problem~\ref{problem:integral_nonlinear} from above, we first restrict the uncountable family $\{ \gamma (\cdot | \phi) \}_{\phi \in \Phi}$ of probability measures in~\ref{problem:integral_nonlinear} to a finite subset that is piecewise constant on the intervals $\Phi_k(\beta)$:
\begin{subequations}\label{additional_general_primal}
\begin{equation}\label{additional_measure}
\gamma(\cdot  \mid  \phi) = \gamma(\cdot \mid \phi') \quad \forall k \in [K], \ \forall \phi, \phi' \in \Phi_k(\beta).
\end{equation}
Under restriction~\eqref{additional_measure}, 
\ref{problem:integral_nonlinear} 
optimizes over finitely many probability measures $\gamma_k$, $k \in [K]$, but it still involves uncountably many decision variables and constraints. To further simplify the problem, we restrict each probability measure in~\ref{problem:integral_nonlinear} to a piecewise constant function via
\begin{equation}\label{restriction_general}
    \gamma_k(A) = \sum_{i \in \mathbb{Z}} p_{k}(i) \cdot \dfrac{|A \cap I_i(\beta)|}{\beta} \quad \forall k \in [K], \ \forall A \in \mathcal{F}
\end{equation}
for a family of probability measures $\{p_k: \mathbb{Z} \mapsto \mathbb{R}_+ \}_{k \in [K]}$ satisfying $\sum_{i \in \mathbb{Z}} p_{k}(i) = 1$ for all $k \in [K]$. We also restrict each probability measure $\gamma_k$ to a bounded support by imposing that 
\begin{equation}\label{eq:general_bound_sup}
    p_{k}(i) = 0 \quad \forall k \in [K],\ \forall i \in \mathbb{Z} \setminus [\pm L]
\end{equation}
for some $L \in \mathbb{N}$. The restrictions~\eqref{restriction_general} and~\eqref{eq:general_bound_sup} are akin to~\eqref{restriction} and~\eqref{eq:bound_sup} from Section~\ref{sec:additive}, respectively.
\end{subequations} 

\begin{proposition}\label{prop:general_bounded}
With the additional constraints~\eqref{additional_general_primal}, \ref{problem:integral_nonlinear} has the same optimal value as
\begin{align}\label{general_L-UB}\tag{$\mathrm{P'}(L,\beta)$}
    \begin{array}{cl}
        \underset{ p }{\mathrm{minimize}} & \displaystyle \beta\cdot  \sum\limits_{k \in [K]} \objweight_k(\beta) \cdot \Big[ \sum\limits_{i \in [\pm L]} c_i(\beta) \cdot p_{k}(i) \Big] \\
         \mathrm{subject\; to} &\displaystyle  p_{k}: [\pm L] \mapsto \mathbb{R}_{+}, \ \sum\limits_{i \in [\pm L]} p_{k}(i) = 1, \ k \in [K] \\
        &\displaystyle \sum\limits_{i \in [\pm L]} \indicate{I_i(\beta) \subseteq A} \cdot p_{k}(i) \leq e^\varepsilon \cdot \sum\limits_{i \in [\pm L]} \indicate{I_i(\beta) + \varphi \subseteq A} \cdot p_{m}(i) + \delta \\
        & \pushright{\forall k,m \in [K], \ \forall (\varphi, A) \in \mathcal{E}'_{km}(L,\beta)}
    \end{array}
\end{align}
where $\objweight_k(\beta) := \beta^{-1} \cdot \int_{\phi \in\Phi_k(\beta)} \objweight(\phi) \diff \phi$ and $\mathcal{E}'_{km}(L, \beta) := [\setvarphi(\beta) \cap \{(m-k-1)\cdot \beta,\ (m-k) \cdot \beta,\ (m-k+1)\cdot \beta\}]\times \mathcal{F}(L, \beta)$ with $\setvarphi(\beta)$ and $\mathcal{F}(L, \beta)$ defined as in Section~\ref{sec:additive}. 
\end{proposition}

Problem~\ref{general_L-UB} constitutes a large-scale but finite-dimensional linear optimization problem that will serve as a building block to our cutting plane algorithm in Section~\ref{sec:algo}. In terms of optimal values, we have the relationship $\mathrm{P'}(L,\beta) \geq \mathrm{P'}$ for all $L \in \mathbb{N}$ and $\beta > 0$. 

To obtain a lower bound on~\ref{problem:integral_nonlinear},
we first propose the dual problem
\begin{align}\label{problem:dual_integral_nonlinear}\tag{$\mathrm{D'}$}
\mspace{-30mu}
    \begin{array}{cl}
    \underset{\theta, \psi}{\text{maximize}}  & \displaystyle \int_{\phi \in \Phi} \left[\theta(\phi) -  \delta \cdot \int_{(\varphi, A) \in \mathcal{E}'(\phi)}\diff \psi (\varphi, A \mid \phi) \right] \diff \phi \\
    \text{subject to}    & \theta : \Phi \mapsto \mathbb{R} \text{ measurable},\ \psi \in \Psi \\
    & \displaystyle \theta(\phi) \leq \int_{(\varphi, A)\in \mathcal{E}'(\phi)} \indicate{x \in A} \diff \psi( \varphi, A \mid \phi) - e^\varepsilon \cdot \int_{(-\varphi, A)\in \mathcal{E}'(\phi)} \indicate{x+\varphi \in A}\diff \psi(\varphi, A \mid \phi - \varphi) \\
    & \pushright{ + c(x) \cdot \objweight(\phi)  \quad \forall \phi \in \Phi, \ \forall x \in \mathbb{R},}
    \end{array}\raisetag{3.5\baselineskip}
\end{align}
where the dual measure is defined over the set
\begin{align*}
\mspace{-5mu}
    \Psi := \left\{
        \psi \ : \ 
            \begin{bmatrix}
                \psi(\cdot \mid \phi) \in \mathcal{M}_{+}(\mathcal{E}'(\phi)), \ \phi \in \Phi \\
                \exists \psi_0 \in \mathcal{M}(\mathcal{E}) \, : \,
                \psi(\cdot \mid \phi) \ll \psi_0,\ \phi \in \Phi  \text{ with }
                (\varphi, A, \phi) \mapsto \dfrac{\mathrm{d}\psi (\varphi, A \mid \phi)}{\mathrm{d}\psi_0(\varphi, A)} \text{ measurable}
            \end{bmatrix}
    \right\}.
\end{align*}
Here, the first condition resembles the domain that we imposed on $\psi$ in problem~\ref{problem:integral_dual} from Section~\ref{sec:additive:duality}. The second condition is new, and it ensures that the integral on the right-hand side of the DP constraint in~\ref{problem:dual_integral_nonlinear}, which varies with $\varphi$, is well-defined.

We first establish weak duality between the problems~\ref{problem:integral_nonlinear} and~\ref{problem:dual_integral_nonlinear}.

\begin{proposition}[Weak Duality]\label{prop:weakdual_general}
\mbox{For any $\gamma$ feasible in \ref{problem:integral_nonlinear} and $(\theta, \psi)$ feasible in \ref{problem:dual_integral_nonlinear}, we have}
$$ \displaystyle \int_{\phi \in \Phi} \objweight(\phi) \cdot \left[ \int_{x \in \mathbb{R}} c(x)  \diff \gamma(x \mid \phi) \right] \diff \phi \geq \displaystyle \int_{\phi \in \Phi} \left[\theta(\phi) -  \delta \cdot \int_{(\varphi, A) \in \mathcal{E}'(\phi)}\diff \psi (\varphi, A \mid \phi) \right] \diff \phi.$$
\end{proposition}

{\color{black} As in the previous section, we defer the discussion of strong duality between \ref{problem:integral_nonlinear} and \ref{problem:dual_integral_nonlinear} to the end of this section.} To construct a tractable lower bound on~\ref{problem:dual_integral_nonlinear}, we impose that
\begin{subequations}\label{general_additional}
\begin{equation}\label{general_dual_23}
\psi(\cdot \mid \phi) = \psi(\cdot \mid \phi')
\;\; \text{and} \;\;
\theta(\phi) = \theta(\phi')
\qquad \forall k \in [K], \ \forall \phi, \phi' \in \Phi_k(\beta),
\end{equation}
that is, $\psi(\cdot | \phi)$ and $\theta(\phi)$ have to be piecewise constant over the intervals $\Phi_k(\beta)$, and we remove all variables $\psi(\varphi, A | \phi)$ indexed by $(\varphi, A) \in \mathcal{E}'(\phi) \setminus \mathcal{E}(L, \beta) $, that is, we impose that
\begin{equation}\label{general_dual_1}
\displaystyle \int_{(\varphi, A) \in \mathcal{E}'(\phi) \setminus \mathcal{E}(L, \beta)} \diff \psi(\varphi, A \mid \phi) = 0 \qquad \forall \phi \in \Phi,
\end{equation}
\end{subequations}
which implies that $\psi (\cdot | \phi)$ vanishes for all $\phi \in \Phi$ on $(\varphi, A)$ with $\varphi$ not divisible by $\beta$ or $A$ outside $\mathcal{F}(L,\beta)$. Constraint~\eqref{general_dual_1} is reminiscent of the constraints~\eqref{constraints:extra} and~\eqref{constraints:extratwo} from Section~\ref{sec:additive}, and it can be interpreted as the dual pendant of the restriction~\eqref{eq:general_bound_sup}.

\begin{proposition}\label{prop:generalized_finite_dual}
With the additional constraints~\eqref{general_additional}, \ref{problem:dual_integral_nonlinear} has the same optimal value as
\begin{align}\label{general_lb-L-beta}\tag{$\mathrm{D'}(L, \beta)$}
    \mspace{-25mu}
    \begin{array}{cl}
    \underset{\theta, \psi}{\mathrm{maximize}} &\displaystyle \beta \cdot \left[ \sum_{k \in [K]} \theta_k - \delta \cdot \sum_{k \in [K]} \sum_{(\varphi, A) \in \mathcal{E}'_{k}(L, \beta)} \psi_k(\varphi,A) \right] \\
    \mathrm{subject\; to}  & \bm{\theta} \in \mathbb{R}^{K}, \ \psi_k : \mathcal{E}'_{k}(L, \beta) \mapsto \mathbb{R}_{+}, \ k \in [K]\\
    & \displaystyle \theta_k \leq \sum_{(\varphi, A)\in \mathcal{E}'_{k}(L, \beta)} \indicate{I_i(\beta) \subseteq A} \cdot \psi_k(\varphi,A) - e^\varepsilon \cdot \sum_{(-\varphi, A) \in \mathcal{E}'_{k}(L, \beta)} \indicate{I_i(\beta) + \varphi \subseteq A} \cdot \psi_{k- \varphi/\beta}(\varphi,A) \\
    & \pushright{+ \underline{c}_i(\beta) \cdot \underline{\objweight}_k(\beta) \quad \forall k \in [K], \  \forall i \in [\pm (L + \Delta f / \beta)],}
    \end{array}\raisetag{3.5\baselineskip}
\end{align}
where $\mathcal{E}'_{k}(L,\beta) := [\setvarphi(\beta) \cap (\Phi - \underline{\Phi}_k(\beta))] \times \mathcal{F}(L, \beta)$ for $\underline{\Phi}_k(\beta):= \inf \{\phi : \phi \in\Phi_k(\beta)\}$, $\underline{c}_i(\beta) := \inf \{ c(x) : x \in I_i(\beta)\}$ and $\underline{\objweight}_k(\beta) := \inf \{ \objweight(\phi) : \phi \in\Phi_k(\beta)\}$, $k\in[K]$ and $i \in \mathbb{Z}$.
\end{proposition}

The large-scale but finite-dimensional linear optimization problem~\ref{general_lb-L-beta} will serve as a building block to our cutting plane algorithm in Section~\ref{sec:algo}. In terms of optimal values, we have the relationship $\mathrm{D'}(L,\beta) \leq \mathrm{D'}$ for all $L \in \mathbb{N}$ and $\beta > 0$. Similar to Section~\ref{sec:additive},~\ref{problem:integral_nonlinear} and~\ref{problem:dual_integral_nonlinear} are sandwiched by the finite-dimensional linear optimization problems~\ref{general_L-UB} and~\ref{general_lb-L-beta}.

We close this section with an analysis of the convergence of the finite-dimensional linear optimization problems~\ref{general_L-UB} and~\ref{general_lb-L-beta}. Similarly to Section~\ref{sec:additive}, recall that by our earlier assumption, $\beta$ divides $\Delta f$, which allows us to equivalently represent $\beta$ as $\Delta f / k$ for some $k \in \mathbb{N}$.

\begin{theorem}\label{thm:generalized_strong_dual}
For any $\xi > 0$, there is $\Lambda' \in \mathbb{N}$ and $k' \in \mathbb{N}$ such that
\begin{equation*}
    \text{\hyperref[{general_L-UB}]{$\mathrm{P}'(\Lambda \cdot k, \Delta f / k)$}} - \text{\hyperref[{general_lb-L-beta}]{$\mathrm{D}'(\Lambda \cdot k, \Delta f /k)$}}\leq \xi
        \qquad \forall \Lambda \geq \Lambda', \ \forall k \geq k'.
\end{equation*}
\end{theorem}
Intuitively, Theorem~\ref{thm:generalized_strong_dual} has a similar interpretation as Theorem~\ref{thm:strong_duality}, that is, both the discretization granularity $\Delta f / k$ needs to shrink \emph{and} the support $[- \Lambda \cdot \Delta f, (\Lambda + 1 / k) \cdot \Delta f)$ of the noise distributions needs to grow for the primal and dual approximations to converge. Here, we additionally observe that shrinking the granularity also results in a larger number of distributions to be optimized over. {\color{black} Similar to Theorem~\ref{thm:strong_duality} in the data independent setting, Theorem~\ref{thm:generalized_strong_dual} implies strong duality of the two infinite-dimensional problems \ref{problem:integral_nonlinear} and \ref{problem:dual_integral_nonlinear}.}

\section{Iterative Solution of the Bounding Problems}\label{sec:algo}

The bounding problems of Sections~\ref{sec:additive} and~\ref{sec:nonlinear} employ a uniform partitioning of the noise distribution $\gamma$. Motivated by our numerical experiments, which indicate that (near-)optimal noise distributions tend to combine steep peaks around $0$ with gradually declining tails, Section~\ref{sec:non-identical} extends our bounding problems to non-uniform partitions of $\gamma$. Non-uniform partitions allow us to compute noise distributions with similar expected losses in shorter computation times.

Unfortunately, even under a non-uniform partitioning the bounding problems cannot be solved monolithically with an off-the-shelf solver due to their exponential scaling in the problem parameters. Instead, Section~\ref{sec:cutting-plane} proposes a cutting plane technique that solves those bounding problems iteratively through an increasingly accurate  sequence of relaxations. At the heart of our cutting plane technique is the identification of the constraints that our incumbent solutions violate with the largest margins. While a na\"ive search for these constraints would require an exponential effort, our algorithm scales polynomially in the size of the problem description.


\subsection{Upper and Lower Bounding Problems with Non-Uniform Partitions}\label{sec:non-identical}

Recall that our bounding problems~\ref{L-UB} and~\ref{lb-L-beta} from Section~\ref{sec:additive} partition the support of the noise distribution $\gamma$ into $2L + 1$ uniform intervals $I_i(\beta) = [i \cdot \beta, (i + 1) \cdot \beta)$, $i \in [\pm L]$. Let $\bm{\pi} \in \{ -L, \ldots, L + 1 \}^{N+1}$ be an index vector satisfying
\begin{equation*}
    -L = \pi_1
    \; < \;
    \pi_2
    \; < \;
    \ldots
    \; < \;
    \pi_N
    \; < \;
    \pi_{N+1} = L + 1,
\end{equation*}
and let $\Pi_j(\beta) = [\pi_j \cdot \beta, \pi_{j+1}\cdot \beta)$, $j \in [N]$, denote the $j$-th interval induced by the consecutive elements of $\bm{\pi}\cdot \beta$. Consider a variant of our upper bounding problem~\ref{L-UB} that enforces equality of all decision variables $p(i)$ and $p(i')$, $i, i' \in [\pm L]$, that satisfy $\pi_j \leq i, i' < \pi_{j+1}$ for some $j \in [N]$. The revised upper bounding problem is equivalent to
\begin{align}\label{new_L-UB}\tag{$\mathrm{P}(\bm{\pi}, \beta)$}
\mspace{-30mu}
    \begin{array}{cll}
        \underset{p}{\mathrm{minimize}} & \displaystyle \sum\limits_{j \in [N]} c_j(\bm{\pi}, \beta) \cdot p(j)  \\
        \mathrm{subject\; to} &\displaystyle p : [N] \mapsto \mathbb{R}_+, \ \sum_{j \in [N]} p(j) = 1 \\
        &\displaystyle \sum\limits_{j \in [N]} p(j) \cdot \dfrac{|A \cap \Pi_{j}(\beta)|}{|\Pi_{j}(\beta)|} \leq e^\varepsilon \cdot \sum\limits_{j \in [N]} p(j)\cdot \dfrac{|(A - \varphi ) \cap \Pi_{j}(\beta)|}{|\Pi_{j}(\beta)|} + \delta  \quad \forall (\varphi, A) \in \mathcal{E}(L, \beta),
    \end{array}\raisetag{3.5\baselineskip}
\end{align}
where $c_j(\bm{\pi}, \beta) := |\Pi_{j}(\beta)|^{-1} \cdot \int_{x \in \Pi_{j}(\beta)} c(x) \diff x$. The new upper bound~\ref{new_L-UB} closely resembles the previous bound~\ref{L-UB}, except that the objective coefficients $c_j$ and the intervals $\Pi_{j}(\beta)$ in the DP constraints now reflect the new partitioning of $\gamma$.

In a similar fashion, we propose the revised lower bound
\begin{align}\label{new_LB}\tag{$\mathrm{D}(\bm{\pi}, \beta)$}
\mspace{-25mu}
    \begin{array}{cll}
        \underset{p}{\mathrm{minimize}} & \displaystyle \sum\limits_{j \in \mathfrak{N}} \underline{c}_j(\bm{\pi}, \beta) \cdot p(j)  \\
        \mathrm{subject\; to} &\displaystyle p : \mathfrak{N} \mapsto \mathbb{R}_+, \ \sum_{j \in  \mathfrak{N} } p(j) = 1 \\
        &\displaystyle \sum\limits_{j \in \mathfrak{N}} p(j) \cdot \dfrac{|A \cap \Pi_{j}(\beta)|}{|\Pi_{j}(\beta)|} \leq e^\varepsilon \cdot \sum\limits_{j \in \mathfrak{N}} p(j) \cdot \dfrac{|(A - \varphi ) \cap \Pi_{j}(\beta)|}{|\Pi_{j}(\beta)|} + \delta \quad \forall (\varphi, A) \in \mathcal{E}(L, \beta)
    \end{array}\raisetag{3.5\baselineskip}
\end{align}
where $\underline{c}_j(\bm{\pi}, \beta) := \inf \{c(x) : x \in \Pi_{j}(\beta) \}$ and the index set $\mathfrak{N} = \{ -\frac{\Delta f}{\beta} + 1, \ldots, N + \frac{\Delta f}{\beta} \}$ emerges from the previous index set $[N]$ by padding it at both ends with $\Delta f/ \beta$ additional elements whose associated interval indices are set to $\pi_{1 - t} = \pi_1 - t$ and $\pi_{N + 1 + t} = \pi_{N+1} + t$, $t = 1, \ldots, \Delta f / \beta$, with the intervals $\Pi_j (\beta)$ extended to $j \in \mathfrak{N} \setminus [N]$ in the obvious way. Again, the revised lower bound closely resembles the previous bound~\ref{lb-L-beta}, with minor changes in the objective coefficients $\underline{c}_j$ and the intervals $\Pi_{j}(\beta)$.
The revised bounding problems enjoy convergence properties akin to those from Section~\ref{sec:additive}.

\begin{corollary}\label{corr:nonuniform}
    For any $\beta > 0$ and $\bm{\pi}$ satisfying
    $-L = \pi_1 < \ldots < \pi_{N+1} = L+1$ for some $L \in \mathbb{N}$, we have \ref{new_L-UB} $\geq$ \ref{problem:integral_main} $=$   \ref{problem:integral_dual} $\geq$ \ref{new_LB}. Moreover, for any $\xi > 0$ there is $\Lambda' \in \mathbb{N}$ and $k' \in \mathbb{N}$ such that $\text{\ref{new_L-UB}} - \text{\ref{new_LB}} \leq \xi$ for any $\bm{\pi}$ whose induced partition $\{ \Pi_j (\beta) \}_{j \in [N]}$ is a refinement of a uniform partition $\{ I_i (\Delta f / k) \}_{i \in [\pm \Lambda \cdot k]}$ with $\Lambda \geq \Lambda'$ and $k \geq k'$.
\end{corollary}

Appendix~\ref{app:dependent_X} presents analogous bounding problems for the data dependent case.

\subsection{Cutting Plane Algorithm }\label{sec:cutting-plane}

Although the revised upper bounding problem \ref{new_L-UB} only contains $N$ decision variables, it remains challenging to solve monolithically since it comprises $\mathcal{O} (2^L \cdot \Delta f / \beta)$ DP constraints. To address this issue, Algorithm~\ref{alg:worst} solves a sequence of relaxations of \ref{new_L-UB} that only involve those constraints that are active at incumbent solutions. In particular, every iteration of Algorithm~\ref{alg:worst} determines a constraint $(\varphi^\star, A^\star)$ with maximum \textit{privacy shortfall}, which is the quantity that the DP constraints in \ref{new_L-UB} require to be non-positive:
\begin{align}\label{greedy_quant}
\displaystyle V(\varphi, A) &= \sum\limits_{j \in [N]} p(j) \cdot \dfrac{|A \cap \Pi_j(\beta)|}{|\Pi_j(\beta)|} - e^\varepsilon \cdot \sum\limits_{j \in [N]} p(j)\cdot \dfrac{|A \cap (\Pi_j(\beta) + \varphi)|}{|\Pi_j(\beta)|} - \delta
\quad \text{for } (\varphi, A) \in \mathcal{E} (L, \beta).
\end{align}

Since~\ref{new_L-UB} contains finitely many constraints, one readily recognizes that Algorithm~\ref{alg:worst} determines an optimal solution to~\ref{new_L-UB} in finitely many iterations.

\begin{algorithm}[tb]
\SetAlgoLined
\SetKwInOut{Input}{input}
\SetKwInOut{Output}{output}
\Input{$\bm{\pi}$, $\beta$, $\Delta f$} 
\Output{optimal solution $p^\star$ to problem~\ref{new_L-UB}}
Initialize $\mathcal{S} = \emptyset$\;
\Do{$\mathcal{S}$ has been updated}{
Let $p^\star$ be an optimal solution to the relaxation of~\ref{new_L-UB} that only contains the privacy constraints indexed by $(\varphi, A) \in \mathcal{S}    $. \\
Find a constraint $(\varphi^\star, A^\star)$ with maximum privacy shortfall under the incumbent solution $p^\star$.\\
\lIf{constraint $(\varphi^\star, A^\star)$ has positive privacy shortfall}
{update $\mathcal{S} = \mathcal{S} \cup (\varphi^\star, A^\star)$.\DontPrintSemicolon}
}
\Return $p^\star$.
\caption{\textit{Cutting plane algorithm for problem~\ref{new_L-UB}}}
\label{alg:worst}
\end{algorithm}

\begin{observation}\label{obs:cutting_plane_works}
    Algorithm~\ref{alg:worst} terminates after a finite number of iterations with an optimal solution $p^\star$ to problem~\ref{new_L-UB}.
\end{observation}

A key step in Algorithm~\ref{alg:worst} is the identification of a constraint $(\varphi^\star, A^\star) \in \mathcal{E} (L, \beta)$ with maximum privacy shortfall. A na\"ive implementation of this step would require the evaluation of $\mathcal{O} (2^L \cdot \Delta f / \beta)$ privacy shortfalls in time $\mathcal{O} (N)$ each. Instead, we employ Algorithm~\ref{alg:worst-event} to identify a constraint $(\varphi^\star, A^\star)$ with maximum privacy shortfall in polynomial time.

\begin{proposition}\label{alg:correct}
For a fixed solution $p$, Algorithm~\ref{alg:worst-event} can be implemented so as to return a constraint of \ref{new_L-UB} with maximum privacy shortfall in time $\mathcal{O}(N^3)$.
\end{proposition}

To illustrate the intuition behind Algorithm~\ref{alg:worst-event} and Proposition~\ref{alg:correct}, fix any $\varphi \in \setvarphi (\beta)$ in problem~\ref{new_L-UB}. To construct the DP constraint $(\varphi, A) \in \mathcal{E} (L, \beta)$ with maximum privacy shortfall across all $A \in \mathcal{F} (L, \beta)$, we need to decide for each interval $I_i (\beta) = [ i \cdot \beta, (i + 1) \cdot \beta)$, $i \in [\pm L]$, whether or not to include the interval $I_i (\beta)$ in $A$. To this end, we first observe that we can include all intervals $I_i (\beta)$ for which $I_i (\beta) \cap (\Pi_{j'} (\beta) + \varphi) = \emptyset$ for all $j' \in [N]$ since the inclusion of those intervals in $A$ cannot decrease $V (\varphi, A)$. For the intervals $I_i (\beta)$ that satisfy both $I_i (\beta) \subseteq \Pi_j (\beta)$ and $I_i (\beta) \subseteq (\Pi_{j'} (\beta) + \varphi)$ for some $j, j' \in [N]$, on the other hand, we compare the magnitude of the positive coefficient $p (j)$ with that of the negative coefficient $- e^\varepsilon \cdot p (j')$ in $V (\varphi, A)$ to decide whether $I_i (\beta)$ should be included in $A$. Algorithm~\ref{alg:worst-event} and Proposition~\ref{alg:correct} refine this idea by \emph{(i)} iterating only over the intervals $\Pi_j (\beta)$, $j \in [N]$, as opposed to the larger set of intervals $I_i (\beta)$, $i \in [\pm L]$; \emph{(ii)} identifying the relevant interval pairs $(j, j') \in [N]^2$ in linear time $\mathcal{O} (N)$ as opposed to quadratic time $\mathcal{O} (N^2)$; and \emph{(iii)} restricting the search over $\varphi \in \setvarphi (\beta)$ with cardinality $\mathcal{O} (\Delta f / \beta)$ to the smaller set in the outer for-loop with cardinality $\mathcal{O} (N^2)$.


\begin{algorithm}[tb]
\SetAlgoLined
\SetKwInOut{Input}{input}
\SetKwInOut{Output}{output}
\Input{$\bm{\pi}, \ \beta, \ p$, \ $\Delta f$}
\Output{constraint $(\varphi^\star, A^\star)$ with maximum privacy shortfall $V(\varphi^\star, A^\star)$}
Initialize $V^\star = 0$\;
\For {$\varphi \in \left\{ (\pi_{j} - \pi_{j'})\cdot \beta \ : \ (\pi_{j} - \pi_{j'})\cdot \beta \in [- \Delta f, \Delta f] \text{\emph{ and }} j,j' \in [N] \right\} \cup \{-\Delta f, \Delta f \}$}
{
Initialize $A = \emptyset$ and $V = 0$\;
\For{$j = 1,\ldots, N$}
{
Let $A_j = \Pi_j(\beta) \setminus [-L\cdot \beta + \varphi, (L+1)\cdot\beta + \varphi)$ and update
$$
A = A \cup A_j, \quad
    V = V + |A_j|\cdot  \dfrac{p(j)}{|\Pi_{j}(\beta)|}.
$$
    \For{$j' = 1,\ldots, N$}
    {
     \If{$p(j)/|\Pi_{j}(\beta)| > e^\varepsilon \cdot p(j')/|\Pi_{j'}(\beta)|$}{
     \vspace{1mm}
        Let $A_{jj'} = \Pi_j(\beta) \cap (\Pi_{j'}(\beta) + \varphi)$ and update
    $$ 
    A = A \cup A_{jj'}, \quad V = V + |A_{jj'}| \cdot \left[\dfrac{p(j)}{|\Pi_{j}(\beta)|} - e^\varepsilon \cdot \dfrac{p(j')}{|\Pi_{j'}(\beta)|}\right].
    $$
    }
    }
}
\If{$V > V^\star$}{
Update $\varphi^\star = \varphi$,  $A^\star = A$ and $V^\star = V$.
}
}

\Return $(\varphi^\star, A^\star)$ and $V^\star(\varphi, A) = V^\star - \delta$.
\caption{\textit{Identification of a constraint in~\ref{new_L-UB} with maximum privacy shortfall}}
\label{alg:worst-event}
\end{algorithm}

{\color{black} As discussed in Section~\ref{sec:additive:formulation}, the constraints of~\ref{new_L-UB} can be interpreted as semi-infinite robust optimization constraints. Under this lens, Algorithm~\ref{alg:worst} follows the tradition of cutting plane schemes from the robust optimization literature (\emph{cf.}~\citealt{bienstock2008computing}, \citealt{mutapcic2009cutting}, \citealt{bertsimas2016reformulation} and \citealt{patzold2020approximate}). The key technical contribution of this section is to identify for a fixed decision $p$ the maximally violated constraints $(\varphi^\star, A^\star) \in \mathcal{E}(L,\beta)$ without inspecting exponentially many events $A \in \mathcal{F} (L, \beta)$.}

This section focused on the cutting plane algorithm for the upper bounding problem \ref{new_L-UB} in the data independent setting. Algorithms~\ref{alg:worst} and~\ref{alg:worst-event} immediately extend to the lower bounding problem \ref{new_LB} if we replace the objective coefficients $c_j$ with $\underline{c}_j$ and extend the domain of the decisions $p$ from $[N]$ to $\mathfrak{N}$. Both algorithms also readily extend to the data dependent setting (\emph{cf.}~Section~\ref{sec:nonlinear}), where a constraint with maximum privacy shortfall can be identified in time $\mathcal{O}(K^2 \cdot N)$. For the sake of brevity, we relegate the details of that algorithm variant to the GitHub repository accompanying this paper.

\section{Numerical Experiments}\label{sec:numerical}

Our numerical experiments are split into two parts. The first part compares the privacy-accuracy trade-off of our optimization-based approach with popular DP mechanisms from the literature, and it examines the runtime of our cutting plane algorithm as well as the convergence of our bounding problems. Since this part focuses on the quality of our bounds as well as their computation times, we use synthetic instances that give us complete control over all parameters. The second part of our experiments investigates whether the improved accuracy on synthetic instances carries over to a better in-sample and out-of-sample performance in machine learning problems involving standard benchmark instances. To this end, we study differentially private variants of the na\"ive Bayes classifier and a proximal coordinate descent method for logistic regression.

Our optimization algorithms are implemented in C++ and use the GUROBI 9.5.2 LP solver. The machine learning algorithms from the second part are implemented in Julia, and all data is processed using Python 3. All experiments are conducted on Intel Xeon 2.66GHz cluster processors with 16GB memory in single-core and single-thread mode. All sourcecodes and datasets, together with more detailed descriptions of our numerical experiments, are available open-source on the GitHub repository accompanying this work.

\subsection{Synthetic Experiments}

\begin{table}[tb]
        \small
        \setlength{\arrayrulewidth}{0.2mm}
        \begin{center}
        \resizebox{1.0\columnwidth}{!}{%
    \hskip-0.8cm
\renewcommand{\arraystretch}{1.82} 
\begin{tabular}{c |  r | c  c c c c c c c c c|}
\multicolumn{2}{c}{} & \multicolumn{10}{c}{\Large $\delta$} 
 \\
\cline{3-12}
\multicolumn{2}{c|}{} & \textbf{0.005} & \textbf{0.010} &\textbf{0.020} & \textbf{0.050} & \textbf{0.100} & \textbf{0.200} & \textbf{0.250} &\textbf{0.300} &\textbf{0.500} &\textbf{0.750}\\
\cline{2-12}
\multirow{10}{*}{\Large $\varepsilon$} & {\textbf{0.005}} & {\cellcolor[gray]{0.976} 1.87\%}& {\cellcolor[gray]{0.988} 1.21\%}& {\cellcolor[gray]{0.992} 1.20\%}& {\cellcolor[gray]{0.932} 6.10\%}& {\cellcolor[gray]{0.876} 18.53\%}& {\cellcolor[gray]{0.944} 3.55\%}& {\cellcolor[gray]{0.628} 49.59\%}& {\cellcolor[gray]{0.792} 23.18\%}& {\cellcolor[gray]{0.616} 49.96\%}& {\cellcolor[gray]{0.74} 33.34\%}
\\ \hhline{~~|}
\cline{3-12}
& {\textbf{0.010}} & {\cellcolor[gray]{0.952} 2.89\%}& {\cellcolor[gray]{0.984} 1.42\%}& {\cellcolor[gray]{0.924} 8.00\%}& {\cellcolor[gray]{0.968} 2.32\%}& {\cellcolor[gray]{0.892} 17.08\%}& {\cellcolor[gray]{0.948} 3.08\%}& {\cellcolor[gray]{0.636} 49.18\%}& {\cellcolor[gray]{0.8} 23.03\%}& {\cellcolor[gray]{0.62} 49.92\%}& {\cellcolor[gray]{0.736} 33.35\%}
\\
\cline{3-12}
& {\textbf{0.020}} & {\cellcolor[gray]{0.956} 2.84\%}& {\cellcolor[gray]{0.96} 2.83\%}& {\cellcolor[gray]{0.996} 0.56\%}& {\cellcolor[gray]{0.9} 15.09\%}& {\cellcolor[gray]{0.904} 14.19\%}& {\cellcolor[gray]{0.6} 67.44\%}& {\cellcolor[gray]{0.644} 48.39\%}& {\cellcolor[gray]{0.804} 22.74\%}& {\cellcolor[gray]{0.624} 49.83\%}& {\cellcolor[gray]{0.728} 33.37\%}
\\
\cline{3-12}
& {\textbf{0.050}} & {\cellcolor[gray]{0.98} 1.85\%}& {\cellcolor[gray]{0.972} 2.11\%}& {\cellcolor[gray]{0.964} 2.57\%}& {\cellcolor[gray]{0.884} 18.36\%}& {\cellcolor[gray]{0.928} 6.22\%}& {\cellcolor[gray]{0.604} 66.03\%}& {\cellcolor[gray]{0.652} 46.15\%}& {\cellcolor[gray]{0.808} 21.95\%}& {\cellcolor[gray]{0.632} 49.58\%}& {\cellcolor[gray]{0.724} 33.42\%}
\\
\cline{3-12}
& {\textbf{0.100}} & {\cellcolor[gray]{0.936} 4.78\%}& {\cellcolor[gray]{0.94} 4.18\%}& {\cellcolor[gray]{0.92} 8.68\%}& {\cellcolor[gray]{0.828} 19.37\%}& {\cellcolor[gray]{0.744} 33.32\%}& {\cellcolor[gray]{0.608} 63.85\%}& {\cellcolor[gray]{0.66} 42.87\%}& {\cellcolor[gray]{0.82} 20.85\%}& {\cellcolor[gray]{0.64} 49.17\%}& {\cellcolor[gray]{0.716} 33.50\%}
\\
\cline{3-12}
&{\textbf{0.200}} & {\cellcolor[gray]{0.908} 10.43\%}& {\cellcolor[gray]{0.912} 9.60\%}& {\cellcolor[gray]{0.916} 8.90\%}& {\cellcolor[gray]{0.888} 17.29\%}& {\cellcolor[gray]{0.824} 19.82\%}& {\cellcolor[gray]{0.612} 60.17\%}& {\cellcolor[gray]{0.692} 37.70\%}& {\cellcolor[gray]{0.832} 19.36\%}& {\cellcolor[gray]{0.648} 48.33\%}& {\cellcolor[gray]{0.712} 33.63\%}
\\
\cline{3-12}
& {\textbf{0.500}} & {\cellcolor[gray]{0.796} 23.10\%}& {\cellcolor[gray]{0.816} 21.41\%}& {\cellcolor[gray]{0.788} 25.45\%}& {\cellcolor[gray]{0.896} 16.36\%}& {\cellcolor[gray]{0.688} 38.70\%}& {\cellcolor[gray]{0.664} 42.56\%}& {\cellcolor[gray]{0.764} 29.52\%}& {\cellcolor[gray]{0.88} 18.48\%}& {\cellcolor[gray]{0.656} 45.85\%}& {\cellcolor[gray]{0.708} 33.79\%}
\\
\cline{3-12}
& {\textbf{1.000}} & {\cellcolor[gray]{0.68} 40.11\%}& {\cellcolor[gray]{0.684} 39.73\%}& {\cellcolor[gray]{0.676} 40.23\%}& {\cellcolor[gray]{0.672} 40.41\%}& {\cellcolor[gray]{0.772} 28.62\%}& {\cellcolor[gray]{0.748} 32.53\%}& {\cellcolor[gray]{0.78} 26.63\%}& {\cellcolor[gray]{0.812} 21.43\%}& {\cellcolor[gray]{0.668} 41.80\%}& {\cellcolor[gray]{0.732} 33.36\%}
\\
\cline{3-12}
&{\textbf{2.000}} & {\cellcolor[gray]{0.704} 33.96\%}& {\cellcolor[gray]{0.7} 34.18\%}& {\cellcolor[gray]{0.72} 33.45\%}& {\cellcolor[gray]{0.756} 31.63\%}& {\cellcolor[gray]{0.752} 32.46\%}& {\cellcolor[gray]{0.768} 28.70\%}& {\cellcolor[gray]{0.776} 27.12\%}& {\cellcolor[gray]{0.784} 25.67\%}& {\cellcolor[gray]{0.696} 34.35\%}& {\cellcolor[gray]{0.76} 30.51\%}
\\
\cline{3-12}
& {\textbf{5.000}} & {\cellcolor[gray]{0.836} 19.32\%}& {\cellcolor[gray]{0.84} 19.31\%}& {\cellcolor[gray]{0.844} 19.29\%}& {\cellcolor[gray]{0.848} 19.23\%}& {\cellcolor[gray]{0.852} 19.15\%}& {\cellcolor[gray]{0.86} 19.01\%}& {\cellcolor[gray]{0.864} 18.94\%}& {\cellcolor[gray]{0.868} 18.88\%}& {\cellcolor[gray]{0.872} 18.65\%}& {\cellcolor[gray]{0.856} 19.06\%}
\\
\cline{2-12}
\end{tabular}}
\end{center}
~\\[-6mm]
\caption{\textit{Suboptimality of the best performing benchmark mechanisms on synthetic data independent instances with $\Delta f = 1$, $\ell_1$-loss and various combinations of $\varepsilon$ and $\delta$.}}
\label{tab:tradefoff_main}
\end{table}

In our first experiment, we compare the privacy-accuracy trade-off of our optimization-based DP scheme with that of popular benchmark mechanisms from the literature. To this end, we consider the data independent noise optimization problem and select 100 combinations of $\varepsilon \in [0.005, 5]$ and $\delta \in [0.005, 0.75]$. We intentionally chose conservative combinations of $\varepsilon$ and $\delta$ (\emph{cf.}~Table~1 of \citet{zhao2019reviewing}); larger values of $\varepsilon$ yield results that are more favorable to our algorithm. We solve our upper and lower bounding problems \ref{new_L-UB} and \ref{new_LB} with $\bm{\pi}$ and $\beta$ set appropriately so that their relative optimality gaps, measured as $100\% \cdot (\text{\ref{new_L-UB}} - \text{\ref{new_LB}})/\text{\ref{new_LB}}$, are strictly less than $1\%$ (with a median gap of $0.28\%$ across our experiments). We then use the midpoint $O := (\text{\ref{new_L-UB}} - \text{\ref{new_LB}}) / 2$ of both bounds as a substitute to the optimal privacy-accuracy trade-off, and we measure the suboptimality of different benchmark mechanisms from the literature: the analytic Gaussian \citep{balle2018improving} and the truncated Laplace \citep{geng2020tight} mechanisms as upper bounds and the `near optimal lower bound' of \citet[Thm 8]{geng2015optimal} and \cite{geng2020tight} as lower bound. Table~\ref{tab:tradefoff_main} records the optimality gaps of the best performing upper and lower bounds $B_{\mathrm{UB}}$ and $B_{\mathrm{LB}}$ from the literature, which consistently turn out to be the truncated Laplace mechanism and the lower bound of \cite{geng2020tight}. The optimality gaps are reported as $100\% \cdot [(B_{\mathrm{UB}} - B_{\mathrm{LB}}) / \max \{ O, 1 \}]$. A breakdown into separate suboptimalities incurred by the upper and lower bounds is presented in Appendix~\ref{app_numerical}. The table shows that the suboptimality of the benchmark approaches increases with $\varepsilon$ and $\delta$, and the optimality gaps are significant in most of the considered privacy regimes.\footnote{The vigilant reader will observe that the suboptimality is not entirely monotone in $\varepsilon$ and $\delta$. This is due to two factors: the computation of the `near optimal lower bound' involves a non-monotonic parameter rounding, and we approximate the optimal privacy-accuracy trade-off with the midpoint $O$ of our upper and lower bounds.} We note that if we replace $\max \{ O, 1 \}$ with $O$ in the denominator of the optimality gap formula, then the gaps in Table~\ref{tab:tradefoff_main} increase to more than $700\%$ for $(\varepsilon, \delta) = (5, 0.75)$. {\color{black} We observe qualitatively similar results as in Table~\ref{tab:tradefoff_main} also for the $\ell_2$-loss; the corresponding table is relegated to the e-companion.}

\begin{figure}[tb]
    \centering
    \resizebox{1.12\columnwidth}{!}{
    \hskip -1.5cm
    \includegraphics[scale = 0.23]{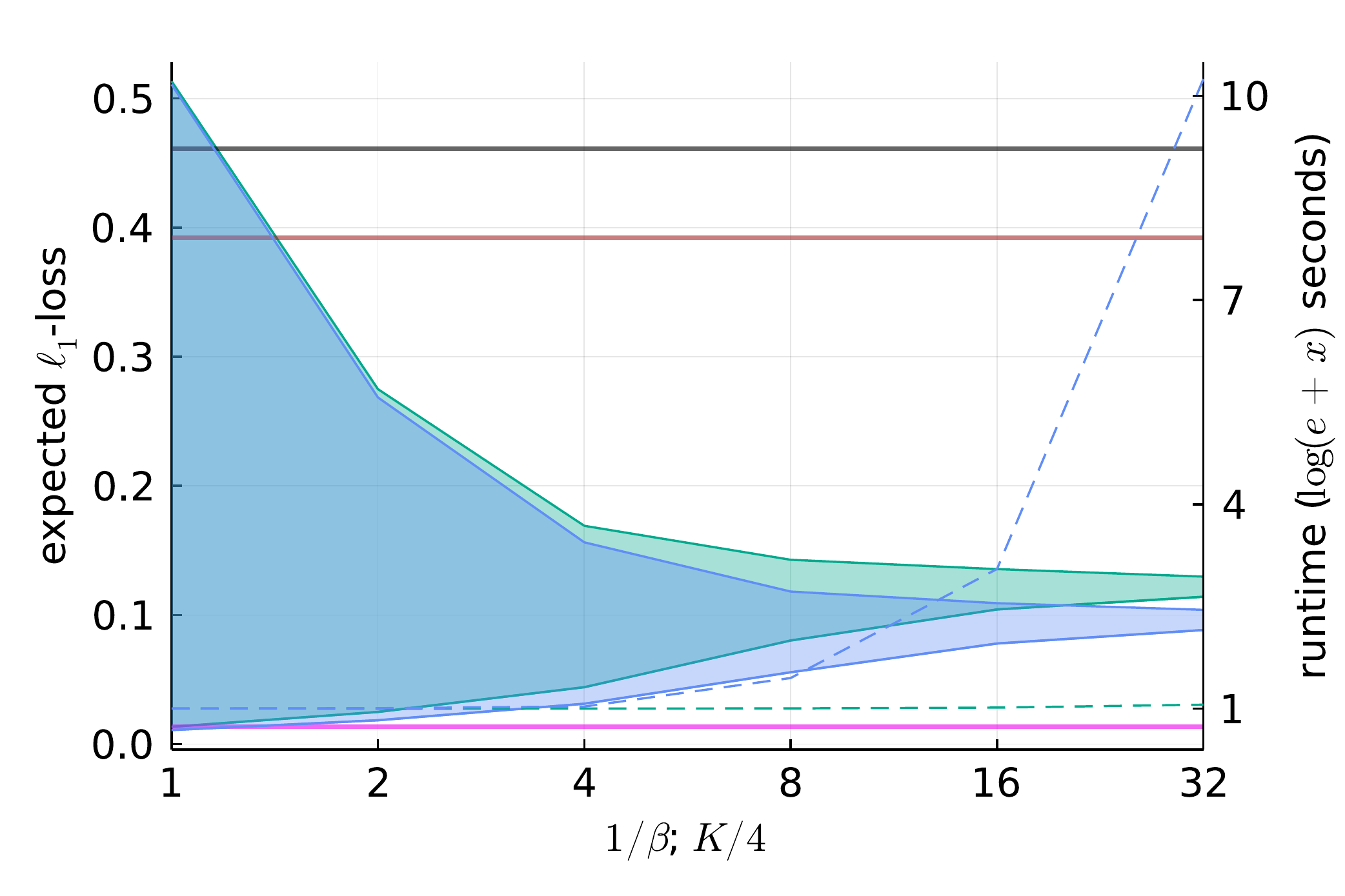}
    \includegraphics[scale = 0.23]{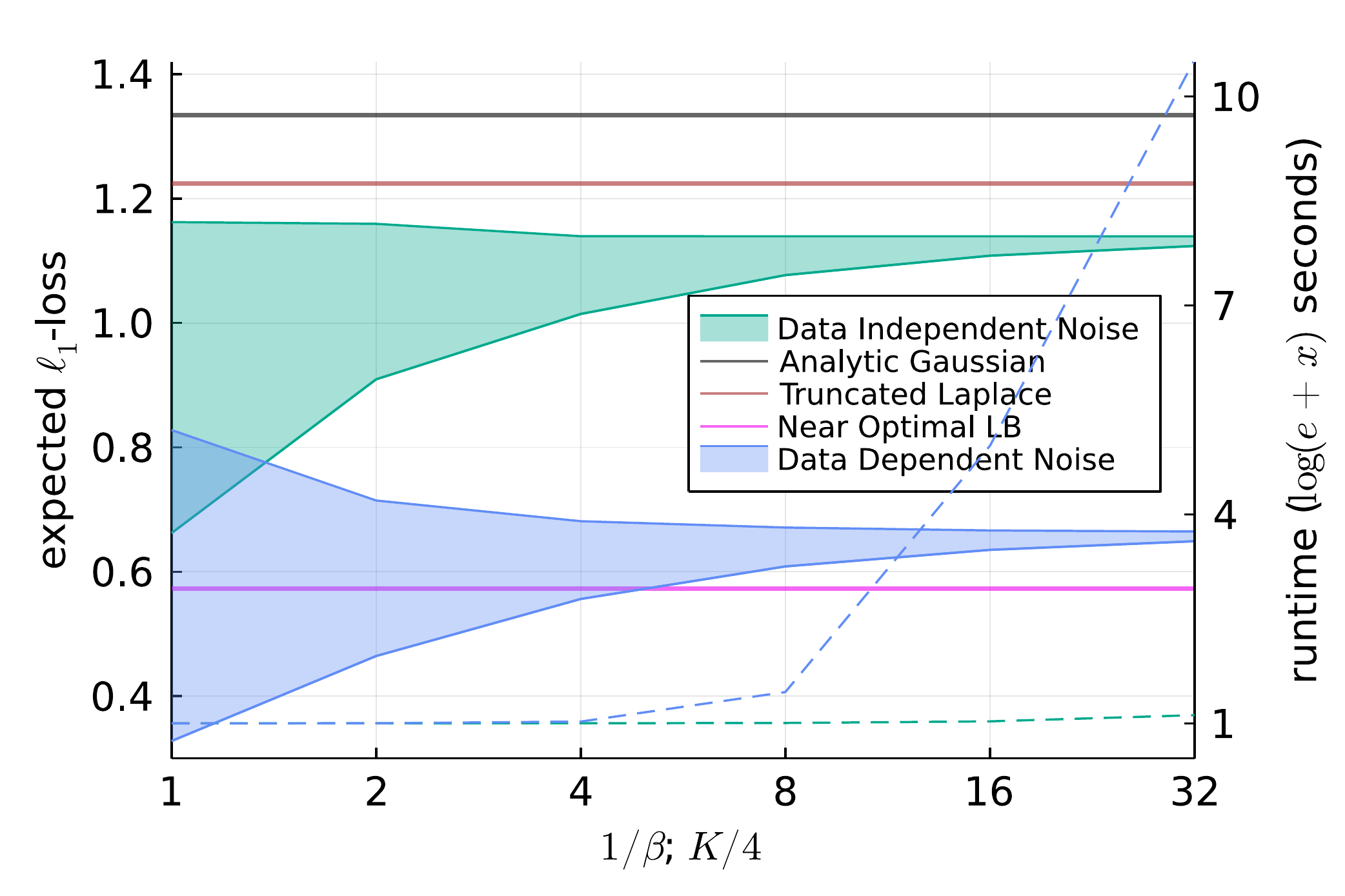}
    \includegraphics[scale = 0.23]{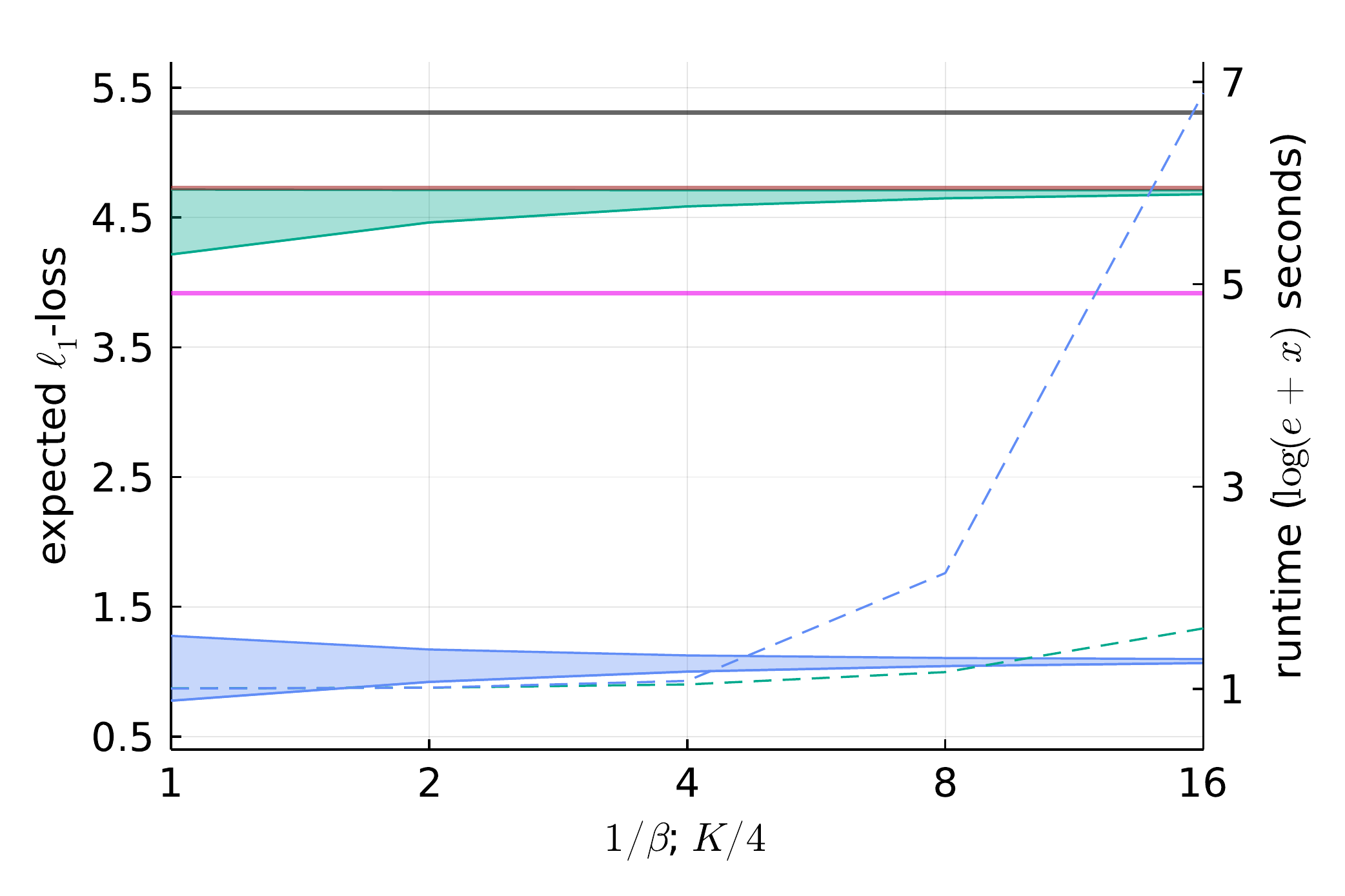}
    }
    \caption{\emph{Comparison of our optimization-based DP schemes with benchmark approaches from the literature on instances with $\Delta f = 2$, $\ell_1$-loss and $| \Phi | = 4$ in low-privacy ($\varepsilon = 5$ and $\delta = 0.25$; left), medium-privacy ($\varepsilon = 1$ and $\delta = 0.2$; middle) and high-privacy ($\varepsilon = 0.2$ and $\delta = 0.05$; right) regimes. All computation times are median values over $10$ repetitions.}}
    \label{fig:ublb}
\end{figure}

Our second experiment investigates the runtime and the convergence of our optimization-based upper and lower bounds in the data independent and data dependent settings. To this end, we consider three privacy regimes: a low-privacy setting with $(\varepsilon, \delta) = (5, 0.25)$, a medium-privacy setting with $(\varepsilon, \delta) = (1, 0.2)$, and a high-privacy setting with $(\varepsilon, \delta) = (0.2, 0.05)$. As in the previous experiment, we compare our optimization-based DP schemes with the analytic Gaussian and the truncated Laplace mechanisms as upper bounds and the `near optimal lower bound' as lower bound. In our DP schemes, we match the support of the truncated Laplace distribution (with the support bounds rounded to nearest integer values) and compute a hierarchy of refined upper and lower bounds by selecting $\beta \in \{ 1, 1/2, \ldots, 1/32 \}$ and---in the data dependent case---$K \in \{ 4, 8, \ldots, 128 \}$. The results are presented in Figure~\ref{fig:ublb}. The figure confirms that the truncated Laplace mechanism is asymptotically optimal among all data independent DP schemes in high-privacy settings. However, the figure also reveals that the truncated Laplace mechanism can be substantially outperformed by our optimization-based data independent noise mechanism in low- and medium-privacy regimes, while it is dominated by our optimization-based data dependent noise mechanism in high-privacy regimes. For low-privacy settings, the difference between our data independent and dependent mechanisms is negligible, but it becomes substantial in medium- and high-privacy regimes, where our data dependent mechanisms significantly outperform the data independent ones. The figure also reveals the computational price to be paid for optimal noise distributions. While the data independent problems were all solved within 2.2 secs, the data dependent problems are more challenging: across all instances, it took up to 5.82 secs (980.44 secs) to reduce the gap between our upper and lower bounds to 10\% (5\%).


\begin{figure}[tb]
    \centering
    \resizebox{1.1\columnwidth}{!}{
    \hskip -1.2cm
    \includegraphics[scale = 0.32]{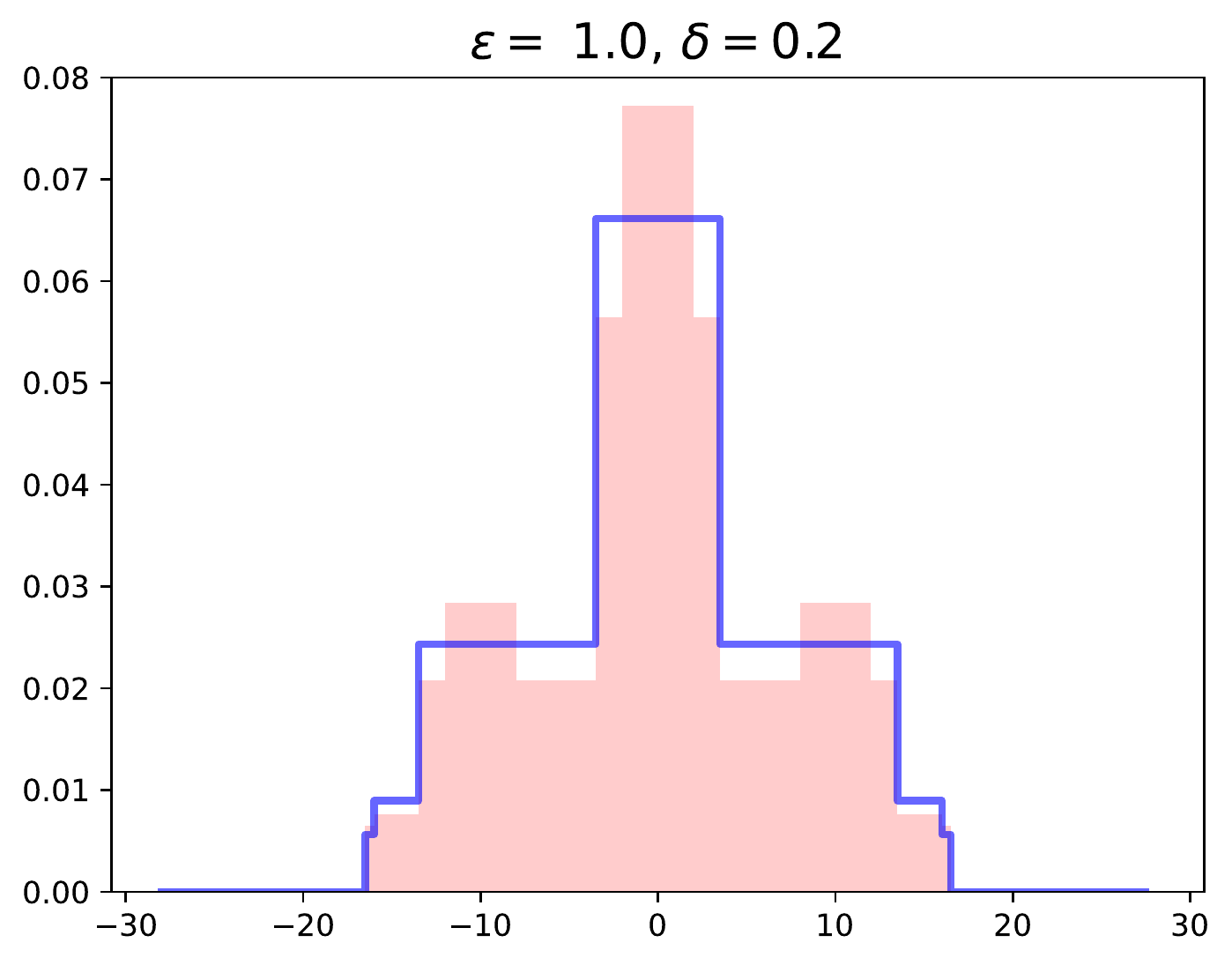}
    \includegraphics[scale = 0.32]{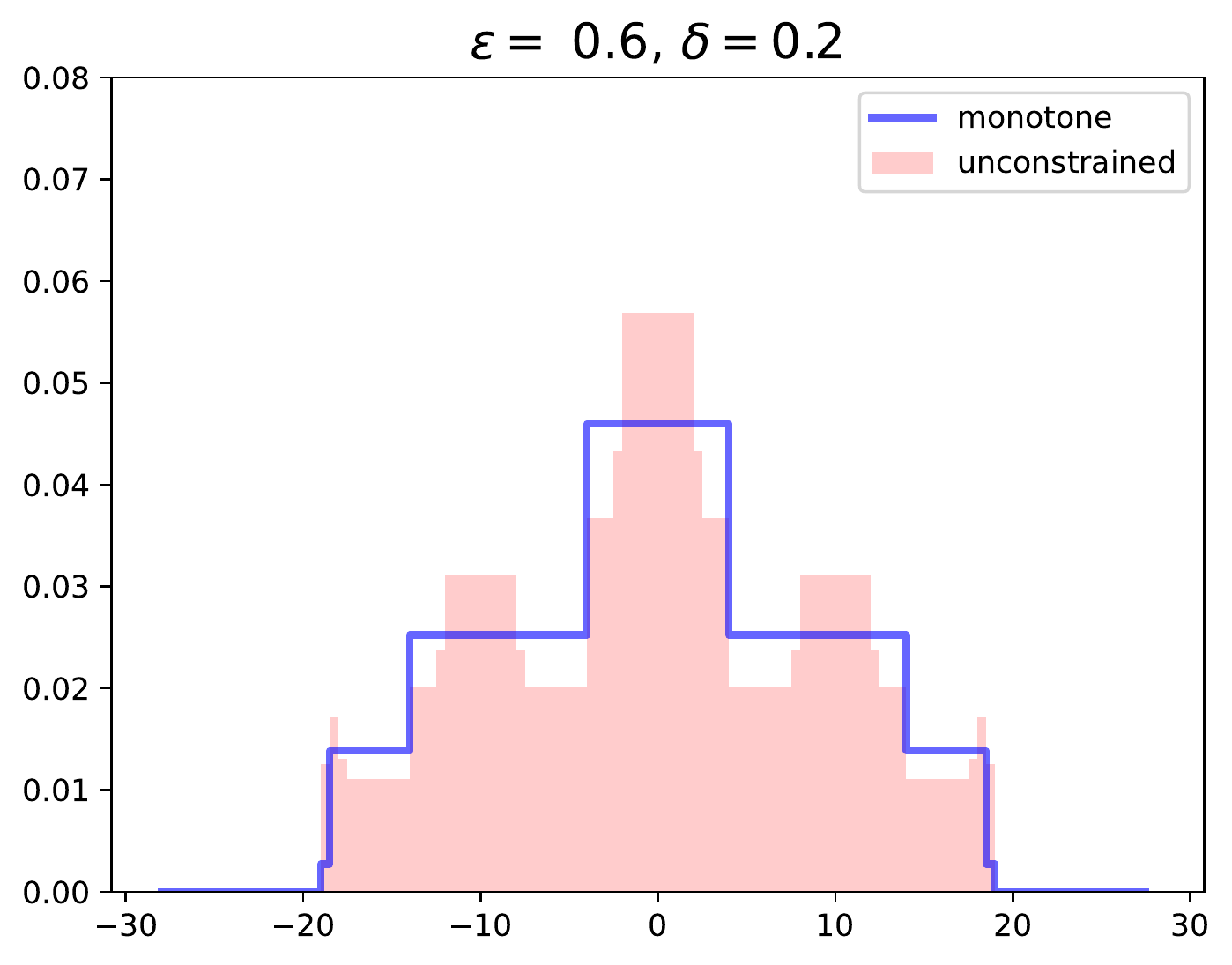}
    \includegraphics[scale = 0.32]{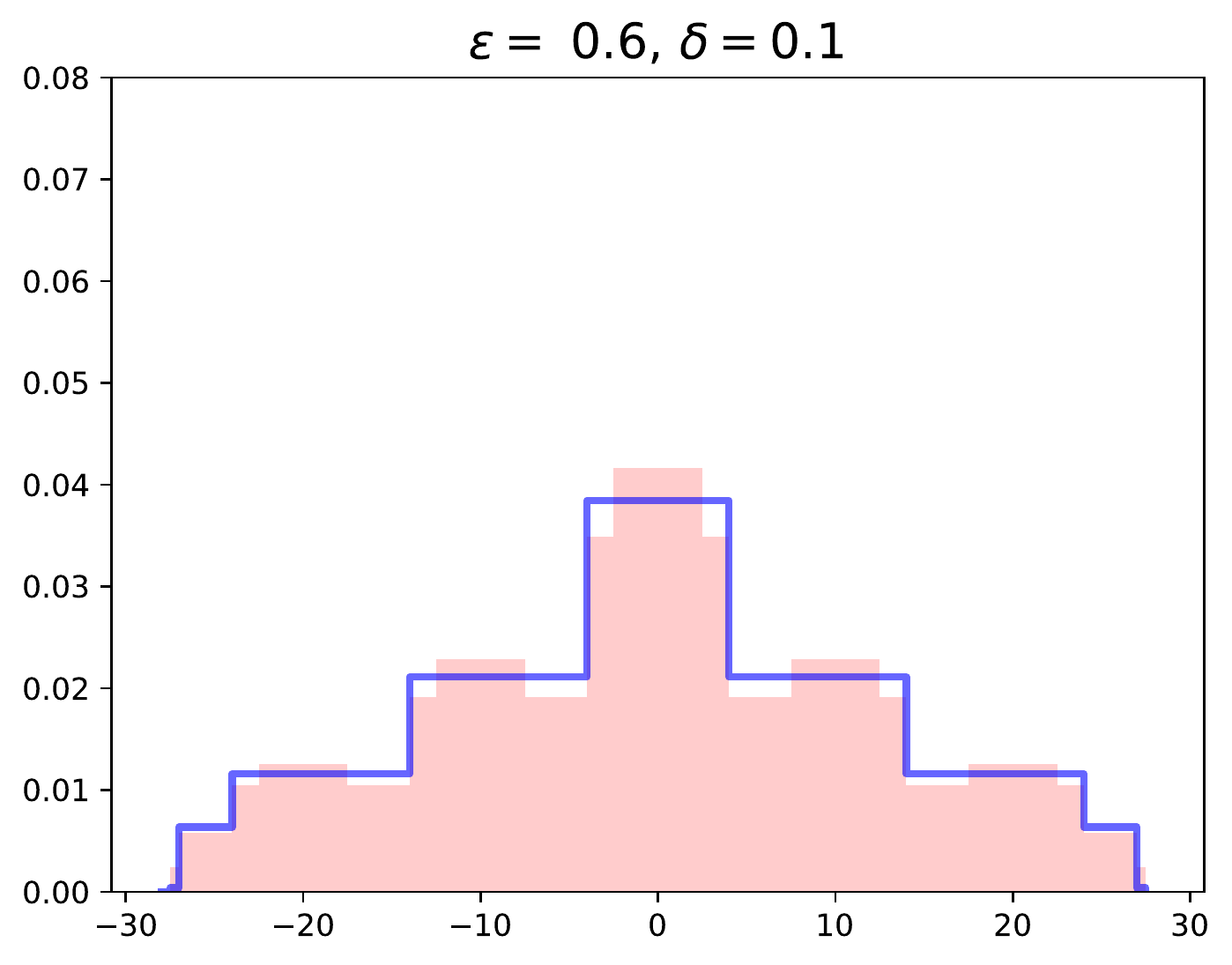}
    }
    \caption{\emph{Optimization-based noise distributions for synthetic data independent instances with $\Delta f = 10$, $\beta = 0.5$, $\ell_1$-loss and various combinations of $\varepsilon$ and $\delta$. The unconstrained distributions are shown in red shading, whereas the best monotone distributions are shown as blue lines.}}
    \label{fig:shape}
\end{figure} 

\begin{figure}[tb]
    \centering
    \resizebox{1.1\columnwidth}{!}{
    \hskip -1.2cm
    \includegraphics[scale = 0.32]{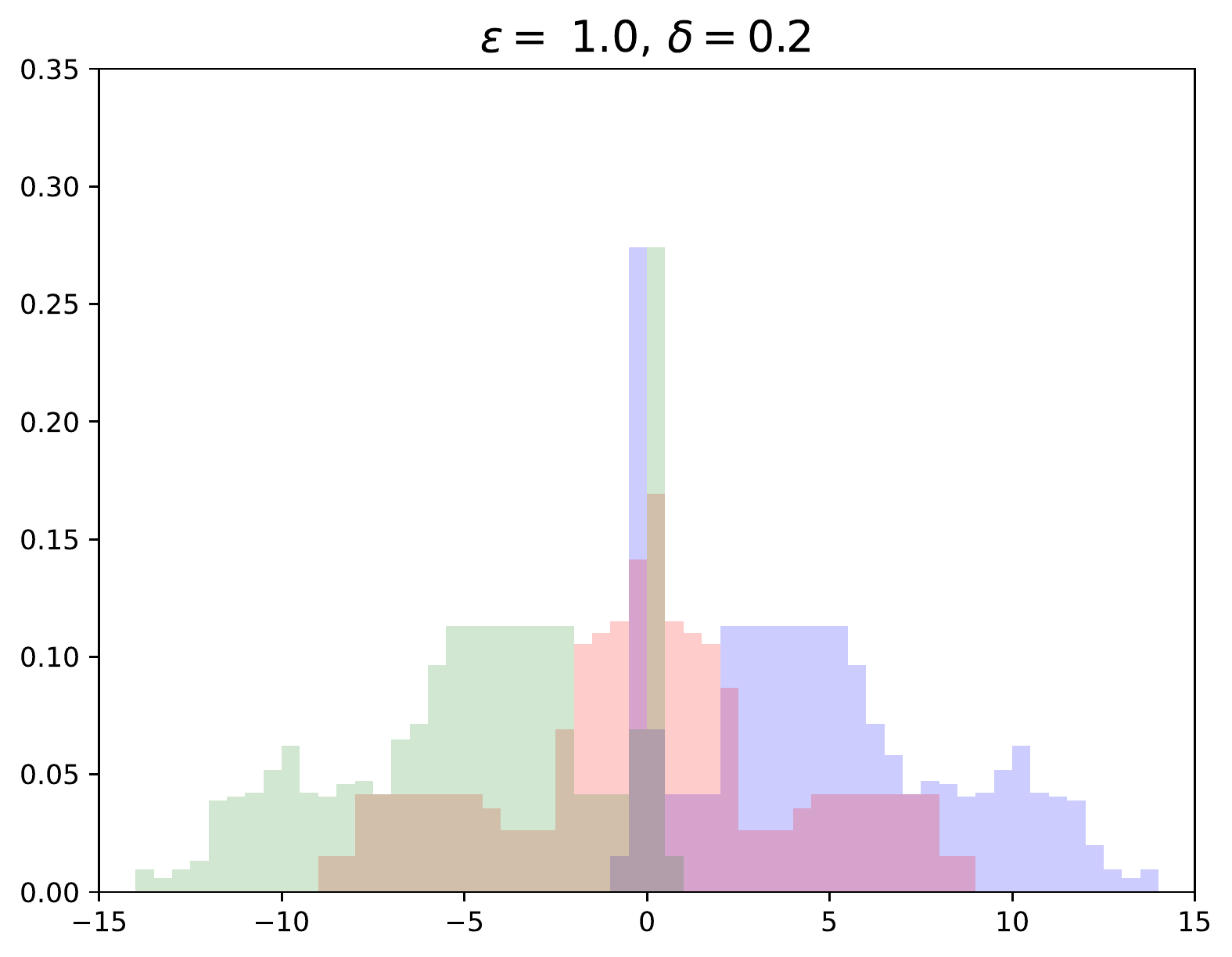}
    \includegraphics[scale = 0.32]{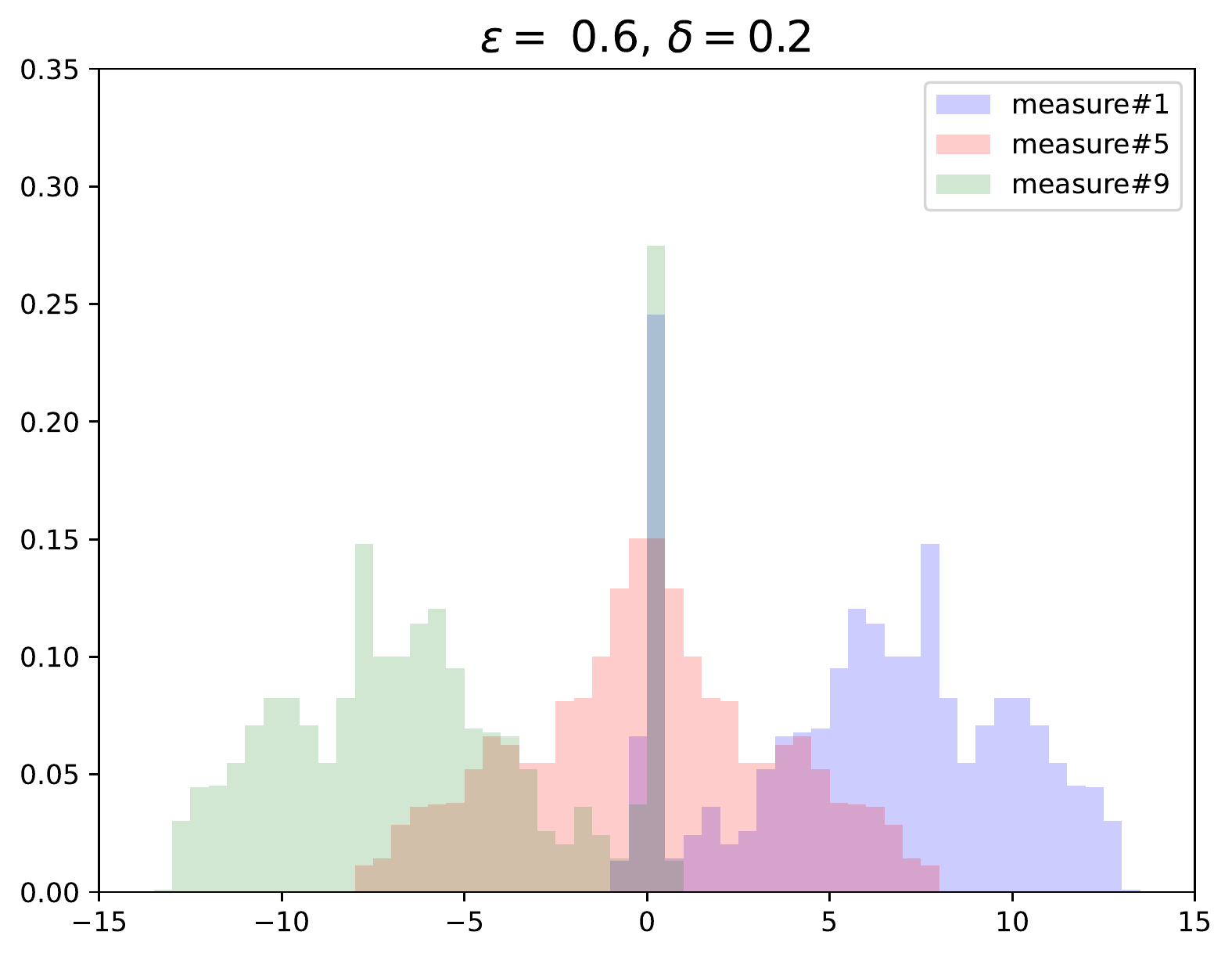}
    \includegraphics[scale = 0.32]{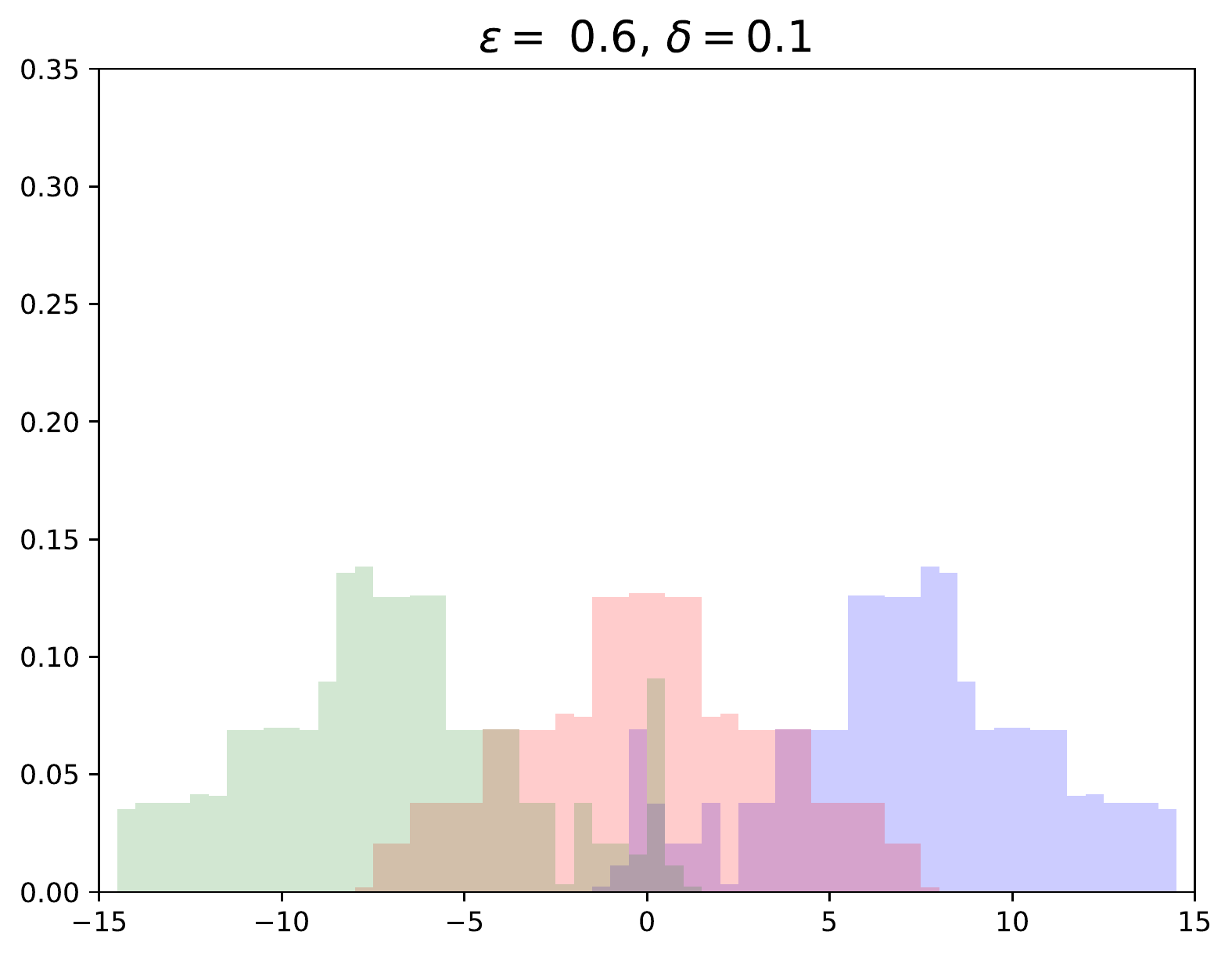}
    }
    \caption{\emph{Optimization-based noise distributions for synthetic data dependent instances with $\Delta f = 10$, $\beta = 0.5$, $\ell_1$-loss and various combinations of $\varepsilon$ and $\delta$. The set $\Phi = [0, 18)$ of query outputs has been partitioned into $9$ intervals of equal length, resulting in $9$ noise distributions.}}
    \label{fig:mult_shape}
\end{figure}

Our final synthetic experiment analyzes the shapes of the noise distributions $\gamma$ obtained by our data independent and data dependent noise optimization problems. For clarity of exposition, we present results for uniform partitions $\{ I_i (\beta) \}_{i \in [\pm L]}$ of the noise realizations as well as the range $\Phi$ of query outputs. We emphasize, however, that better results can normally be obtained by non-uniform partitions that combine finer discretizations around 0 with coarser discretizations of the tails. Figure~\ref{fig:shape} visualizes optimization-based data independent noise distributions for different privacy regimes. We make multiple observations. Firstly, the optimal noise distribution may not be monotone. Indeed, we verified in separate experiments that the lower bounds of the best monotone noise distributions can strictly exceed the upper bounds of the best non-monotone noise distributions (details are available in the GitHub repository). This is noteworthy as all of the benchmark DP mechanisms employ monotone noise distributions, and monotonicity assumptions are commonly made in the literature without scrutinizing their impact on optimality. Secondly, the shapes of the optimization-based noise distributions are non-trivial, and they appear to depend on the problem parameters in a non-trivial manner.  Finally, we note that the shapes of the optimization-based noise distributions differ with the loss function; in particular, asymmetric loss functions (such as the pinball loss) result in asymmetric noise distributions (details are relegated to the GitHub repository). We take our last two observations as an indication of the inherent complexity of optimal noise distributions, which emphasizes the need for optimization-based approaches as opposed to closed-form solutions. Figure~\ref{fig:mult_shape} reports optimization-based data dependent noise distributions for different privacy regimes. In addition to the previous remarks, which continue to apply to the data dependent setting, we additionally observe that the noise distributions corresponding to different intervals of the query output range $\Phi$ differ in non-trivial ways. Again, this confirms our belief that the superior privacy-accuracy trade-offs achieved by optimization-based noise distributions are unlikely to be matched by DP mechanisms relying on closed-form expressions or the tuning of a small number of hyperparameters.

\subsection{Differentially Private Na\"ive Bayes Classification}\label{sec:DPNB}

Given a dataset $(\bm{x}^i, y^i)_{i = 1}^n$ with feature vectors $\bm{x}^i$ comprising numerical features $x^i_v$, $v \in \mathcal{V}_{\text{num}}$, and/or categorical features $x^i_v$, $v \in \mathcal{V}_{\text{cat}}$, as well as categorical outputs $y^i \in \mathcal{C}$, the na\"ive Bayes classifier employs the class-conditional independence assumption to predict the output $c^\star$ corresponding to the feature values $\bm{x} = \bm{\chi}$ of a new sample as the label $c \in \mathcal{C}$ that maximizes
\begin{align*}
  \mathbb{P}[y = c \mid \bm{x}] = \dfrac{\mathbb{P}[y = c] \cdot \prod_{v \in \mathcal{V}_{\text{num}} \cup \mathcal{V}_{\text{cat}}} \mathbb{P}[x_v = \chi_v \mid y = c]}{\mathbb{P}[\bm{x}]}.
\end{align*}
The na\"ive Bayes classifier replaces the unknown probabilities $\mathbb{P}[y = c]$ and $\mathbb{P}[x_v = \chi_v \, | \, y = c]$, $v \in \mathcal{V}_{\text{cat}}$, with their empirical frequencies in the dataset $(\bm{x}^i, y^i)_{i = 1}^n$, and it makes a normality assumption to replace $\mathbb{P}[x_v = \chi_v \, | \, y = c]$, $v \in \mathcal{V}_{\text{num}}$, with an empirical density in the dataset $(\bm{x}^i, y^i)_{i = 1}^n$ using the empirical means $\mu_{ \{ x_v | y = c\}}$ and standard deviations $\sigma_{ \{ x_v | y = c\}}$. We refer to \cite{hastie2009elements} for a detailed description of the na\"ive Bayes classifier.

We follow \cite{vaidya2013differentially} and \cite{dpnb_used} to construct a differentially private na\"ive Bayes classifier. Assuming that the number of training samples $n$ is public knowledge, the only data-related information used by our classifier is the number $n_{\{ y = c \}}$ of training samples with label $c \in \mathcal{C}$, the number $n_{\{ x_v = \chi_v \wedge y = c\}}$ of training samples with label $c \in \mathcal{C}$ whose categorical feature $v \in \mathcal{V}_{\text{cat}}$ attains value $\chi_v$, as well as the conditional empirical means $\mu_{ \{ x_v | y = c\}}$ and standard deviations $\sigma_{ \{ x_v | y = c\}}$ of the numerical features $v \in \mathcal{V}_{\text{num}}$. The post processing property \citep[Proposition 2.1]{dwork2014algorithmic} then allows us to design a differentially private na\"ive Bayes classifier by perturbing these statistics according to their individual sensitivities and using the perturbed statistics to classify new samples. Since $n_{\{ y = c \}}$ and $n_{\{ x_v = \chi_v \wedge y = c\}}$, $v \in \mathcal{V}_{\text{cat}}$, are simple counting queries, they can change by at most $1$ among any two neighboring datasets. For the numerical features $v \in \mathcal{V}_{\text{num}}$, we assume given upper and lower bounds $u_v$ and $l_v$ for the feature values. In this case, the value of $x_v$ can differ by at most $u_v - l_v$ for any two neighboring datasets, which implies that the sensitivity of $\psi_{ \{ x_v | y = c\}}$ and $\sigma_{ \{ x_v | y = c\}}$ is bounded from above by $(u_v - l_v)/n_{\{ y = c \}}$ and $(u_v - l_v)/\sqrt{n_{\{ y = c \}}}$, respectively. We note that $\mu_{ \{ x_v | y = c\}}$ and $\sigma_{ \{ x_v | y = c\}}$ satisfy our surjectivity assumption (\emph{cf.}~Assumption~\ref{assumption_varphi} in Section~\ref{sec:additive} and Assumption~\ref{assumptions_f} in Section~\ref{sec:nonlinear}), whereas the assumption is violated by the count queries $n_{\{ y = c \}}$ and $n_{\{ x_v = \chi_v \wedge y = c\}}$. Thus, our optimization-based approaches provide feasible but potentially overly conservative noise distributions.

\begin{table}[tb]
    $\mspace{-40mu}$
    \resizebox{1.1 \columnwidth}{!}{%
     \begin{tabular}{lc@{$\;\;$}c@{$\;\;$}c@{$\;\;$}c|ccccc|ccccc}
\toprule
\multicolumn{5}{c}{UCI dataset descriptions} & \multicolumn{5}{c}{In-sample errors} & \multicolumn{5}{c}{Out-of-sample errors}  \\
\midrule
Dataset &  $n$ & $|\mathcal{V}_{\mathrm{num}}|$ & $|\mathcal{V}_{\mathrm{cat}}|$ & $|\mathcal{C}|$ &  GN  &  TLN & OPT  & NB  & \textit{Imp} & GN  & TLN & OPT  & NB & \textit{Imp}\\
\midrule
      \ul{post-operative}  &  86  & 1 & 7 & 2 &  36.94\% (33.41\%) & 32.42\% & \textbf{31.43\%} (29.61\%)  & 25.65\% & \textit{14.62\%} & 41.28\% (39.61\%) & 39.07\% & \textbf{38.30\%} (37.03\%) &  35.12\% & \textit{19.53\%} \\
      \ul{adult}  & 45,222 & 5 & 8 & 2 & 38.78\% (21.98\%) & 21.73\% & \textbf{20.49\%} (18.84\%) & 17.37\% & \textit{28.25\%} &  38.82\% (22.04\%)& 21.80\% & \textbf{20.56\%} (18.91\%) & 17.45\% & \textit{28.28\%}  \\
      \ul{breast-cancer} & 683 & 0 & 9 & 2 & 2.45\% (2.23\%) & 2.20\% & \textbf{2.17\%} (2.14\%) & 2.12\% & \textit{36.80\%} & 3.82\% (3.57\%) & 3.54\% & \textbf{3.51\%} (3.48\%)  & 3.46\%  & \textit{36.75\%} \\
      \ul{contraceptive} & 1,473 & 2 & 7 & 3 & 58.84\% {(52.84\%)} & 52.38\% & \textbf{51.32\%} (50.46\%) & 49.14\% & \textit{32.64\%} & 59.53\% {(54.22\%)} & 53.76\% & \textbf{52.89\%} (52.19\%)  & 51.08\%  & \textit{32.55\%} \\
      \ul{dermatology} & 366 & 2 & 32 & 6 & 1.68\% {(1.03\%)} & 0.91\% & \textbf{0.90\%} (0.60\%) & 0.49\% &\textit{1.07\%} & 35.96\% {(35.61\%)} & 35.53\% & \textbf{35.52\%} (35.28\%)  & 35.25\% & \textit{1.33\%}\\
      \ul{cylinder-bands} & 539 & 19 & 14 & 2  & 40.95\% {(35.63\%}) & 34.69\% & \textbf{33.98\%} (31.46\%)  & 23.43\% & \textit{6.35\%} &  41.83\% {(37.20\%)} & 36.42\% & \textbf{35.81\%} (33.70\%)  & 26.89\% & \textit{6.40\%} \\
      \ul{annealing} & 898 & 6 & 18 & 5 & 13.65\% {(7.93\%)} & 7.47\% & \textbf{7.39\%} (7.31\%)  & 7.32\%  & \textit{51.56\%} & 14.23\% {(8.38\%)} & 7.89\% & \textbf{7.80\%} (7.69\%) & 7.70\% & \textit{44.41\%} \\
      \ul{spect} & 160 & 0 & 22 & 2 & 26.46\% {(25.99\%)} & 25.89\% & \textbf{25.78\%} (25.77\%) & 25.67\% & \textit{47.48\%} & 29.34\% {(28.77\%)} & 28.66\% & \textbf{28.59\%} (28.57\%) & 28.50\% & \textit{43.74\%} \\ 
      \ul{bank} & 45,211 & 4 & 12 & 2 & 13.98\% {(12.40\%)} & 12.35\% & \textbf{12.32\%} (12.24\%) & 12.13\% & \textit{14.20\%} & 14.05\% {(12.48\%)} & 12.43\% & \textbf{12.40\%} (12.33\%)  & 12.22\% & \textit{14.26\%} \\
      \ul{abalone} & 4,177 & 7 & 1 & 2  & 39.25\% {(31.85\%)} & 31.36\% & \textbf{30.64\%} (28.75\%) & 26.46\% & \textit{14.77\%} & 39.42\% {(32.14\%)} & 31.66\% & \textbf{30.95\%} (29.09\%)  & 26.86\% & \textit{14.76\%}\\
      \ul{spambase} & 4,601 & 57 & 0 & 2 & 39.40\% {(37.09\%)} & 35.50\% & \textbf{34.99\%} (31.56\%) & 18.09\%  & \textit{2.91\%} & 39.26\% {(36.97\%)} & 35.41\% & \textbf{34.91\%} (31.54\%) & 18.26\% & \textit{2.90\%}\\
      \ul{ecoli} & 336 & 5 & 2 & 2 & 48.09\% {(31.40\%)} & 27.00\% & \textbf{26.51\%} (18.76\%) & 3.97\% & \textit{2.15\%} &48.33\% {(31.81\%)} & 27.44\% & \textbf{26.95\%} (19.28\%)  & 4.52\% & \textit{2.12\%} \\
      \ul{absent} & 740 & 12 & 8 & 2 & 45.95\% {(30.22\%)} & 28.24\% & \textbf{28.01\%} (26.56\%)  & 24.56\%  & \textit{2.60\%} &  47.14\% {(33.99\%)} & 32.26\% & \textbf{32.03\%} (30.73\%)  & 29.18\%  & \textit{2.59\%} \\
\bottomrule
\end{tabular}}
    \caption{\textit{In-sample and out-of-sample errors of our optimization-based differentially private na\"ive Bayes classifier as well as several DP mechanisms from the literature on UCI datasets. Bold printing highlights the smallest errors obtained across all data independent DP mechanisms.}}
    \label{tab:NB} 
\end{table}

Table~\ref{tab:NB} presents the in-sample and out-of-sample errors of our optimization-based differentially private na\"ive Bayes classifier using the $\ell_1$-loss in a data independent (``OPT'') as well as data dependent (``(OPT)'') setting on the most popular UCI classification datasets \citep{UCI}. The table also compares our results with diffentially private na\"ive Bayes classifiers employing a Gaussian noise (``GN''), an analytic Gaussian noise (``(GN)'') and a truncated Laplace noise (``TLN''), as well as the classical (non-private) na\"ive Bayes classifier (``NB''). We fix $(\varepsilon, \delta) = (1, 0.1)$ in all DP mechanisms. The reported errors are mean errors over 100 random splits of the datasets into training sets (80\% of the data) and test sets (20\% of the data) as well as, for each split, 1,000 simulations of all differentially private na\"ive Bayes implementations. The column `Imp' records the percentage of the gap between NB and the best method from the literature (which turns out to be TLN) that is closed by OPT. The table reveals that our optimization-based data independent and data dependent noise distributions consistently outperform the considered competitors. To see whether this outperformance is statistically significant, we computed the p-values of a t-test with a null hypothesis that the second best approach is as good as OPT. The t-test averages the 1,000 simulated errors for each training set-test set split and considers the differences of the 100 averages corresponding to different training set-test set splits \citep{salzberg1997comparing}. In all experiments, the p-values are less than $10^{-7}$, except for \underline{dermatology} where the p-value is $10^{-1}$. Further details on the experimental setting as well as the applied the t-tests are relegated to the GitHub repository.

\subsection{Differentially Private Proximal Coordinate Descent}

Given a dataset $(\bm{x}^i, y^i)_{i = 1}^n$ with feature vectors $\bm{x}^i \in \mathbb{R}^d$ comprising numerical and/or categorical features as well as binary outputs $y^i \in \{ -1, +1 \}$, the $\ell_1$-regularized logistic regression assumes that $\mathbb{P}[y \, | \, \bm{x} = \bm{\chi}] = [1 + \exp(-y \cdot \bm{h}^0{}^\top \bm{\chi})]^{-1}$ for some unknown hyperplane $\bm{h}^0 \in \mathbb{R}^d$, and it determines a hyperplane $\bm{h}^\star \in \mathbb{R}^d$ that minimizes the empirical logistic loss
\begin{equation*}
    \dfrac{1}{n} \cdot \sum_{i=1}^n \log(1 + \exp(- y^i \cdot \bm{h}^\top \bm{x}^i)) + \lambda \cdot  ||\bm{h}||_1,
\end{equation*}
where $\lambda > 0$ is a hyperparameter. Subsequently, the output of a new sample with feature values $\bm{x} = \bm{\chi}$ is predicted to be the label $y \in \{ -1, +1 \}$ that maximizes $[1 + \exp(-y \cdot \bm{h}^\top \bm{\chi})]^{-1}$.

To solve the logistic regression problem, proximal coordinate descent \citep{friedman2010regularization, richtarik2014iteration} starts at a randomly selected initial solution $\bm{h}^0$ and conducts $t = 1, \ldots, T$ iterations, $T \in \mathbb{N}$, each of which applies the proximal operator to a random subset $i^t_1, \ldots, i^t_K \in [d]$ of the components via
\begin{equation}\label{proximal_update}
    h^{t,k}_l = \mathrm{prox}_{\lambda |\cdot|} \left({h}^{t,k-1}_l - \dfrac{1}{n}\cdot \left[ \sum_{i=1}^n \dfrac{\exp(- y^i \cdot \bm{h}^{t,k}{}^\top \bm{x}^i)}{1 + \exp(- y^i \cdot \bm{h}^{t,k}{}^\top \bm{x}^i)} \cdot (-y^i \cdot {x}^i_l) \right] \right)
\end{equation}
if $l = i^t_k$, and $h^{t,k}_l = h^{t,k-1}_l$ otherwise, for all $k = 1, \ldots, K$. Here, we fix $\bm{h}^{t,0} = \bm{h}^{t-1}$, and $\mathrm{prox}_{\lambda |\cdot|}(\cdot)$ denotes the proximal operator \citep{parikh2014proximal} associated with the $\ell_1$-regularization:
\begin{align*}
    \mathrm{prox}_{\lambda |\cdot|} (w_j) := \underset{v \in \mathbb{R}}{\arg \min} \left\{\frac{1}{2} \cdot (w_j - v)^2 + \lambda \cdot |v| \right\} =
    \begin{cases}
    w_j - \lambda & \text{ if } w_j \geq \lambda \\
    w_j + \lambda& \text{ if } w_j \leq - \lambda\\
    0            & \text{ if } |w_j| \leq \lambda
    \end{cases}
\end{align*}
Upon completion of the $K$ applications of the proximal operator in iteration $t$, we set $\bm{h}^t = (1 / K) \cdot \sum_{k = 1}^K \bm{h}^{t,k}$ and continue with iteration $t + 1$. The algorithm terminates with $\bm{h}^T$ as an approximately optimal solution to the regularized logistic regression problem.

We follow \cite{mangold2022differentially} to construct a differentially private proximal coordinate descent method for the logistic regression problem. The only data-related information used by our algorithm is contained in the proximal updates~\eqref{proximal_update}. Assuming that the feature vectors $\bm{x}^i$ are normalized so that $\lVert \bm{x}^i \rVert_{\infty} \leq 1, \ i \in [n]$, we have
\begin{align*}
    \sum_{i=1}^n \underbrace{\dfrac{\exp(- y^i \cdot \bm{h}^{t,k}{}^\top \bm{x}^i)}{1 + \exp(- y^i \cdot \bm{h}^{t,k}{}^\top \bm{x}^i)}}_{\in (0,1)} \cdot \underbrace{(-y^i \cdot {x}^i_l)}_{\in [- 1, + 1]} \in (-n, n),
\end{align*}
and thus the sensitivity of this summation, which is determined by the maximum change achievable by modifying a single training sample $i \in [n]$, is $2$. The post processing property \citep[Proposition 2.1]{dwork2014algorithmic} then allows us to design a differentially private proximal coordinate descent method by perturbing the sum inside the proximal updates~\eqref{proximal_update} accordingly.

\begin{table}[tb]
    $\mspace{-40mu}$
    \resizebox{1.1 \columnwidth}{!}{%
     \begin{tabular}{lcc|ccccc|ccccc}
\toprule
\multicolumn{3}{c}{UCI dataset descriptions} & \multicolumn{5}{c}{In-sample errors} & \multicolumn{5}{c}{Out-of-sample errors}  \\
\midrule
Dataset &  $n$ & $d$ &  GN  &  TLN & OPT  & PCD  & \textit{Imp} & GN  & TLN & OPT  & PCD & \textit{Imp}\\
\midrule
      \ul{post-operative}  &  86  & 14 &  40.23\% (34.96\%) & 33.75\% & \textbf{33.37\%} {(30.67\%)} & 27.31\% & \textit{5.94\%} & 45.67\% (42.99\%) & 42.08\% & \textbf{41.86\%} {(41.22\%)} & 35.56\% & \textit{3.38\%}\\ 
      \ul{adult}  & 45,222 & 57 & 19.77\% (19.77\%) & 19.75\% & \textbf{19.73\%} {(19.67\%)} & 19.75\% & \textit{573.31\%} & 19.79\% (19.79\%) & 19.77\% & \textbf{19.75\%} {(19.69\%)} & 19.77\% & \textit{640.17\%} \\
      \ul{breast-cancer} & 683 & 26 &  4.60\% (4.36\%) & 4.36\% & \textbf{4.34\%} {(3.93\%)} & 4.31\% & \textit{28.05\%} & 4.75\% (4.52\%) & 4.51\% & \textbf{4.50\%} {(4.09\%)} & 4.46\% & \textit{22.25\%} \\
      \ul{contraceptive} & 1,473 & 18 & 38.52\% (38.30\%) & 38.29\% & \textbf{38.27\%} {(37.39\%)} & 38.21\% & \textit{19.28\%} & 39.90\% (39.71\%) & 39.70\% & \textbf{39.69\%} {(38.93\%)} & 39.64\% & \textit{20.17\%}\\
      \ul{dermatology} & 366 & 98  & 18.09\% (14.56\%) & 14.23\% & \textbf{14.14\%} {(7.95\%)} & 13.14\% & \textit{8.07\%} & 21.38\% (18.21\%) & 17.93\% & \textbf{17.85\%} {(11.96\%)} & 16.95\% & \textit{8.17\%}\\
      \ul{cylinder-bands}& 539 & 63 &30.27\% (28.74\%) & 28.65\% & \textbf{28.59\%} {(25.30\%)} & 28.20\% & \textit{12.96\%} & 32.65\% (31.36\%) & 31.29\% & \textbf{31.24\%} {(28.56\%)} & 30.90\% & \textit{14.65\%} \\
      \ul{annealing} & 898 & 42 & 16.70\% (16.27\%) & 16.30\% & \textbf{16.29\%} {(14.74\%)} & 16.20\% & \textit{6.16\%} & 17.52\% (17.10\%) & 17.13\% & \textbf{17.13\%} {(15.61\%)} & 17.03\% & \textit{0.83\%}\\
      \ul{spect} &160 &23  & 28.89\% (23.62\%) & 22.74\% & \textbf{22.57\%} {(18.60\%)} & 19.61\% & \textit{5.23\%} & 31.47\% (27.19\%) & 26.37\% & \textbf{26.24\%} {(23.58\%)} & 23.71\% & \textit{4.84\%}\\
      \ul{bank} & 45,211 & 44 &12.20\% (12.20\%) & 12.20\% & \textbf{12.20\%} {(11.69\%)} & 12.22\% & \textit{25.43\%} & 12.21\% (12.21\%) & 12.21\% & \textbf{12.21\%} {(11.71\%)} & 12.23\% & \textit{20.31\%} \\
      \ul{abalone} & 4,177 & 10 &27.45\% (27.44\%) & 27.44\% & \textbf{27.43\%} {(27.34\%)} & 27.43\% & \textit{87.76\%} & 27.53\% (27.52\%) & 27.52\% & \textbf{27.51\%} {(27.43\%)} & 27.51\% & \textit{79.42\%} \\
      \ul{spambase} & 4,601 & 58 & 39.25\% (39.26\%) & 39.23\% & \textbf{39.21\%} {(38.79\%)} & 39.26\% & \textit{65.85\%} & 39.60\% (39.60\%) & 39.57\% & \textbf{39.55\%} {(39.13\%)} & 39.61\% & \textit{41.07\%}\\
      \ul{ecoli} & 336 & 8 &9.38\% (7.29\%) & 7.09\% & \textbf{7.01\%} {(6.31\%)} & 6.52\% & \textit{14.81\%} & 9.88\% (7.73\%) & 7.53\% & \textbf{7.45\%} {(6.70\%)} & 6.96\% & \textit{13.67\%}\\
      \ul{absent} & 740 & 70 & 33.77\% (32.88\%) & 32.78\% & \textbf{32.74\%} {(29.45\%)} & 32.54\% & \textit{17.16\%} & 35.63\% (34.83\%) & 34.74\% & \textbf{34.68\%} {(31.95\%)} & 34.52\% & \textit{26.34\%}\\
      \ul{colon-cancer} & 62 & 2,000 & 18.72\% (10.85\%) & 9.20\% & \textbf{8.62\%} {(0.03\%)} & 0.00\% & \textit{6.34\%} & 32.09\% (31.63\%) & 31.62\% & \textbf{31.54\%} {(30.41\%)} & 30.67\% & \textit{7.77\%}\\
\bottomrule
\end{tabular}}
    \caption{\textit{In-sample and out-of-sample errors of our optimization-based differentially private logistic classifier as well as several DP mechanisms from the literature on UCI datasets. Bold printing highlights the lowest errors obtained across all data independent DP mechanisms.}}
    \label{tab:PCD} 
\end{table}

Table~\ref{tab:PCD} presents the in-sample and out-of-sample errors of our optimization-based differentially private proximal coordinate descent algorithm in a data independent (``OPT'') as well as data dependent (``(OPT)'') setting for $T = 100$ iterations and $K = \lceil d / 4 \rceil$ proximal updates per iteration, regularization parameter $\lambda = 10^{-8}$ and $(\varepsilon, \delta) = (1, 0.1)$. {\color{black} While our data independent algorithm minimizes the expected $\ell_1$-loss, our data dependent algorithm performs much better under a handcrafted loss function that resembles the $\ell_1$-loss in the vicinity of the origin but has steep slopes of -1,000 and 1,000 for large negative and positive values away from the origin, respectively. The large slopes penalize switching the sign of gradient, which tends to slow down convergence (recall from Figure~\ref{fig:mult_shape} that noise distributions associated with negative values tend to have large probability mass on the positive side and vice versa).} We use the same datasets as in Section~\ref{sec:DPNB}, but we \emph{(i)} convert non-binary output labels into binary ones via binning (if the output is ordinal) or distinguishing the majority class from all other classes (if the output is nominal) and \emph{(ii)} apply one-hot encoding for the nominal input features. Additionally, as the proximal coordinate descent method is commonly used for datasets where $d \gg n$, we also include the \ul{colon-cancer} dataset that is available in LIBSVM \citep{chang2011libsvm}. As in the previous section, we compare our optimization-based algorithms with differentially private proximal coordinate descent methods employing a Gaussian noise (``GN''), an analytic Gaussian noise (``(GN)'') and a truncated Laplace noise (``TLN''), as well as the classical (non-private) proximal coordinate descent scheme (``PCD''). As before, the reported errors are mean errors over 100 random splits of the datasets into training sets (80\% of the data) and test sets (20\% of the data) as well as, for each split, 1,000 simulations of all differentially private algorithm implementations. The column `Imp' records the percentage of the gap between PCD and the best method from the literature that is closed by OPT. As in the experiment from the previous section, our optimization-based data independent and data dependent noise distributions consistently outperform the considered competitors. The p-values of a t-test similar to the previous section were always smaller than $10^{-7}$, except for the \ul{annealing}, \ul{breast-cancer} and \ul{contraceptive} datasets, where the $p$-values are $10^{-4}$, and except for the \ul{bank} dataset, where there is no significance. Further details on the experimental setting can be found in the GitHub repository. Interestingly, we observe that on all datasets except for \ul{post-operative}, our data dependent differentially private classifier (OPT) outperforms the non-private classifier PCD. This result is due to our handcrafted loss function (see above), and it does not hold for $\ell_1$- and $\ell_2$-losses. Effects of this types have been observed previously in gradient-based learning algorithms (see, \textit{e.g.}, \citealt{neelakantan2015adding}), and they further highlight the benefits of an optimization-based approach towards DP, which readily caters for non-standard loss functions that can be tuned towards the task at hand.

\section{Conclusions}\label{sec:conclusions}

With the widespread adoption of analytics, privacy concerns have witnessed a remarkable resurgence in the public discourse. While DP has established itself as a predominant privacy paradigm in both academic research and industrial practice, the existing DP mechanisms almost exclusively focus on privacy to the detriment of accuracy. The few methodological studies on the privacy-accuracy trade-off focus on asymptotic performance and/or restricted classes of mechanisms.

In our view, the privacy-accuracy trade-off is most naturally studied through the lens of optimization theory, which gives rise to infinite-dimensional DP mechanism design problems that are similar to---but at the same time notedly distinct from---continuous linear programs and distributionally robust optimization problems. We developed a hierarchy of converging upper and lower bounds on these problems that result in DP mechanisms with rigorous privacy guarantees and deterministic bounds on their accuracy. Our numerical results demonstrate that our mechanisms can achieve significant improvements on both synthetic and real-world problems.

A key advantage of an optimization-based DP approach is its versatility. Our upper and lower bounds can be readily extended to incorporate monotonicity and symmetry constraints as well as a bounded range for the query output in the case of data dependent mechanisms. We can account for different loss functions, and multiple loss functions can be combined in a multi-objective framework. Our approach also allows us to incorporate tighter bounds on the probability of distinguishing events $A \in \mathcal{F}$, $\mathbb{P} [ \mathcal{A}(D) \in A ] > 0$ and $\mathbb{P} [ \mathcal{A}(D') \in A ] = 0$ for some $(D, D') \in \mathcal{N}$, that enable an adversary to exclude certain databases altogether \citep[Remark 1.3]{dwork2016concentrated}. 

We regard this work as a first step towards optimization-based DP, and it opens up several avenues for future research. Firstly, while our approach generalizes to multi-dimensional query functions $f$ via the composition theorem, better results can be expected if one directly optimizes over multi-dimensional noise distributions. Since a na\"ive implementation of our algorithms would scale exponentially in that dimension, this may necessitate the development of further approximations. {\color{black} Our GitHub supplement reports on some initial numerical experiments for the multi-dimensional case as well as a potential avenues to alleviate the curse of dimensionality.} Secondly, our work measures accuracy through a simple loss function. In many interesting practical applications, the query output may be the solution to an optimization problem, and in those cases accuracy may be best measured in terms of the expected performance of the perturbed output (\emph{e.g.}, the expected total discounted reward of the perturbed policy in the context of differentially private Markov decision processes). {\color{black} Thirdly}, DP as currently defined in the literature is tailored to the traditional noise distributions such as the Laplace and the Gaussian mechanisms. The versatility of an optimization-based view on DP allows us to explore other, potentially more general notions of differential privacy as well. {\color{black} Finally, it would be instructive to further study the connections between DP and robust optimization. Interesting avenues for further exploration include the development of uncertainty set-based mechanisms for DP that may avoid the curse of dimensionality \citep{bertsimas2004price}, as well as applying robust satisficing techniques~\citep{long2023robust} to offer group privacy guarantees. We refer to GitHub supplement for a more detailed discussion of those topics.}

\ACKNOWLEDGMENT{Financial support from the EPSRC grants EP/R045518/1 and EP/W003317/1 as well as NSF China grant 12301403 is gratefully acknowledged. The authors thank Damien Desfontaines, Daniel Kuhn and Juba Ziani for insightful discussions. The first author acknowledges the support of The Alan Turing Institute’s Enrichment Scheme. The authors thank the area editor Sam Burer as well as the anonymous review team for their constructive and insightful comments that substantially strengthened the paper.}


%
%
%

\bibliographystyle{informs2014} 
\bibliography{refinf.bib}

\newpage
\appendix
\addtocontents{toc}{\protect\setcounter{tocdepth}{1}}
\addcontentsline{toc}{section}{Appendices}

\newtheorem{lemmaA}{Lemma}
\renewcommand{\thelemmaA}{\Alph{section}.\arabic{lemmaA}}

\newtheorem{observationA}{Observation}
\renewcommand{\theobservationA}{\Alph{section}.\arabic{observationA}}
\setcounter{page}{1}

\begin{center}
\section*{\color{black} E-Companion}   
\end{center}

\section{Proofs of Section~\ref{sec:additive}}

{\color{black}
\subsection{Proof of Observation~\ref{observation_primal}}

Assumption~\ref{assumption_varphi} allows us to replace the DP constraint in~\eqref{problem:integral_main:preliminary_formulation} with
\begin{equation*}
    \mspace{-20mu}\begin{array}{r@{\quad}l@{\quad}l}
        &  \displaystyle \int_{x \in \mathbb{R}}\indicate{x \in A} \diff \gamma(x)  \leq e^\varepsilon \cdot \displaystyle \int_{x \in \mathbb{R}} \indicate{f(D') - f(D) + x \in A} \diff \gamma(x) + \delta & \displaystyle \forall (D, D') \in \mathcal{N}, \ \forall A \in \mathcal{F} \\[4mm]
        \displaystyle\Longleftrightarrow &
        \displaystyle \int_{x \in \mathbb{R}}\indicate{x \in A} \diff \gamma(x)  \leq e^\varepsilon \cdot \displaystyle \int_{x \in \mathbb{R}} \indicate{x + \varphi \in A} \diff \gamma(x) + \delta & \displaystyle \forall (\varphi, A) \in \mathcal{E}.
    \end{array}
\end{equation*}
Here, the first line holds since $\{ A : A \in \mathcal{F}\} = \{A + f(D) : A \in \mathcal{F}\}$ for any $D \in {\color{black}\mathcal{D}}$, whereas the second line is due to Assumption~\ref{assumption_varphi}. Replacing the latter representation of the DP constraints with those in~\eqref{problem:integral_main:preliminary_formulation} gives~\ref{problem:integral_main} and thus concludes the observation. \qed
}

\subsection{Proof of Lemma~\ref{prop:finite}}\label{app_proof_beta_discretiez}

The proof of Lemma~\ref{prop:finite} relies on three auxiliary results, which we state and prove first. Lemma~\ref{lemma_beginning} shows that under restriction~\eqref{restriction}, problem~\ref{problem:integral_main} can be expressed entirely in terms of the decision variables $p$. The resulting problem coincides with~\ref{problem:restricted_main_finite} in terms of the decision variables, but it still comprises a larger set of constraints. Lemma~\ref{lemma:worst-fixed} identifies for each query output difference $\varphi \in [-\Delta f, \Delta f]$ a constraint $(\varphi, A^\star (\varphi))$ that weakly dominates all constraints $(\varphi, A)$, $A \in \mathcal{F}$, and Lemma~\ref{lemma:linear} identifies a set of values for $\varphi$ such that the associated constraints $(\varphi, A^\star (\varphi))$ weakly dominate all constraints of the problem. Lemma~\ref{prop:finite} then combines these results to prove the equivalence between the problems~\ref{problem:integral_main} and~\ref{problem:restricted_main_finite} under restriction~\eqref{restriction}.

\begin{lemmaA}\label{lemma_beginning}
    Under restriction~\eqref{restriction},~\ref{problem:integral_main} has the same optimal value as
    \begin{align}\label{problem:restricted_main}
    \begin{array}{cll}
        \underset{p}{\mathrm{minimize}} & \displaystyle \sum\limits_{i \in \mathbb{Z}} c_i(\beta) \cdot p(i) & \\
        \mathrm{subject\; to} & \displaystyle p:\mathbb{Z} \mapsto \mathbb{R}_{+}, \ \sum_{i \in \mathbb{Z}} p(i) = 1 &  \\[1.5em]
        &\displaystyle  \sum\limits_{i \in \mathbb{Z}} p(i) \cdot \dfrac{|A \cap I_{i}(\beta)|}{\beta} \leq e^\varepsilon \cdot \sum\limits_{i \in \mathbb{Z}} p(i)  \cdot \dfrac{|(A - \varphi ) \cap I_{i}(\beta)|}{\beta} + \delta & \forall (\varphi, A) \in \mathcal{E}.
    \end{array}
\end{align}
\end{lemmaA}

\begin{proof}
We use restriction~\eqref{restriction} to replace $\gamma$ in problem~\ref{problem:integral_main} with the new decision variables $p$. To this end, observe first that under restriction~\eqref{restriction}, $\gamma$ affords a density function via
\begin{align*}
    \displaystyle \gamma(A) =  \sum_{i \in \mathbb{Z}} p(i) \cdot \frac{|A\cap I_{i}(\beta)|}{\beta} = & \displaystyle \sum_{i \in \mathbb{Z}} p(i)  \cdot \int_{x \in \mathbb{R}} \frac{\indicate{x \in A} \cdot \indicate{x \in I_{i}(\beta)}}{\beta} \diff x \\
    =  & \int_{x \in \mathbb{R}} \indicate{x \in A} \cdot \left(\displaystyle\sum_{i \in \mathbb{Z}}  p(i) \cdot \dfrac{\indicate{x \in I_{i}(\beta)}}{\beta}\right) \diff x,
\end{align*}
$A \in \mathcal{F}$, where the last step holds by Fubini's theorem since $\gamma(A) \in [0,1]$. This derivation shows that $\sum_{i \in \mathbb{Z}}  p(i) \cdot \indicate{x \in I_{i}(\beta)}/\beta$ is the density function of $\gamma$.
Using this density function, the objective function of~\ref{problem:integral_main} can be rewritten as
\begin{align*}
    \displaystyle \int_{x \in \mathbb{R}} c(x) \diff \gamma(x) = \displaystyle \int_{x \in \mathbb{R}} c(x)\cdot \left(\sum_{i \in \mathbb{Z}}  p(i)\cdot \dfrac{\indicate{x \in I_{i}(\beta)}}{\beta} \right) \diff x & = \sum_{i \in \mathbb{Z}} \frac{p(i)}{\beta} \cdot \int_{x \in \mathbb{R}} c(x)\cdot  \indicate{x \in I_{i}(\beta)} \diff x\\
    & = \sum_{i \in \mathbb{Z}} c_i(\beta) \cdot p(i),
\end{align*}
where the second equality holds by Fubini's theorem, and the last step substitutes $c_i(\beta) = \beta^{-1}\cdot \int_{x \in I_{i}(\beta)} c(x) \mathrm{d}x$ for all $i \in \mathbb{Z}$. The final expression coincides with the objective of problem~\eqref{problem:restricted_main}.

In view of the DP constraints in problem~\ref{problem:integral_main}, we note that
\begin{align*}
    \displaystyle \int_{x \in \mathbb{R}} \indicate{x \in A} \diff \gamma(x)  = \gamma(A) = \sum_{i \in \mathbb{Z}} p(i) \cdot \frac{|A\cap I_{i}(\beta)|}{\beta}
\end{align*}
as well as 
\begin{align*}
    \displaystyle \int_{x \in \mathbb{R}} \indicate{x + \varphi \in A} \diff \gamma(x)  = \gamma(A - \varphi ) = \sum_{i \in \mathbb{Z}} p(i) \cdot \frac{|(A - \varphi)\cap I_{i}(\beta)|}{\beta},
\end{align*}
and the right-hand side expressions coincide with the expressions on either side of the DP constraints of problem~\eqref{problem:restricted_main}. This concludes the proof.
\end{proof}

To reduce the number of DP constraints in problem~\eqref{problem:restricted_main}, we first characterize the tightest privacy constraint $(\varphi, A) \in \mathcal{E}$ in~\eqref{problem:restricted_main} for a fixed decision $p$ and a fixed query output difference $\varphi \in [-\Delta f, \Delta f]$. To this end, we define the privacy shortfall as
\begin{align*}
V(\varphi, A) = \sum_{i \in \mathbb{Z}} p(i) \cdot |A \cap I_{i}(\beta)| - e^\varepsilon\cdot \sum_{i \in \mathbb{Z}} p(i) \cdot |(A - \varphi) \cap I_{i}(\beta)|.
\end{align*}
Note that $V (\varphi, A)$ coincides with the slack of the DP constraint $(\varphi, A)$ in problem~\eqref{problem:restricted_main}, shifted by $- \delta$ and scaled by $\beta$.\footnote{Our definition here differs slightly from a later definition of privacy shortfall in Section~\ref{sec:cutting-plane}. The context will always make it clear which definition is being used.} In particular, maximizers $(\varphi, A) \in \mathcal{E}$ of the privacy shortfall $V$ correspond to the tightest constraints in problem~\eqref{problem:restricted_main}.

\begin{observationA}\label{obs:linearity}
The privacy shortfall is linear over partitions of $A$, that is, we have $V(\varphi, A) = \sum_\ell V(\varphi, A_\ell)$ for any partition $\{A_\ell\}_\ell$ of $A$. 
\end{observationA}

We next show that for any fixed $\varphi \in [-\Delta f, \Delta f]$, the largest privacy shortfall $\sup \{V(\varphi, A) : A \in \mathcal{F}\}$ is attained by a worst-case event $A^\star(\varphi) \in \mathcal{F}$ of a simple structure.

\begin{definition}\label{definition:partition}
For any $i \in \mathbb{Z}$, let $j := \lceil i - \varphi/\beta \rceil$ be the unique integer satisfying $\varphi + j \cdot \beta \in I_i(\beta)$, and define $I_i^1(\varphi, \beta)$ and $I_i^2(\varphi, \beta)$ as the following partition of $I_i(\beta)$:
\begin{enumerate}[(i)]
    \item $I_i^1(\varphi, \beta) := I_i(\beta) \cap (I_{j-1}(\beta) + \varphi) = [i\cdot \beta, \varphi + j \cdot \beta)$
    \item $I_i^2(\varphi, \beta) := I_i(\beta) \setminus I_i^1(\varphi, \beta) = I_i(\beta) \cap (I_{j}(\beta) + \varphi) = [\varphi + j\cdot \beta, (i+1)\cdot \beta)$. 
\end{enumerate}
\end{definition}
For any $i \in \mathbb{Z}$, Definition~\ref{definition:partition} implies that $I_i^1(\varphi,\beta) \cap I_{i'}(\beta) = \emptyset$ for all $i' \neq i$. Moreover, since $I_i^1(\varphi,\beta) \subseteq I_{j-1}(\beta) + \varphi$ it also follows that $(I_i^1(\varphi,\beta)-\varphi) \subseteq I_{j-1}(\beta)$ and therefore $(I_i^1(\varphi,\beta)-\varphi) \cap I_{j'}(\beta) = \emptyset$ for all $j' \neq j - 1$. Applying a similar reasoning also to $I_i^2(\varphi,\beta)$ allows us to simplify the expressions for the privacy shortfall over subsets of $I_i^1(\varphi,\beta)$ and $I_i^2(\varphi,\beta)$.

\begin{observationA}\label{obs:discretized_interval}
For any $i \in \mathbb{Z}$ and $\varphi \in [- \Delta f, \Delta f]$ we have
\begin{equation*}
    V(\varphi, A) =
    \begin{cases}
        |A| \cdot (p(i) - e^\varepsilon \cdot p(j-1)) & \text{for all } A \subseteq I_i^1(\varphi,\beta), \\
        |A| \cdot(p(i) - e^\varepsilon \cdot p(j)) & \text{for all } A \subseteq I_i^2(\varphi,\beta),
    \end{cases}
\end{equation*}
where $j = \lceil i - \varphi / \beta \rceil$.
\end{observationA}

\begin{figure}[tb]
    \centering
    \begin{tikzpicture}[scale=0.8, every node/.style={scale=0.8},every node/.append style={thick,rounded corners =1mm}]
    \draw[-, very thick, color = black] (2,2) -- (19,2);
    \node[xshift=3cm,yshift = -0.2cm,draw,fill=gray!25,text width=5cm,align=center] 
    {
    $j = \lceil i - \varphi /\beta \rceil$ satisfies $i \cdot \beta \leq \varphi + j\cdot \beta < (i+1)\cdot \beta$};
    \node[scale = 0.5] at (0.8,2){\textbullet};
    \node[scale = 0.5] at (1.2,2){\textbullet};
    \node[scale = 0.5] at (1.6,2){\textbullet};
    \node[scale = 0.5] at (19.4,2){\textbullet};
    \node[scale = 0.5] at (19.8,2){\textbullet};
    \node[scale = 0.5] at (20.2,2){\textbullet};
    \draw[black, thick] (3,2.2) -- (3,1.8);
    \draw[black, thick] (9,2.2) -- (9,1.8);
    \draw[black, thick] (15,2.2) -- (15,1.8);
    \draw[darkred, thick] (6,2.2) -- (6,1.8);
    \draw[darkred, thick] (12,2.2) -- (12,1.8);
    \draw[darkred, thick] (18,2.2) -- (18,1.8);
    \node[scale = 1] at (3,2.5){$(i-1)\cdot \beta$};
    \node[scale = 1] at (9,2.5){$i\cdot \beta$};
    \node[scale = 1] at (15,2.5){$(i+1)\cdot \beta$};
    \node[darkred, scale = 1] at (6,1.5){$\varphi + (j-1)\cdot \beta$};
    \node[darkred, scale = 1] at (12,1.5){$\varphi + j\cdot \beta$};
    \node[darkred, scale = 1] at (18,1.5){$\varphi + (j+1)\cdot \beta$};
    \draw [thick, decorate,decoration={brace,amplitude=10pt},xshift=0pt,yshift=0pt] (3,3) -- (8.92,3) node [black,midway, yshift = 0.8cm] {$I_{i-1}(\beta)$};
    \draw [thick, decorate,decoration={brace,amplitude=10pt},xshift=0pt,yshift=0pt] (9,3) -- (15,3) node [black,midway, yshift = 0.8cm] {$I_{i}(\beta)$};
    \draw [thick, darkred,decorate,decoration={brace,amplitude=10pt,mirror},xshift=0pt,yshift=0pt] (6,0) -- (11.92,0) node [darkred,midway, yshift = -0.8cm] {$I_{j-1}(\beta) + \varphi$};
    \draw [thick, darkred,decorate,decoration={brace,amplitude=10pt,mirror},xshift=0pt,yshift=0pt] (12,0) -- (18,0) node [darkred,midway, yshift = -0.8cm] {$I_{j}(\beta) + \varphi$};
    \draw [blue,decorate,decoration={brace,amplitude=10pt,mirror},xshift=0pt,yshift=0pt] (9,1) -- (11.92,1) node [blue,midway, yshift = -0.8cm] {$I_i^1(\varphi,\beta)$};
    \draw [blue,decorate,decoration={brace,amplitude=10pt,mirror},xshift=0pt,yshift=0pt] (12,1) -- (15,1) node [blue,midway, yshift = -0.8cm] {$I_i^2(\varphi,\beta)$};
    \draw[fill=blue!25] (10,1.8) -- (11,1.8) -- (11,2.2) -- (10,2.2) -- (10,1.8);
    \node[blue, scale = 1] at (10.5,2){$A$};
    \end{tikzpicture}
    \caption{\textit{The interval $I_i^1(\varphi,\beta)$ satisfies $I_i^1(\varphi,\beta) \subseteq I_i(\beta)$ as well as $I_i^1(\varphi,\beta) \subseteq I_{j-1}(\beta) + \varphi$. Therefore, for any $A \subseteq I_i^1(\varphi,\beta)$ we have $|A\cap I_i(\beta)| = |A|$ and $|A\cap I_{i'}(\beta)| = 0$ for all $i' \neq i$; similarly, $|(A-\varphi) \cap I_{j-1}(\beta)| = |A|$ and $|(A-\varphi) \cap I_{j'-1}(\beta)| = 0$ for all $j' \neq j$. This shows the first case in Observation~\ref{obs:discretized_interval}; the second case can be verified analogously.}}
    \label{figure:value}
\end{figure}

Figure~\ref{figure:value} illustrates the intuition underlying Observation~\ref{obs:discretized_interval}. We now show that there is always a worst-case event $A^\star(\varphi)$ \mbox{that constitutes a union of intervals $I_i^1(\varphi,\beta)$ and $I_i^2(\varphi,\beta)$, $i \in \mathbb{Z}$.}

\begin{lemmaA}\label{lemma:worst-fixed}
For any $\varphi \in [-\Delta f, \Delta f]$, there is an event
\begin{align}\label{set:worst_A}
A^\star(\varphi) = \bigcup_{i \in \mathcal{I}_1} I_i^1(\varphi,\beta) \cup \bigcup_{i \in \mathcal{I}_2} I_i^2(\varphi,\beta) \in \mathcal{F}
\qquad \text{for some } \mathcal{I}_1, \mathcal{I}_2 \subseteq \mathbb{Z}
\end{align}
that attains the largest privacy shortfall $\sup \{V(\varphi, A) : A \in \mathcal{F}\}$.
\end{lemmaA}

\begin{proof}
By Definition~\ref{definition:partition}, $I_i^1(\varphi,\beta)$ and $I_i^2(\varphi,\beta)$ partition $I_i(\beta)$ for any $i \in \mathbb{Z}$, and thus $\{I_i^1(\varphi,\beta) \cup  I_i^2(\varphi,\beta)\}_{i \in \mathbb{Z}}$ partitions $\mathbb{R}$. Observation~\ref{obs:linearity} and the sub-additivity of the supremum operator then imply that
\begin{align*}
    \displaystyle \underset{A \in \mathcal{F}}{\sup} \{ V(\varphi, A)\} & =  \underset{A \in \mathcal{F}}{\sup} \left\{ \sum_{i \in \mathbb{Z}} V(\varphi, A\cap I_i^1(\varphi,\beta)) + \sum_{i \in \mathbb{Z}} V(\varphi, A\cap I_i^2(\varphi,\beta))  \right\} \\ 
    & \leq\sum_{i \in \mathbb{Z}} \underset{A \subseteq I_i^1(\varphi,\beta)}{\sup} \{ V(\varphi,A) \} + \sum_{i \in \mathbb{Z}} \underset{A \subseteq I_i^2(\varphi,\beta)}{\sup} \{ V(\varphi,A) \}.
\end{align*}
We will show that each supremum on the right-hand side of the inequality is attained and then construct $A^\star(\varphi) = \bigcup_{i \in \mathbb{Z}} A^1_i (\varphi) \cup \bigcup_{i \in \mathbb{Z}} A^2_i (\varphi)$, where $A_i^1(\varphi), A_i^2(\varphi) \in \mathcal{F}$ are defined as
\begin{align*}
    & A^1_i (\varphi) \in \underset{A \subseteq I_i^1(\varphi,\beta)}{\arg \max} \{ V(\varphi,A) \}
    \quad \text{and} \quad
    A^2_i (\varphi) \in \underset{A \subseteq I_i^2(\varphi,\beta)}{\arg \max} \{ V(\varphi,A) \}.
\end{align*}
The statement then follows from the fact that $A^\star(\varphi) \in \mathcal{F}$ by construction.

In view of $A^1_i (\varphi)$, Observation~\ref{obs:discretized_interval} implies that $V (\varphi, A)$ is maximized by
\begin{align*}
    A_i^1(\varphi) = \begin{cases}
         I_i^1(\varphi,\beta) & \text{ if } p(i) - e^\varepsilon \cdot p(j-1) > 0, \\
         \emptyset & \text{ otherwise.}
    \end{cases}
\end{align*}
Applying a similar reasoning to $A^2_i (\varphi)$, we observe that
\begin{align*}
    A_i^2(\varphi) = \begin{cases}
         I_i^2(\varphi,\beta) & \text{ if } p(i) - e^\varepsilon \cdot p(j) > 0, \\
         \emptyset & \text{ otherwise}.
    \end{cases}
\end{align*}
The statement of the lemma thus follows.
\end{proof}

The next result shows that $V(\varphi, A^\star(\varphi))$ is maximized by $\varphi^\star = k \cdot \beta$ for some $k \in \mathbb{Z}$.

\begin{lemmaA}\label{lemma:linear}
The function $\varphi \mapsto V(\varphi, A^\star(\varphi))$ is affine over each interval $[k \cdot \beta, (k+1)\cdot \beta]$, $k \in \mathbb{Z}$.
\end{lemmaA}

\begin{proof}
For any $k \in \mathbb{Z}$, the construction of $A^\star(\varphi)$ in the proof of Lemma~\ref{lemma:worst-fixed} implies that the sets $\mathcal{I}_1$ and $\mathcal{I}_2$ in the statement of the lemma coincide for all $\varphi \in [k \cdot \beta, (k+1)\cdot \beta) = I_k(\beta)$. Therefore, for all $\varphi \in I_k(\beta)$ we have
\begin{align}\label{eq_alpha}
\mspace{-20mu}
    V(\varphi, A^\star(\varphi)) &= \sum_{i \in \mathcal{I}_1} V(\varphi, I_i^1(\varphi,\beta)) + \sum_{i \in \mathcal{I}_2} V(\varphi, I_i^2(\varphi,\beta)) \notag \\
    & = \sum_{i \in \mathcal{I}_1} |I_i^1(\varphi,\beta)|\cdot (p(i) - e^\varepsilon \cdot p(j-1)) + \sum_{i \in \mathcal{I}_2} |I_i^2(\varphi,\beta)|\cdot (p(i) - e^\varepsilon \cdot p(j)) \notag\\
    & =  (\varphi \ \mathrm{mod} \ \beta) \cdot \sum_{i \in \mathcal{I}_1} (p(i) - e^\varepsilon \cdot p(j-1)) + (\beta - (\varphi \ \mathrm{mod} \ \beta)) \cdot  \sum_{i \in \mathcal{I}_2} (p(i) - e^\varepsilon \cdot p(j)) \notag \\
    & = (\varphi \ \mathrm{mod} \ \beta) \cdot \left[\sum_{i \in \mathcal{I}_1} (p(i) - e^\varepsilon \cdot p(j-1)) - \sum_{i \in \mathcal{I}_2} (p(i) - e^\varepsilon \cdot p(j))  \right] + \beta \cdot \sum_{i \in \mathcal{I}_2} (p(i) - e^\varepsilon \cdot p(j)),
\end{align}
where $j = \lceil i - \varphi / \beta \rceil$ as specified by Definition~\ref{definition:partition}. Here, the first equality follows from Lemma~\ref{lemma:worst-fixed} and Observation~\ref{obs:linearity}, the second equality is due to Observation~\ref{obs:discretized_interval}, the third equality holds since
\begin{equation*}
    |I_i^1(\varphi, \beta)| = \varphi + \lceil i - \varphi / \beta \rceil \cdot\beta - i\cdot \beta = \varphi - \beta \cdot \lfloor \varphi / \beta \rfloor = (\varphi \ \mathrm{mod} \ \beta)
\end{equation*}
and $|I_i^2(\varphi, \beta)| = \beta - |I_i^1(\varphi, \beta)| = \beta -  (\varphi \ \mathrm{mod} \ \beta)$. In the final expression, all terms except for $(\varphi \ \mathrm{mod} \ \beta)$ are independent of $\varphi$, and $\varphi \mapsto (\varphi \ \mathrm{mod}  \ \beta)$ is affine over $\varphi \in I_k(\beta)$.
We thus conclude that $\varphi \mapsto V(\varphi, A^\star(\varphi))$ is affine over $\varphi \in I_k(\beta)$. 

To conclude the proof, we show that the result holds for the closure of $I_k(\beta)$, that is, $\varphi \mapsto V(\varphi, A^\star(\varphi))$ is not discontinuous at $\overline{\varphi} = (k+1)\cdot \beta$. In other words, we show that
\begin{align*}
    \underset{\varphi \rightarrow \overline{\varphi}}{\lim} \  V(\varphi, A^\star(\varphi))=  V(\overline{\varphi}, A^\star(\overline{\varphi})).
\end{align*}
To this end, we first note that $(\overline{\varphi} \ \mathrm{mod} \ \beta) = 0$ as well as $j = \lceil i - \overline{\varphi}/\beta \rceil = i - k - 1$, and hence \eqref{eq_alpha} implies that
\begin{align*}
     V(\overline{\varphi}, A^\star(\overline{\varphi})) = \beta \cdot \sum_{i \in \mathcal{I}_2} (p(i) - e^\varepsilon \cdot p(j)) & = \beta \cdot \sum_{i \in \mathcal{I}_2} (p(i) - e^\varepsilon \cdot p(i - k - 1)) \\
     & = \beta \cdot \sum_{i \in \mathbb{Z}} \max\{p(i) - e^\varepsilon \cdot p(i - k - 1), 0\}.
\end{align*}
Since we also have $j = \lceil i - \varphi / \beta \rceil = i - k$ for all $k \cdot \beta \leq \varphi < (k+1) \cdot \beta$, it follows that
\begin{align*}
    \mspace{-60mu}
    \underset{\varphi \rightarrow \overline{\varphi}}{\lim} V(\varphi, A^\star(\varphi))
    = & \underset{\varphi \rightarrow \overline{\varphi}}{\lim} \Big( (\varphi \ \mathrm{mod} \ \beta) \cdot \Big[\sum_{i \in \mathcal{I}_1} (p(i) - e^\varepsilon \cdot p(j-1)) - \sum_{i \in \mathcal{I}_2} (p(i) - e^\varepsilon \cdot p(j))  \Big] + \beta \cdot \sum_{i \in \mathcal{I}_2} (p(i) - e^\varepsilon \cdot p(j)) \Big) \\
    = & \Big[\sum_{i \in \mathcal{I}_1} (p(i) - e^\varepsilon \cdot p(j-1)) - \sum_{i \in \mathcal{I}_2} (p(i) - e^\varepsilon \cdot p(j))  \Big] \cdot \underset{\varphi \rightarrow \overline{\varphi}}{\lim} \{ (\varphi \ \mathrm{mod} \ \beta)  \}+ \beta \cdot \sum_{i  \in \mathcal{I}_2} (p(i) - e^\varepsilon \cdot p(j))\\
    = & \beta \cdot \sum_{i \in \mathcal{I}_1} (p(i) - e^\varepsilon \cdot p(j-1)) = \beta \cdot \sum_{i \in \mathcal{I}_1} (p(i) - e^\varepsilon \cdot p(i-k-1)) \\
    = & \beta \cdot \sum_{i \in \mathbb{Z}} \max\{p(i) - e^\varepsilon \cdot p(i - k - 1), 0\} = V(\overline{\varphi}, A^\star(\overline{\varphi})),
\end{align*}
where the first equality follows from~\eqref{eq_alpha}, the second equality exploits the linearity of limits, the third equality holds since $\underset{\varphi \rightarrow \overline{\varphi}}{\lim} (\varphi \ \mathrm{mod} \ \beta) = \beta$, and the final equalities follow from substituting $j = i-k$ and using the construction of $\mathcal{I}_1$, which includes all indices $i\in \mathbb{Z}$ for which the incremental privacy shortfall $p(i) - e^\varepsilon \cdot p(i - k - 1)$ is positive. This shows that $\varphi \mapsto V(\varphi, A^\star(\varphi))$ is not discontinuous at $\overline{\varphi} = (k+1)\cdot \beta$ and therefore concludes the proof.
\end{proof}

We can now prove Lemma~\ref{prop:finite} by showing that problem~\eqref{problem:restricted_main} in the statement of Lemma~\ref{lemma_beginning} has the same optimal value as problem~\ref{problem:restricted_main_finite}. \\[-2mm]

\noindent \textbf{Proof of Lemma~\ref{prop:finite}.} $\;$
First notice that the DP constraints of~\ref{problem:restricted_main_finite} can be written as
\begin{align*}
    \displaystyle \sum\limits_{i \in \mathbb{Z}} p(i)\cdot \dfrac{|A \cap I_i(\beta)|}{\beta} \leq e^\varepsilon \cdot \sum\limits_{i \in \mathbb{Z}} p(i) \cdot \dfrac{|(A - \varphi ) \cap I_i(\beta)|}{\beta} + \delta \quad \forall (\varphi, A) \in \mathcal{E}(\beta)
\end{align*}
since for any $(\varphi, A) \in \mathcal{E}(\beta)$ we have $|A \cap I_i(\beta)| / \beta = \indicate{I_i(\beta) \subseteq A}$ as well as $|(A-\varphi) \cap I_i(\beta)| / \beta = \indicate{I_i(\beta) + \varphi \subseteq A}$ by definition. This shows that~\ref{problem:restricted_main_finite} is a relaxation of problem~\eqref{problem:restricted_main} since $\mathcal{E}(\beta) \subset \mathcal{E}$. Hence, if~\ref{problem:restricted_main_finite} is infeasible, then so is problem~\eqref{problem:restricted_main}, and the result follows. To complete the proof, we show that any $p$ feasible in~\ref{problem:restricted_main_finite} is also feasible in problem~\eqref{problem:restricted_main}.

Fix any feasible solution $p$ to~\ref{problem:restricted_main_finite}, and assume to the contrary that $p$ violates a DP constraint $(\varphi, A) \in \mathcal{E}$ in problem~\eqref{problem:restricted_main}. In that case, Lemmas~\ref{lemma:worst-fixed} and~\ref{lemma:linear} imply that there is a constraint $(\varphi^\star, A^\star (\varphi^\star)) \in \mathcal{E} (\beta)$ with a weakly higher privacy shortfall than $(\varphi, A)$. Indeed, Lemma~\ref{lemma:linear} and our earlier assumption that $\Delta f$ is divisible by $\beta$ imply that $\varphi^\star$ can without loss of generality be chosen such that $\varphi^\star \in [-\Delta f, \Delta f] \cap \{k\cdot\beta\}_{k \in \mathbb{Z}} =  \setvarphi(\beta)$. Since such $\varphi^\star$ is a multiple of $\beta$, we have $I_i^1(\varphi,\beta) = \emptyset$ and $I_i^2(\varphi,\beta) = I_i(\beta)$ for all $i \in \mathbb{Z}$. Hence, Lemma~\ref{lemma:worst-fixed} implies that $A^\star(\varphi^\star)$ can be chosen such that $A^\star(\varphi^\star) = \bigcup_{i \in \mathcal{I}_2} I_i(\beta)$ for some $\mathcal{I}_2 \subseteq \mathbb{Z}$, which in turn implies that $A^\star(\varphi) \in \mathcal{F}(\beta)$. Thus, there must be a violated DP constraint $(\varphi^\star, A^\star (\varphi^\star))$ such that $\varphi^\star \in \setvarphi(\beta)$ and $A^\star(\varphi) \in \mathcal{F}(\beta)$, that is, $(\varphi^\star, A^\star (\varphi^\star)) \in \mathcal{E} (\beta)$. This contradicts our earlier assumption that $p$ is feasible in~\ref{problem:restricted_main_finite}, and thus the result follows.
\qed

\subsection{Proof of Proposition~\ref{prop:bounded}}

Appending the additional constraint~\eqref{eq:bound_sup} to problem~\ref{problem:restricted_main_finite} yields
\begin{align*}
    \begin{array}{cll}
        \underset{p}{\mathrm{minimize}} & \displaystyle \sum\limits_{i \in [\pm L]} c_i(\beta) \cdot p(i) & \\[1.5em]
        \mathrm{subject\; to} &\displaystyle p: [\pm L] \mapsto \mathbb{R}_{+}, \ \sum_{i \in [\pm L]} p(i) = 1 & \\[1.5em]
        &\displaystyle \sum\limits_{i \in [\pm L]} \indicate{I_i(\beta) \subseteq A} \cdot p(i) \leq e^\varepsilon \cdot \sum\limits_{i \in [\pm L]} \indicate{I_i(\beta) + \varphi \subseteq A} \cdot p(i) + \delta & \forall (\varphi, A) \in \mathcal{E}(\beta).\\
    \end{array}
\end{align*}
The result follows if we show that any DP constraint $(\varphi, A) \in \mathcal{E}(\beta)\setminus\mathcal{E}(L, \beta)$ is redundant in the above problem. To this end, fix any $p$ that satisfies the first constraint. For any $(\varphi, A) \in \mathcal{E}(\beta)$, define $A_L := A \cap [-L\cdot \beta, (L+1)\cdot \beta)$ so that $(\varphi, A_L) \in \mathcal{E}(L, \beta)$. We show that if $p$ satisfies the DP constraint $(\varphi, A_L)$, then it also satisfies the DP constraint $(\varphi, A)$. Indeed, we observe that
\begin{align*}
    & \sum_{i \in [\pm L]} \indicate{I_i(\beta)\subseteq A } \cdot p(i)  -  e^\varepsilon \cdot \sum_{i \in [\pm L]}  \indicate{(I_i(\beta)+ \varphi) \subseteq A}\cdot p(i) \\
    = & \sum_{i \in [\pm L]} \indicate{I_i(\beta)\subseteq A_L } \cdot p(i)  -  e^\varepsilon \cdot \sum_{i \in [\pm L]}  \indicate{(I_i(\beta)+ \varphi) \subseteq A} \cdot p(i) \\
    \leq & \sum_{i \in [\pm L]} \indicate{I_i(\beta)\subseteq A_L } \cdot p(i)  -  e^\varepsilon \cdot \sum_{i \in [\pm L]}  \indicate{(I_i(\beta)+ \varphi) \subseteq A_L} \cdot p(i),
\end{align*}
where the equality follows from
$$\sum_{i \in [\pm L]} \indicate{I_i(\beta)\subseteq A } \cdot p(i) = \sum_{i \in [\pm L]} \indicate{I_i(\beta)\subseteq A_L} \cdot p(i)  + \sum_{i \in [\pm L]} \underbrace{\indicate{I_i(\beta)\subseteq A \setminus A_L } \cdot p(i) }_{=0}$$
which holds since no $i \in [\pm L]$ can satisfy $I_{i}(\beta) \subseteq A \setminus A_L$. The inequality in the third row, on the other hand, follows from $A_L \subseteq A$. Thus, the DP constraints $\mathcal{E}(\beta)\setminus \mathcal{E}(L, \beta)$ are redundant since they are weakly dominated by the DP constraints $(\varphi, A_L) \in \mathcal{E}(L, \beta)$.
\qed

\subsection{Proof of Proposition~\ref{prop:weakdual}}

As $(\theta,\psi) \in  \mathbb{R}\times \mathcal{M}_{+}(\mathcal{E})$ is feasible in \ref{problem:integral_dual} and $\delta > 0$, we have that $\int_{(\varphi, A) \in \mathcal{E}} \diff \psi(\varphi, A) < \infty$ from the objective function of \ref{problem:integral_dual}, which shows that $\mathcal{E}$ is $\sigma$-finite with measure $\psi$. Moreover, $\gamma$ is a probability measure on $\mathbb{R}$, and hence it is also $\sigma$-finite. We now observe that
\begin{align*}
    & \displaystyle \int_{x \in \mathbb{R}} c(x) \diff \gamma(x) \\
    \geq & \int_{x \in \mathbb{R}} \left[ \theta - \int_{(\varphi, A)\in\mathcal{E}} \indicate{x \in A} \diff \psi(\varphi, A) + e^\varepsilon \cdot \int_{(\varphi, A) \in \mathcal{E}} \indicate{x + \varphi \in A}\diff \psi(\varphi, A) \right] \diff \gamma(x)\\
    = & \theta -\int_{(\varphi, A)\in\mathcal{E}} \left[ \int_{x \in \mathbb{R}} \indicate{x \in A}\diff \gamma(x) - e^\varepsilon \cdot \int_{x \in \mathbb{R}} \indicate{x+\varphi \in A} \diff \gamma(x)\right]\diff \psi(\varphi, A)  \\
    \geq & \theta- \int_{(\varphi, A) \in \mathcal{E}} \delta \diff \psi(\varphi, A) .
\end{align*}
Here, the first inequality follows from the constraints of problem~\ref{problem:integral_dual}. The equality follows from Fubini's theorem, which is applicable since the indicator functions are integrable on $\mathbb{R}\times \mathcal{E}$ with the associated product measure and the fact that $\int_{x \in \mathbb{R}} \diff \gamma(x) = 1$. The second inequality follows from the constraints of problem~\ref{problem:integral_main} as well as the fact that $\psi$ is a non-negative measure.
\qed

\subsection{Proof of Lemma~\ref{prop:lb-beta}}

We first show that under the additional constraint~\eqref{constraints:extra}, the DP constraints in~\ref{problem:integral_dual} reduce to
\begin{align}
    \theta \leq c(x) + \int_{(\varphi, A) \in \mathcal{E}(\beta)} \indicate{I_i(\beta) \subseteq A}\mathrm{d}\psi(\varphi, A) - e^\varepsilon \cdot \int_{(\varphi, A) \in \mathcal{E}(\beta)}\indicate{I_i(\beta) + \varphi \subseteq A} \mathrm{d}\psi(\varphi, A)\nonumber \\
    \pushright{\forall i \in \mathbb{Z}, \ \forall x \in I_i(\beta).} \label{eq:the_secret_life_of_dp_constraints}
\end{align}
We then argue that for every $i \in \mathbb{Z}$, all constraints~\eqref{eq:the_secret_life_of_dp_constraints} indexed by $(i, x)$, $x \in I_i(\beta)$, are simultaneously satisfied if and only if the DP constraint indexed by $i$ is satisfied in~\ref{LB-beta}. The result then follows since both the decision variables and the objective function in~\ref{problem:integral_dual} coincide with their counterparts in~\ref{LB-beta}, restricted to the elements $(\varphi, A) \in \mathcal{E} (\beta)$ as stipulated by~\eqref{constraints:extra}.

In view of the first step, fix any $x \in \mathbb{R}$, and select $i \in \mathbb{Z}$ such that $x \in I_i(\beta)$. Under the additional constraint~\eqref{constraints:extra}, the first integral in the DP constraint of~\ref{problem:integral_dual} indexed by $x$ reduces to
\begin{align*}
    \displaystyle \int_{(\varphi, A) \in \mathcal{E}} \indicate{x \in A}\mathrm{d}\psi(\varphi, A) &= \displaystyle \int_{(\varphi, A) \in \mathcal{E}(\beta)} \indicate{x \in A}\mathrm{d}\psi(\varphi, A)  + \underbrace{\int_{(\varphi, A) \in \mathcal{E} \setminus \mathcal{E}(\beta)} \indicate{x \in A}\mathrm{d}\psi(\varphi, A) }_{= 0} \\
    & =\displaystyle \int_{(\varphi, A) \in \mathcal{E}(\beta)} \indicate{I_i(\beta) \subseteq A}\mathrm{d}\psi(\varphi, A).
\end{align*}
Here, the last integral in the first row vanishes due to~\eqref{constraints:extra}, whereas the second equality holds since for all $(\varphi, A) \in \mathcal{E}(\beta)$, the requirement that $A \in \mathcal{F} (\beta)$ implies that $x \in A$ only if $I_i (\beta) \subseteq A$. Note that the integral in the second row above coincides with the first integral in~\eqref{eq:the_secret_life_of_dp_constraints} indexed by $(i, x)$. A similar argument shows that under the additional constraint~\eqref{constraints:extra}, the second integral in the DP constraint of~\ref{problem:integral_dual} indexed by $x$ reduces to the second integral in~\eqref{eq:the_secret_life_of_dp_constraints} indexed by $(i, x)$. In summary, under the additional constraint~\eqref{constraints:extra} the DP constraints in~\ref{problem:integral_dual} indeed reduce to~\eqref{eq:the_secret_life_of_dp_constraints}.

As for the second step, note that for any fixed $i \in \mathbb{Z}$, the constraints~\eqref{eq:the_secret_life_of_dp_constraints} indexed by $(i, x)$, $x \in I_i(\beta)$, only differ in their additive terms $c (x)$. Thus, for any fixed $i \in \mathbb{Z}$, all constraints~\eqref{eq:the_secret_life_of_dp_constraints} indexed by $(i, x)$, $x \in I_i(\beta)$, are satisfied if and only if they are satisfied for the smallest value $c(x)$, $x \in I_i(\beta)$, which is precisely what the DP constraint in~\ref{LB-beta} indexed by $i$ stipulates.
\qed

\subsection{Proof of Proposition~\ref{prop:lb-L-beta}}

First observe that the additional constraint~\eqref{constraints:extratwo} reduces the uncountable set $\mathcal{E} (\beta)$ in the definition of the decision variables as well as the objective function and the constraints of~\ref{LB-beta} to the finite subset $\mathcal{E} (L, \beta)$, which allows us to replace the measure $\psi \in \mathcal{M}_+ (\mathcal{E} (\beta))$ in~\ref{LB-beta} with the discrete map $\psi: \mathcal{E}(L, \beta) \mapsto \mathbb{R}_{+}$ in~\ref{lb-L-beta} as well as replace all integrals in~\ref{LB-beta} with sums in~\ref{lb-L-beta}.

The result now follows if we show that under the additional constraint~\eqref{constraints:extratwo}, all DP constraints in~\ref{LB-beta} indexed by $i \in \mathbb{Z} \setminus [\pm (L + \Delta f / \beta)]$ are weakly dominated by DP constraints indexed by $i \in [\pm (L + \Delta f / \beta)]$. Indeed, observe that the DP constraints indexed by $i \in \mathbb{Z} \setminus [\pm (L + \Delta f / \beta)]$ simplify to
\begin{subequations}
\begin{gather}
\theta \leq \underline{c}_i(\beta) \quad \forall i \in \mathbb{Z} \setminus [\pm (L + \Delta f / \beta )] \label{eliminated}\\
\intertext{since $\indicate{I_i(\beta) \subseteq A} = \indicate{I_i(\beta) + \varphi \subseteq A} = 0$ for all $(\varphi, A) \in \mathcal{E}(L, \beta)$ whenever $i \in \mathbb{Z} \setminus [\pm (L + \Delta f / \beta )]$. In contrast, the constraints indexed by $i \in [\pm (L + \Delta f / \beta)] \setminus [\pm L]$ simplify to}
\theta \leq - e^\varepsilon \cdot \sum_{(\varphi, A) \in \mathcal{E}(L, \beta)}\indicate{I_i(\beta) + \varphi \subseteq A} \cdot \psi(\varphi, A) + \underline{c}_i(\beta)  \label{eliminates}
\end{gather}
\end{subequations}
since $\indicate{I_i(\beta) \subseteq A} = 0$ for all $(\varphi, A) \in \mathcal{E}(L, \beta)$ whenever $i \in  [\pm (L + \Delta f / \beta)] \setminus [\pm L]$. Note that $\underline{c}_i(\beta)$ inherits monotonicity from $c_i(\beta)$, that is, we have $\underline{c}_i(\beta) \leq \underline{c}_{i + 1}(\beta)$ for all $i \geq 0$ as well as $\underline{c}_i(\beta) \leq \underline{c}_{i - 1}(\beta)$ for all $i \leq 0$. This property, along with the non-negativity of $\psi$, shows that the constraints~\eqref{eliminated} are implied by constraints~\eqref{eliminates}, and the result thus follows.
\qed

\subsection{Proof of Theorem~\ref{thm:strong_duality}}

The proof of Theorem~\ref{thm:strong_duality} relies on the feasibility and monotonicity of the upper bounding problems~\ref{L-UB}, which they inherit from the upper bounding problems~\ref{problem:restricted_main_finite}. We prove these results first in Sections~\ref{subsection:1} and~\ref{subsection:2}, and we subsequently prove Theorem~\ref{thm:strong_duality} in Section~\ref{subsection:3}.

\subsubsection{Monotonicity and Feasibility of~\ref{problem:restricted_main_finite}}\label{subsection:1}

The upper bound~\ref{problem:restricted_main_finite} employs a discretization that is parametrized by $\beta$. We first show that the optimal value of this problem is monotonically non-decreasing in $\beta$ in the following sense.

\begin{lemmaA}\label{lemma:monotonic}
For any $\varepsilon > 0$, $\delta > 0$ and $\beta > 0$, the optimal value of problem~\ref{problem:restricted_main_finite} satisfies $\text{\ref{problem:restricted_main_finite}} \geq \text{\hyperref[{problem:restricted_main_finite}]{$\mathrm{P}(\beta / k)$}}$ for all $k \in \mathbb{N}$.
\end{lemmaA}

\begin{proof}
The result trivially holds if~\ref{problem:restricted_main_finite} is infeasible. Assume therefore that~\ref{problem:restricted_main_finite} is feasible and fix an arbitrary feasible solution $p'$ in~\ref{problem:restricted_main_finite}. For any $k \in \mathbb{N}$, problem~\hyperref[{problem:restricted_main_finite}]{$\mathrm{P}(\beta / k)$} can be written as
\begin{align}\label{problem:finer_beta}\tag{$\mathrm{P}(\beta / k)$}
\begin{array}{cl}
    \underset{p}{\mathrm{minimize}} & \displaystyle \sum\limits_{i \in \mathbb{Z}} \sum\limits_{l \in [k]} c_{il}(\beta) \cdot p(i,l)  \\
    \mathrm{subject\; to} &\displaystyle p: \mathbb{Z} \times [k] \mapsto \mathbb{R}_{+}, \ \sum\limits_{i \in \mathbb{Z}} \sum\limits_{l \in [k]} p(i,l) = 1 \\
    &\displaystyle \sum\limits_{i \in \mathbb{Z}} \sum\limits_{l \in [k]} \indicate{I_{il}(\beta) \subseteq A} \cdot p(i,l) \leq e^\varepsilon \cdot \sum\limits_{i \in \mathbb{Z}} \sum\limits_{l \in [k]} \indicate{I_{il}(\beta) + \varphi \subseteq A} \cdot p(i,l) + \delta \\
    & \pushright{\forall (\varphi, A) \in \mathcal{E}(\beta / k),}
    \end{array}
\end{align}
where $I_{il}(\beta) := [(i + (l-1) / k) \cdot \beta, (i + l / k) \cdot \beta)$ and $c_{il}(\beta) := (\beta / k)^{-1} \cdot \int_{x \in I_{il}(\beta)} c(x) \diff x$. One readily verifies that $p'' (i, l) = p'(i) / k$, $i \in \mathbb{Z}$ and $l \in [k]$, is feasible in problem~\ref{problem:finer_beta} and attains the same objective value as $p'$ in~\ref{problem:restricted_main_finite}. The statement thus follows.
\end{proof}

We next show that problem~\ref{problem:restricted_main_finite} is feasible. 
{\color{black}
\begin{lemmaA}\label{lemma:staircase}
For any $\delta > 0$, there is $M \in \mathbb{R}$ such that $\text{\text{\hyperref[{problem:restricted_main_finite}]{$\mathrm{P}(\Delta f / k)$}}} \leq M$ for all $\varepsilon > 0$ and $k \in \mathbb{N}$.
\end{lemmaA}
}
{\color{black}
\begin{proof}
    Fix $\delta > 0$ and denote by~\hyperref[{problem:restricted_main_finite}]{$\mathrm{P}^{0}(\Delta f)$} the variant of \hyperref[{problem:restricted_main_finite}]{$\mathrm{P}(\Delta f)$} that replaces $\varepsilon$ with $0$. We show that there exists $M \in \mathbb{R}$ such that $\text{\hyperref[{problem:restricted_main_finite}]{$\mathrm{P}^{0}(\Delta f)$}} \leq M$. The statement then follows since for any $\varepsilon > 0$ and $k \in \mathbb{N}$, we have $\text{\hyperref[{problem:restricted_main_finite}]{$\mathrm{P}^{0}(\Delta f)$}} \geq \text{\hyperref[{problem:restricted_main_finite}]{$\mathrm{P}(\Delta f)$}} \geq \text{\hyperref[{problem:restricted_main_finite}]{$\mathrm{P}(\Delta f / k)$}}$, where the first inequality is direct and the second inequality is due to Lemma~\ref{lemma:monotonic}.
    
    Consider the following solution to \hyperref[{problem:restricted_main_finite}]{$\mathrm{P}^{0}(\Delta f)$}:
    \begin{align*}
        p(i) \; = \;
        \begin{cases}
            \dfrac{1}{2  \lceil 1 / (2\delta) \rceil} & \text{ if } i \in \{- \lceil 1 / (2\delta)\rceil, \ldots,  \lceil 1 / (2\delta)\rceil - 1\}, \\
            0 & \text{ otherwise}
        \end{cases} 
        \quad \forall i \in \mathbb{Z}
    \end{align*}
    To confirm that $p$ is feasible in \hyperref[{problem:restricted_main_finite}]{$\mathrm{P}^{0}(\Delta f)$}, first note that by construction, $p$ is a valid probability distribution.
    To see that $p$ also satisfies the DP constraints, note that in problem \hyperref[{problem:restricted_main_finite}]{$\mathrm{P}^{0}(\Delta f)$}, these constraints simplify to
    \begin{align*}
        \sum_{i \in \mathbb{Z}} \indicate{I_i(\Delta f) \subseteq A} \cdot p(i) - \sum_{i \in \mathbb{Z}} \indicate{I_i(\Delta f) + \varphi \subseteq A} \cdot p(i) \leq \delta \quad \forall \varphi \in \{ -\Delta f , 0 , \Delta f\}, \ A \in \mathcal{F}(\Delta f).
    \end{align*}
    Since the constraints associated with $\varphi = 0$ are vacuously satisfied, it is sufficient to investigate the cases where $\varphi = \pm \Delta f$. Consider the constraints associated with $\varphi = \Delta f$:
    \begin{align*}
        & \sum_{i \in \mathbb{Z}} (\indicate{I_i(\Delta f) \subseteq A} - \indicate{I_{i+1}(\Delta f) \subseteq A})  \cdot p(i) \leq \delta \mspace{125mu} \forall A \in \mathcal{F}(\Delta f) \\
        \Longleftrightarrow \quad &
        \sup_{A \in \mathcal{F}(\Delta f)} \left[ \sum_{i \in \mathbb{Z}} (\indicate{I_i(\Delta f) \subseteq A} - \indicate{I_{i+1}(\Delta f) \subseteq A}) \cdot p(i) \right] \leq \delta
    \end{align*}
    The supremum in the second row is attained, among others, by the worst-case event $A^\star = I_{\lceil 1 / (2\delta) \rceil - 1}(\Delta f) \in \mathcal{F}(\Delta f)$. Indeed, one readily observes that $i = \lceil 1 / (2\delta) \rceil - 1$ is the only index for which $\indicate{I_{i}(\Delta f) \subseteq A} = 1$ and $\indicate{I_{i+1}(\Delta f) \subseteq A} = 0$. For $A = A^\star$, however, the DP constraint reduces to $p(\lceil 1 / (2\delta) \rceil - 1) = 1 / (2 \lceil 1 / (2 \delta) \rceil) \leq \delta$, which is satisfied by construction. We thus conclude that $p$ satisfies all DP constraints of \hyperref[{problem:restricted_main_finite}]{$\mathrm{P}^{0}(\Delta f)$} associated with $\varphi = \Delta f$ and $A \in \mathcal{F} (\Delta f)$. An analogous argument for $\varphi = -\Delta f$ shows that $p$ indeed satisfies all DP constraints of \hyperref[{problem:restricted_main_finite}]{$\mathrm{P}^{0}(\Delta f)$}.

    The solution $p$ attains a finite objective value in \hyperref[{problem:restricted_main_finite}]{$\mathrm{P}^{0}(\Delta f)$}, finally, since
    \begin{align*}
        \sum_{i \in \mathbb{Z}} c_i(\Delta f) \cdot p(i) = \dfrac{1}{2  \lceil 1 / (2\delta) \rceil} \cdot \sum_{i = - \lceil 1 / (2\delta)\rceil}^{\lceil 1 / (2\delta)\rceil - 1} c_i(\Delta f) =: M <\infty.
    \end{align*}
    Since~\hyperref[{problem:restricted_main_finite}]{$\mathrm{P}^{0}(\Delta f)$} is a minimization problem, $ \text{\hyperref[{problem:restricted_main_finite}]{$\mathrm{P}^{0}(\Delta f)$}} \leq M$ thus follows.
\end{proof}}

Lemma~\ref{lemma:staircase} implies that \hyperref[{problem:restricted_main_finite}]{$\mathrm{P}(\Delta f / k )$} is feasible for any fixed $\varepsilon, \delta > 0$ and $k \in \mathbb{N}$.

\subsubsection{Monotonicity and Feasibility of~\ref{L-UB}}\label{subsection:2}

We first show that problem~\ref{L-UB} is monotonically non-increasing in $L$ and monotonically non-decreasing in $\beta$ in the following sense.

\begin{lemmaA}\label{lemma:monotonic_finite}
For any $\varepsilon > 0$, $\delta > 0$, $L' \in \mathbb{N}$ and $\beta > 0$, the optimal value of problem~\hyperref[{L-UB}]{$\mathrm{P}(L', \beta)$} satisfies $\text{\hyperref[{L-UB}]{$\mathrm{P}(L', \beta)$}} \geq \text{\hyperref[{L-UB}]{$\mathrm{P}(L, \beta / k)$}}$ for all $k \in \mathbb{N}$ and $L \geq L' \cdot k + k - 1$.
\end{lemmaA}

\begin{proof}
The result trivially holds if~\hyperref[{L-UB}]{$\mathrm{P}(L', \beta)$} is infeasible. We thus assume that~\hyperref[{L-UB}]{$\mathrm{P}(L', \beta)$} is feasible, and we fix any $k \in \mathbb{N}$ as well as $L = L' \cdot k + k - 1$. We proceed in two steps. We first derive an upper bound $\mathrm{P'}(L, \beta / k)$ to~\text{\hyperref[{L-UB}]{$\mathrm{P}(L, \beta / k)$}} in which the noise distribution has the same support as in~\hyperref[{L-UB}]{$\mathrm{P}(L', \beta)$}. We then show that~\hyperref[{L-UB}]{$\mathrm{P}(L', \beta)$} bounds $\mathrm{P'}(L, \beta / k)$ from above. The result then follows from the fact that \text{\hyperref[{L-UB}]{$\mathrm{P}(L, \beta / k)$}} is monotonically non-increasing in $L$ for any fixed $\beta$ and $k$.

In view of the first step, we construct the upper bound~$\mathrm{P'}(L, \beta / k)$ to problem~\text{\hyperref[{L-UB}]{$\mathrm{P}(L, \beta / k)$}} by adding to~\hyperref[{L-UB}]{$\mathrm{P}(L, \beta / k)$} the constraint that $p(i) = 0$ for $i = -(L'\cdot k + k - 1), \ldots, - (L' \cdot k + 1)$, that is, we remove the first $k - 1$ elements from the domain of $p$. This ensures that despite its finer interval granularity of $\beta / k$, the support of the noise distribution in problem~$\mathrm{P'}(L, \beta / k)$ is the same as in the more coarsely discretized problem~\hyperref[{L-UB}]{$\mathrm{P}(L', \beta)$}, namely $[-L' \cdot \beta, (L' + 1) \cdot \beta)$.

As for the second step, note that the upper bound~$\mathrm{P'}(L, \beta / k)$ can be formulated as 
\begin{align*}
	\begin{array}{cl}
		\underset{p}{\mathrm{minimize}} & \displaystyle \sum\limits_{i \in \mathbb{Z}} \sum\limits_{l \in [k]} c_{il}(\beta) \cdot p(i,l)  \\[1.5em]
		\mathrm{subject\; to} &\displaystyle p: [\pm L'] \times [k] \mapsto \mathbb{R}_{+}, \ \sum\limits_{i \in \mathbb{Z}} \sum\limits_{l \in [k]} p(i,l) = 1\\[1.5em]
        &\displaystyle \sum\limits_{i \in [\pm L']} \sum\limits_{l \in [k]} \indicate{I_{il}(\beta) \subseteq A} \cdot p(i,l) \leq e^\varepsilon \cdot \sum\limits_{i \in [\pm L']} \sum\limits_{l \in [k]} \indicate{I_{il}(\beta) + \varphi \subseteq A}\cdot p(i,l) + \delta \\
		& \pushright{\forall (\varphi, A) \in \mathcal{E}(\beta / k)},
	\end{array}
\end{align*}
where $I_{il}(\beta) = [(i + (l-1) / k) \cdot \beta, (i + l / k) \cdot \beta)$ and $c_{il}(\beta) = (\beta / k)^{-1} \cdot \int_{x \in I_{il}(\beta)} c(x) \diff x$. Fix any feasible solution $p'$ in problem~\hyperref[{L-UB}]{$\mathrm{P}(L', \beta)$} . One readily observes that the solution $p''(i,l) = p'(i) / k$, $i \in [\pm L']$ and $l \in [k]$, is feasible in~$\mathrm{P}'(L, \beta / k)$ and attains the same objective value as $p'$ in~\hyperref[{L-UB}]{$\mathrm{P}(L', \beta)$}. We thus conclude that \hyperref[{L-UB}]{$\mathrm{P}(L', \beta)$} bounds $\text{\hyperref[{L-UB}]{$\mathrm{P}(L, \beta / k)$}}$ from above, as desired.
\end{proof}

%
%

We next bound the maximum constraint violation of a solution $p'$ in~\ref{problem:restricted_main_finite} that is obtained by truncating any feasible solution $p$ in~\ref{problem:restricted_main_finite} to a bounded domain. This will later enable us to determine values of $L$ that ensure the feasibility of~\ref{L-UB} for any fixed $\beta$.

\begin{lemmaA}\label{lemma:tau}
Let $p$ be an arbitrary feasible solution to problem~\ref{problem:restricted_main_finite} and fix $L \in \mathbb{N}$ such that $\sum_{i \in [\pm L]} p(i) \geq 1 - \tau$ for some $\tau > 0$. Construct another candidate solution $p'$ to~\ref{problem:restricted_main_finite} where $p'(i) = 0$ for all $i \in \mathbb{Z} \setminus [\pm L]$, $p'(L) = \sum_{i \geq L} p(i)$ as well as $p'(-L) = \sum_{i \leq -L} p(i)$, and $p'(i) = p(i)$ otherwise. Then $p'$ violates the DP constraints of problem~\ref{problem:restricted_main_finite} by at most $(1+e^\varepsilon) \cdot \tau$, that is,
\begin{align*}
    \underset{(\varphi, A) \in \mathcal{E}(\beta)}{\sup} \left\{ \sum\limits_{i \in \mathbb{Z}} \indicate{I_i(\beta) \subseteq A} \cdot p'(i) - e^\varepsilon \cdot \sum\limits_{i \in \mathbb{Z}}  \indicate{(I_i(\beta)  + \varphi) \subseteq A} \cdot p'(i) - \delta\right\} \leq (1+e^\varepsilon) \cdot \tau.
\end{align*}
\end{lemmaA}

Note that the constant $L$ in the statement of Lemma~\ref{lemma:tau} exists since for any probability measure $\gamma \in \mathcal{P}_0$ and any $\tau > 0$, there is $L' \in \mathbb{N}$ such that $\gamma([-L, L]) \geq 1 - \tau$ for all $L \geq L'$. \\[-4mm]

\noindent \textbf{Proof of Lemma~\ref{lemma:tau}.} $\;$
Since $p$ is feasible in problem~\ref{problem:restricted_main_finite}, it satisfies 
$$\displaystyle \sum\limits_{i \in \mathbb{Z}} \indicate{I_i(\beta) \subseteq A} \cdot p(i) \leq e^\varepsilon \cdot \sum\limits_{i \in \mathbb{Z}} \indicate{(I_i(\beta)  + \varphi) \subseteq A} \cdot p(i) + \delta \quad \forall (\varphi, A) \in \mathcal{E}(\beta).$$
On the other hand, for any $(\varphi, A) \in \mathcal{E}(\beta)$, the constructed solution $p'$ satisfies
\begin{align*}
    \mspace{-5mu}
	& \displaystyle \sum\limits_{i \in \mathbb{Z}}  \indicate{I_i(\beta) \subseteq A} \cdot p'(i) - e^\varepsilon \cdot \sum\limits_{i \in \mathbb{Z}} \indicate{(I_i(\beta)  + \varphi) \subseteq A} \cdot p'(i) - \delta \\
    = & \displaystyle \sum\limits_{i \in \mathbb{Z}}  \indicate{I_i(\beta) \subseteq A} \cdot [p(i) + (p'(i) - p(i))] - e^\varepsilon \cdot \sum\limits_{i \in \mathbb{Z}} \indicate{(I_i(\beta)  + \varphi) \subseteq A} \cdot [p(i) + (p'(i) - p(i))] - \delta \\
    = & \displaystyle \underbrace{ \sum\limits_{i \in \mathbb{Z}} \indicate{I_i(\beta) \subseteq A} \cdot p(i) - e^\varepsilon \cdot \sum\limits_{i \in \mathbb{Z}}  \indicate{(I_i(\beta)  + \varphi) \subseteq A} \cdot p(i) - \delta}_{\leq 0 \text{ as $p$ is feasible in~\eqref{problem:restricted_main_finite}}} + \\
    & \displaystyle \underbrace{\sum\limits_{i \in \mathbb{Z}}  \indicate{I_i(\beta) \subseteq A}\cdot (p'(i) - p(i)) }_{\leq \tau} - e^\varepsilon \cdot \underbrace{\sum\limits_{i \in \mathbb{Z}}\indicate{(I_i(\beta)  + \varphi) \subseteq A} \cdot (p'(i) - p(i))}_{\geq  -\tau} \leq (1 + e^\varepsilon) \cdot \tau,
\end{align*}
which implies the statement. \qed
\\ 

We can now prove the feasibility of~\ref{L-UB}.

\begin{lemmaA}\label{lemma:feasible}
For any $\varepsilon > 0$, $\delta > 0$ and $\beta > 0$, there exists $L' \in \mathbb{N}$ such that problem~\ref{L-UB} is feasible for all $L \geq L'$.
\end{lemmaA}

\begin{proof}
{\color{black}
Denote by~\hyperref[{problem:restricted_main_finite}]{$\mathrm{P}_{\delta/2}(\beta)$} and \hyperref[L-UB]{$\mathrm{P}_{\delta/2}(L, \beta)$} the variants of~\ref{problem:restricted_main_finite} and~\ref{L-UB} that replace $\delta$ with $\delta / 2$, respectively. Fix any feasible solution $p$ in problem~\hyperref[{problem:restricted_main_finite}]{$\mathrm{P}_{\delta/2}(\beta)$}, whose existence is guaranteed by Lemma~\ref{lemma:staircase}, and choose $L \in \mathbb{N}$ large enough so that $\sum_{i \in [\pm L]} p(i) \geq 1 - (\delta/2) / (1 + e^\varepsilon)$. Lemma~\ref{lemma:tau} allows us to construct a solution $p'$ from $p$ that violates the DP constraints of \hyperref[L-UB]{$\mathrm{P}_{\delta/2}(L, \beta)$} by at most $\delta / 2$. By construction, $p'$ is thus feasible in \hyperref[L-UB]{$\mathrm{P}(L, \beta)$}, and the statement follows.
}
\end{proof}

\subsubsection{Proof of Theorem~\ref{thm:strong_duality}}\label{subsection:3}


To show our convergence result, we define the following auxiliary problem:
\begin{align}\label{c-dual-lb-L-beta}\tag{$\mathrm{M}(L,\beta)$}
    \begin{array}{cl}
        \underset{p}{\mathrm{minimize}} & \displaystyle \sum\limits_{i \in [\pm (L + \Delta f / \beta)]} c_i(\beta) \cdot p(i) \\[1.5em]
        \mathrm{subject\; to} &\displaystyle 
        p:[\pm (L + \Delta f / \beta)] \mapsto \mathbb{R}_{+}, \ \sum\limits_{i \in [\pm (L + \Delta f / \beta)]} p(i) = 1 \\[1.5em]
        &\displaystyle \sum\limits_{i \in [\pm (L + \Delta f / \beta)]} \indicate{I_i(\beta) \subseteq A} \cdot p(i) \leq e^\varepsilon \cdot \sum\limits_{i \in [\pm (L + \Delta f / \beta)]} \indicate{I_i(\beta) + \varphi \subseteq A} \cdot p(i) + \delta \\
        & \pushright{\forall (\varphi, A) \in \mathcal{E}(L, \beta).}
    \end{array}
\end{align}
In the following, we will show that \emph{(i)} problem~\ref{c-dual-lb-L-beta} differs from~\ref{L-UB} only in the domain of the decision variable $p$; \emph{(ii)} problem~\ref{c-dual-lb-L-beta} differs from the strong dual of~\ref{lb-L-beta} only in the objective coefficients; and \emph{(iii)} the relationship $\text{\ref{L-UB}} \geq \text{\ref{c-dual-lb-L-beta}} \geq \text{\ref{lb-L-beta}}$ holds for all $L \in \mathbb{N}$ and $\beta > 0$. Thus, instead of analyzing the convergence of~\ref{L-UB} and~\ref{lb-L-beta} directly, we can analyze separately the convergence of~\ref{L-UB} and~\ref{c-dual-lb-L-beta} (\emph{cf.}~Lemma~\ref{lemma:convergence1}) as well as of~\ref{c-dual-lb-L-beta} and the strong dual of~\ref{lb-L-beta} (\emph{cf.}~Lemma~\ref{lemma:convergence2}).


Since~\ref{c-dual-lb-L-beta} coincides with~\ref{L-UB} except for the additional decision variables $p(i)$, $i \in [\pm (L + \Delta f / \beta)] \setminus [\pm L]$, we have $\text{\ref{L-UB}} \geq \text{\ref{c-dual-lb-L-beta}}$. To show convergence of both problems (\emph{cf.}~Lemma~\ref{lemma:convergence1}), we need to ensure that these additional decision variables take sufficiently small values in optimal solutions to~\ref{c-dual-lb-L-beta}. This is guaranteed by the next result. 

{\color{black}
\begin{lemmaA}\label{lemma:bound_tail}
For any $\varepsilon > 0$, $\delta > 0$ and $\tau > 0$, there exists $L' \in \mathbb{N}$ such that for all $k \in \mathbb{N}$ and $L \geq L'\cdot k + k - 1$, any optimal solution $p^\star$ to~\hyperref[{c-dual-lb-L-beta}]{$\mathrm{M}(L, \Delta f/k)$} satisfies $\sum\limits_{i \in [\pm (L + k)] \setminus [\pm L]} p^\star(i) < \tau$.
\end{lemmaA}
}
{\color{black}
\begin{proof}
Fix $\varepsilon > 0$, $\delta > 0$ and $\tau > 0$. By Lemma~\ref{lemma:feasible}, there is $L_1 \in \mathbb{N}$ such that problem~\hyperref[{L-UB}]{$\mathrm{P}(L, \Delta f)$} is feasible for all $L \geq L_1$. Select $L_2 \in \mathbb{N}$ large enough such that
\begin{align}\label{use_O}
    c(x) > \dfrac{\hyperref[{L-UB}]{\mathrm{P}(L_1, \Delta f)} }{\tau} \quad \forall x \in \mathbb{R}: \  \lvert x \rvert \geq L_2 \cdot \Delta f;
\end{align}
such values exist due to Assumption~\ref{assumptions_c}~\emph{(b)}. We claim that the statement of the lemma holds for $L' = \max \{ L_1, L_2 \}$. To see this, fix any $k \in \mathbb{N}$ and $L \geq L' \cdot k + k - 1$. 

Take any optimal solution $p^\star : [\pm (L + k)] \mapsto \mathbb{R}_{+}$ to problem~\hyperref[{c-dual-lb-L-beta}]{$M(L, \Delta f/k)$} and assume to the contrary that $\sum_{i \in [\pm (L + k)] \setminus [\pm L]} p^\star(i) \geq \tau$. We then observe that
\begin{align}\label{one_single_number_new}
\begin{split}
    \sum_{i \in [\pm (L + k)]} c_i(\Delta f / k) \cdot p^\star(i) &=  \sum_{i \in [\pm L]} c_i(\Delta f / k) \cdot p^\star(i)  + \sum_{i \in [\pm (L + k)] \setminus [\pm L]} c_i(\Delta f / k) \cdot p^\star(i)  \\
    & > \sum_{i \in [\pm (L + k)] \setminus [\pm L]} \dfrac{\hyperref[{L-UB}]{\mathrm{P}(L_1, \Delta f)} }{\tau} \cdot p^\star(i) \\
    & \geq \tau \cdot \dfrac{\hyperref[{L-UB}]{\mathrm{P}(L_1, \Delta f)} }{\tau} = \hyperref[{L-UB}]{\mathrm{P}(L_1, \Delta f)} \geq \hyperref[{L-UB}]{\mathrm{P}(L, \Delta f / k)} \geq \hyperref[{c-dual-lb-L-beta}]{\mathrm{M}(L, \Delta f / k)}.
\end{split}
\end{align}
Here, the first inequality holds since $c_i (\Delta f / k) \geq 0$ for all $i \in \mathbb{Z}$ as well as
\begin{align*}
    c_i (\Delta f / k) \; &= \; \left(\dfrac{\Delta f}{k}\right)^{-1} \cdot \int_{x \in I_i (\Delta f / k)} c(x) \mathrm{d}x \; > \; \left(\dfrac{\Delta f}{k}\right)^{-1} \cdot \int_{x \in I_i (\Delta f / k)} \dfrac{\hyperref[{L-UB}]{\mathrm{P}(L_1, \Delta f)} }{\tau}  \mathrm{d}x = \dfrac{\hyperref[{L-UB}]{\mathrm{P}(L_1, \Delta f)} }{\tau} 
\end{align*}
for all $i \in [\pm (L + k)]\setminus [\pm L]$. The second inequality in~\eqref{one_single_number_new} follows from our earlier assumption that $\sum_{i \in [\pm (L + k)] \setminus [\pm L]} p^\star(i) \geq \tau$. The third inequality in~\eqref{one_single_number_new} is due to Lemma~\ref{lemma:monotonic_finite} and the fact that $L \geq L_1 \cdot k + k - 1$, and the last inequality in~\eqref{one_single_number_new} holds by construction of problem $\hyperref[{c-dual-lb-L-beta}]{\mathrm{M}(L, \Delta f / k)}$. We thus conclude that $p^\star$ cannot be optimal in problem~\hyperref[{c-dual-lb-L-beta}]{$M(L, \Delta f/k)$}, which yields the desired contradiction.
\end{proof}
}

Lemma~\ref{lemma:bound_tail} allows us to prove the convergence between problems~\ref{c-dual-lb-L-beta} and~\ref{L-UB}. 

\begin{lemmaA}\label{lemma:convergence1}
For any $\varepsilon > 0$, $\delta > 0$ and $\xi > 0$, there exists $L' \in \mathbb{N}$ such that $\hyperref[{L-UB}]{\mathrm{P}(L, \Delta f / k)} -\hyperref[{c-dual-lb-L-beta}]{\mathrm{M}(L, \Delta f / k)} \leq \xi$ for all $k \in \mathbb{N}$ and all $L \geq L' \cdot k + k - 1$.
\end{lemmaA}

Intuitively, Lemma~\ref{lemma:convergence1} shows that $\hyperref[{L-UB}]{\mathrm{P}(L, \beta)} -\hyperref[{c-dual-lb-L-beta}]{\mathrm{M}(L, \beta)} \rightarrow 0$ for any $\beta > 0$ as long as $L$ grows sufficiently quickly relative to $1 / \beta$. Recall that the size of the support of the noise distribution $\gamma$ is $L \cdot \beta$. Thus, $\hyperref[{L-UB}]{\mathrm{P}(L, \beta)} -\hyperref[{c-dual-lb-L-beta}]{\mathrm{M}(L, \beta)} \rightarrow 0$ for any $\beta > 0$ as long as the support of $\gamma$ grows large. \\[-4mm]

{\color{black}
\noindent \textbf{Proof of Lemma~\ref{lemma:convergence1}.} $\;$
Fix $\varepsilon > 0$, $\delta > 0$ and $\xi > 0$, select any $\alpha \in (0, \delta)$, set $\hat{\delta} = \delta - \alpha$ and denote by~\hyperref[{L-UB}]{$\mathrm{P}_{ \hat{\delta}}(L, \Delta f/k)$} the variant of~\hyperref[{L-UB}]{$\mathrm{P}(L, \Delta f / k)$} that replaces $\delta$ with $\hat{\delta}$. We invoke Lemma~\ref{lemma:feasible} to select $L_0 \in \mathbb{N}$ so that~\hyperref[{L-UB}]{$\mathrm{P}_{ \hat{\delta}}(L, \Delta f)$} is feasible for all $L \geq L_0$, and we denote by $M$ the optimal value of~\hyperref[{L-UB}]{$\mathrm{P}_{ \hat{\delta}}(L_0, \Delta f)$}. We next invoke Lemma~\ref{lemma:monotonic_finite} to conclude that~\hyperref[{L-UB}]{$\mathrm{P}_{ \hat{\delta}}(L, \Delta f/k)$} remains feasible with an optimal value bounded from above by $M$ for all $k\in \mathbb{N}$ and all $L \geq L_0 \cdot k + k - 1$. 

Set $\tau = \xi \cdot \alpha/M$. The remainder of the proof shows the statement in four steps. Step~1 constructs a solution $p_\tau$ to problem \hyperref[{L-UB}]{$\mathrm{P}(L, \Delta f / k)$} whose expected loss is bounded from above by the optimal value of \hyperref[{c-dual-lb-L-beta}]{$\mathrm{M}(L,\Delta f / k)$}, but that may violate the DP constraints in \hyperref[{L-UB}]{$\mathrm{P}(L, \Delta f / k)$} by up to $\tau$. Step~2 then constructs a convex combination $p^\star$ of $p_\tau$ and $p_{\hat{\delta}}$, where the latter is an optimal solution to problem \text{\hyperref[{L-UB}]{$\mathrm{P}_{ \hat{\delta}}(L, \Delta f/k)$}}. Step~3 shows that the convex combination $p^\star$ is feasible in \hyperref[{L-UB}]{$\mathrm{P}(L, \Delta f / k)$}, and Step~4 shows that the expected loss of $p^\star$ in \hyperref[{L-UB}]{$\mathrm{P}(L, \Delta f / k)$} is bounded from above by $\text{\hyperref[{c-dual-lb-L-beta}]{$\mathrm{M}(L,\Delta f / k)$}} + \xi$, as desired.

In view of Step~1, note that problem~\hyperref[{c-dual-lb-L-beta}]{$\mathrm{M}(L,\Delta f / k)$} is feasible by construction, and it is bounded since the objective coefficients are non-negative. Lemma~\ref{lemma:bound_tail} then ensures the existence of $L_1 \in \mathbb{N}$ such that any optimal solution to \hyperref[{c-dual-lb-L-beta}]{$\mathrm{M}(L,\Delta f / k)$}, $L \geq L_1 \cdot k + k - 1$, places a probability of strictly less than $\tau / (1 + e^\varepsilon)$ outside the index range $[\pm L]$. Select $L_2 \in \mathbb{N}$ large enough such that
\begin{align}\label{use_unb}
    c(x) \geq \underset{x' : \lvert x' \rvert \leq \Delta f}{\max} c(x') \quad \forall x \in \mathbb{R} : \ \lvert x \rvert \geq L_2 \cdot \Delta f;
\end{align}
such values exist due to Assumption~\ref{assumptions_c}~\emph{(b)}. We claim that $L' = \max\{L_0, L_1, L_2\}$ satisfies the statement of this lemma. Take any $k \in \mathbb{N}$ and any $L \geq L'\cdot k + k - 1$, and denote by $p_{\mathrm{M}}$ an optimal solution to problem~\hyperref[{c-dual-lb-L-beta}]{$\mathrm{M}(L,\Delta f / k)$}, which satisfies $\sum_{i \in [\pm (L + k)] \setminus[ \pm L]} p_{\mathrm{M}}(i) < \tau / (1 + e^\varepsilon)$ by the selection of $L'$. Construct a new truncated solution $p_\tau$ via
\begin{align*}
    \displaystyle p_\tau(i) = \begin{cases}
    0 & \text{if } i \in  [\pm (L + k)] \setminus [\pm L]\\
    p_{\mathrm{M}}(0) + \sum_{i' > L} p_{\mathrm{M}}(i') + \sum_{i' < -L} p_{\mathrm{M}}(i') & \text{if } i = 0 \\
    p_{\mathrm{M}}(i) & \text{otherwise}
    \end{cases} \qquad \forall i \in [\pm (L + k)].
\end{align*}
We then have 
\begin{align}\label{use_later_m_ub}
\mspace{-20mu}
    \sum_{i \in [\pm L]} c_i(\Delta f / k) \cdot p_\tau(i) =&  \sum_{i \in [\pm L]} c_i(\Delta f / k) \cdot p_{\mathrm{M}}(i) + c_0(\Delta f / k) \cdot \left(\sum_{ i' > L} p_{\mathrm{M}}(i') + \sum_{ i' < -L} p_{\mathrm{M}}(i')\right) \notag \\
    \leq & \sum_{i \in [\pm (L + k)]} c_i(\Delta f / k) \cdot p_{\mathrm{M}}(i),
\end{align}
where the inequality holds because~\eqref{use_unb} implies 
\begin{align*}
\mspace{-40mu}
    c_0(\Delta f / k) = (\Delta f / k)^{-1} \int_{x \in I_0(\Delta f /k)} c(x) \leq (\Delta f / k)^{-1} \int_{x \in I_0(\Delta f /k)} \  \underset{x':\lvert x' \rvert \leq \Delta f}{\max} c(x') = \max_{x':\lvert x' \rvert \leq \Delta f} c(x') \leq c_i (\Delta f / k)    
\end{align*}
for all $i$ satisfying $\lvert i \rvert > L$. This shows that the truncated solution $p_\tau$ achieves a weakly smaller objective value in problem~\hyperref[{L-UB}]{$\mathrm{P}(L, \Delta f / k)$} than the optimal value of~\hyperref[{c-dual-lb-L-beta}]{$\mathrm{M}(L,\Delta f / k)$}. However, a similar reasoning as in the proof of Lemma~\ref{lemma:tau} shows that $p_\tau$ can violate the DP constraints in~\hyperref[{L-UB}]{$\mathrm{P}(L, \Delta f / k)$} by up to $\tau$.

As for Step~2, we define $p_{\hat{\delta}}$ as an optimal solution to problem~\hyperref[{L-UB}]{$\mathrm{P}_{ \hat{\delta}}(L,\Delta f / k)$}, which exists as \hyperref[{L-UB}]{$\mathrm{P}_{ \hat{\delta}}(L,\Delta f / k)$} is feasible by the selection of $L$. We then construct the solution $p^\star$ via
\begin{align*}
    p^\star = \lambda \cdot p_\tau + (1 - \lambda) \cdot p_{\hat{\delta}} \quad \text{for } \lambda = \dfrac{\alpha}{\tau + \alpha }.
\end{align*}
The next two steps will show that $p^\star$ is feasible in problem~\hyperref[{L-UB}]{$\mathrm{P}(L, \Delta f / k)$} and that its expected loss is bounded from above by $\text{\hyperref[{c-dual-lb-L-beta}]{$\mathrm{M}(L,\Delta f / k)$}} + \xi$.

In view of Step~3, first notice that $p^\star(i) = 0$ for all $i \in [\pm (L + k)] \setminus [\pm L]$ since $p^\star$ is a convex combination of two solutions, neither of which places positive probability on indices $i \in [\pm (L + k)] \setminus [\pm L] $. Consider now any DP constraint $(\varphi, A)$ in problem~\hyperref[{L-UB}]{$\mathrm{P}(L, \Delta f / k)$}. We observe that
\begin{align*}
    & \sum_{i \in [\pm L]} \indicate{I_i(\Delta f / k) \subseteq A} \cdot p^\star(i) - e^\varepsilon \cdot \sum_{i \in  [\pm L]} \indicate{I_i(\Delta f / k) + \varphi \subseteq A}\cdot p^\star(i) \\
    = & \lambda \cdot \left( \sum_{i \in [\pm L]}  \indicate{I_i(\Delta f / k) \subseteq A} \cdot p_\tau(i) - e^\varepsilon \cdot \sum_{i \in [\pm L]}  \indicate{I_i(\Delta f / k) + \varphi \subseteq A} \cdot  p_\tau(i) \right) +\\
     & \qquad (1 - \lambda) \cdot  \left( \sum_{i \in [\pm L]}  \indicate{I_i(\Delta f / k) \subseteq A} \cdot p_{\hat{\delta}}(i) - e^\varepsilon \cdot \sum_{i \in [\pm L]}  \indicate{I_i(\Delta f / k) + \varphi \subseteq A}\cdot p_{\hat{\delta}}(i) \right) \\
    \leq & \lambda \cdot  (\delta + \tau) + (1- \lambda) \cdot  \hat{\delta} \\
    = &  \lambda \cdot (\delta + \tau) + (1- \lambda) \cdot (\delta - \alpha) = \lambda \cdot (\tau + \alpha) + \delta - \alpha = \delta,
\end{align*}
where the first equality uses the definition of $p^\star$, the inequality holds since $p_\tau$ violates the DP constraints in problem~\ref{L-UB} by up to $\tau$ and $p_{\hat{\delta}}$ is feasible in~\hyperref[{L-UB}]{$\mathrm{P}_{\hat{\delta}} (L, \Delta f / k)$}, and the final equalities follow from the definition of $\hat\delta$, rearranging terms and from the definition of $\lambda$, respectively. We thus conclude that $p^\star$ satisfies the DP constraint $(\varphi, A)$ in problem~\hyperref[{L-UB}]{$\mathrm{P}(L, \Delta f / k)$}, and since the choice of $(\varphi, A)$ was arbitrary, $p^\star$ must indeed be feasible in~\hyperref[{L-UB}]{$\mathrm{P}(L, \Delta f / k)$}.

As for Step~4, first note that the solution $p^\star$ achieves an objective value of $\sum_{i\in [\pm L ]} c_i(\Delta f / k) \cdot p^\star(i)$ in~\hyperref[{L-UB}]{$\mathrm{P}(L, \Delta f / k)$}, which itself satisfies the following due to~\eqref{use_later_m_ub}:
\begin{align}\label{ineq_to_use_obj}
    \sum_{i\in [\pm L ]} c_i(\Delta f / k) \cdot p^\star(i)  & = \lambda \sum_{i \in [\pm L]} c_i(\Delta f / k) \cdot p_\tau(i) + (1- \lambda)\cdot  \sum_{i\in [\pm L]} c_i(\Delta f / k) \cdot p_{\hat{\delta}}(i) \notag \\
    & \leq \lambda \sum_{i \in [\pm (L + k)]} c_i(\Delta f / k) \cdot p_{\mathrm{M}}(i) + (1-\lambda) \cdot \sum_{i \in [\pm L]} c_i(\Delta f / k) \cdot p_{\hat{\delta}}(i).
\end{align}
We can use~\eqref{ineq_to_use_obj} to bound the difference $\text{\hyperref[{L-UB}]{$\mathrm{P}(L, \Delta f / k)$}} - \text{\hyperref[{c-dual-lb-L-beta}]{$\mathrm{M}(L, \Delta f / k)$}}$ as follows: 
\begin{align*}
    \text{\hyperref[{L-UB}]{$\mathrm{P}(L, \Delta f / k)$}} - \text{\hyperref[{c-dual-lb-L-beta}]{$\mathrm{M}(L, \Delta f / k)$}} &\leq \sum_{i \in [\pm L]} c_i(\Delta f / k)  \cdot p^\star(i) - \sum_{i \in  [\pm (L + k)]} c_i(\Delta f / k) \cdot  p_{\mathrm{M}}(i)  \\
    & \leq (1 - \lambda) \left( \sum_{i \in  [\pm L]} c_i(\Delta f / k) \cdot p_{\hat{\delta}}(i) - \sum_{i \in [\pm (L + k)]} c_i(\Delta f / k) \cdot p_{\mathrm{M}}(i) \right) \\ 
    & \leq (1 - \lambda) M
    = \dfrac{\tau}{\tau + \alpha} \cdot M = \xi.
\end{align*}
Here, the first inequality holds since $p^\star$ is feasible in \hyperref[{L-UB}]{$\mathrm{P}(L, \Delta f / k)$} and $p_{\mathrm{M}}$ is optimal in \hyperref[{c-dual-lb-L-beta}]{$\mathrm{M}(L,\Delta f / k)$}. The second inequality employs~\eqref{ineq_to_use_obj}. The third inequality bounds the objective value of $p_{\hat{\delta}}$ from above by $M$ and the objective value of $p_{\mathrm{M}}$ from below by $0$, respectively. The two identities, finally, follow from substituting back the definitions of $\lambda$ and $\tau$, respectively.
\qed
}

We next prove the convergence between problems~\ref{c-dual-lb-L-beta} and~\ref{lb-L-beta}.

\begin{lemmaA}\label{lemma:convergence2}
For any $\varepsilon > 0$, $\delta > 0$, $\xi > 0$ and $\Lambda \in \mathbb{N}$, there exists $k' \in \mathbb{N}$ such that $\text{\hyperref[{c-dual-lb-L-beta}]{$\mathrm{M}(\Lambda\cdot k, \Delta f / k) $}} - \text{\hyperref[{lb-L-beta}]{$\mathrm{D}(\Lambda \cdot k,\Delta f / k)$}} \leq \xi$ for all $k \geq k'$.
\end{lemmaA}

Intuitively, Lemma~\ref{lemma:convergence2} shows that $\hyperref[{c-dual-lb-L-beta}]{\mathrm{M}(L, \beta)} - \text{\hyperref[{lb-L-beta}]{$\mathrm{D}(L, \beta)$}} \rightarrow 0$ for any fixed size of the support of the noise distribution $\gamma$ as long as the discretization granularity $\beta$ vanishes to zero. \\[-4mm]

\noindent \textbf{Proof of Lemma~\ref{lemma:convergence2}.} $\;$
Fix $\varepsilon > 0$, $\delta > 0$, $\xi > 0$ and $\Lambda \in \mathbb{N}$. We will show that there is $k' \in \mathbb{N}$ such that $\text{\hyperref[{c-dual-lb-L-beta}]{$\mathrm{M}(\Lambda\cdot k, \Delta f / k) $}} - \text{\hyperref[{dual-lb-L-beta}]{$\overline{D}(\Lambda\cdot k, \Delta f / k) $}} \leq \xi$ for all $k \geq k'$, where~\hyperref[{dual-lb-L-beta}]{$\overline{D}(\Lambda\cdot k, \Delta f / k) $} denotes the dual to $\text{\hyperref[{lb-L-beta}]{$\mathrm{D}(\Lambda \cdot k,\Delta f / k)$}}$. Indeed, strong duality holds between $\text{\hyperref[{lb-L-beta}]{$\mathrm{D}(\Lambda \cdot k,\Delta f / k)$}}$ and~\hyperref[{dual-lb-L-beta}]{$\overline{D}(\Lambda\cdot k, \Delta f / k) $} since $(\theta, \psi)$ with $\theta = \min \{\underline{c}_i({\color{black}\Delta f / k}) : i \in [\pm (\Lambda \cdot k + {\color{black}k})] \}$ and $\psi (\varphi, A) = 0$ for all $\mathcal{E}(\Lambda \cdot k,{\color{black}\Delta f / k})$ is feasible in~\hyperref[{lb-L-beta}]{$\mathrm{D},\Delta f / k)$}. The dual problem \hyperref[{dual-lb-L-beta}]{$\overline{D}(\Lambda\cdot k, \Delta f / k) $} can be formulated as
\begin{align}\label{dual-lb-L-beta}\tag{$\overline{\mathrm{D}}(\Lambda \cdot k,\Delta f / k)$}
\mspace{-15mu}
    \begin{array}{cl}
        \underset{p}{\mathrm{minimize}} & \displaystyle \sum\limits_{i \in [\pm (\Lambda \cdot k + k)]} \underline{c}_i(\Delta f / k) \cdot p(i)  \\[1.5em]
        \mathrm{subject\; to} &\displaystyle p: [\pm (\Lambda \cdot k + k)]\mapsto \mathbb{R}_{+}, \  \sum\limits_{i \in [\pm (\Lambda \cdot k + k)]} p(i) = 1 \\[1.5em]
        &\displaystyle \sum\limits_{i \in [\pm (\Lambda \cdot k + k)]} \indicate{I_i(\Delta f / k) \subseteq A} \cdot p(i) \leq e^\varepsilon \cdot \sum\limits_{i \in [\pm (\Lambda \cdot k + k)]} \indicate{I_i(\Delta f / k) + \varphi \subseteq A}\cdot  p(i) + \delta \\
        & \pushright{\forall (\varphi, A) \in \mathcal{E}(\Lambda \cdot k, \Delta f / k).}
    \end{array}
\end{align}
Note that~\ref{dual-lb-L-beta} and~\hyperref[{c-dual-lb-L-beta}]{$\mathrm{M}(\Lambda\cdot k, \Delta f / k) $} only differ in their objective coefficients $\underline{c}_i({\color{black}\Delta f / k})$ and $c_i({\color{black} \Delta f / k})$, respectively. We now show that the difference between those two coefficient sets can be made arbitrarily small, uniformly across all $i \in [\pm (\Lambda \cdot k + k)]$, by increasing $k$. Indeed, the loss function $c$ is continuous by Assumption~\ref{assumptions_c}~\emph{(a)}, and it is therefore uniformly continuous over the (closure of the) finite interval $\bigcup \{ I_i (\Delta f / k) \, : \, i \in [\pm (\Lambda \cdot k + k)] \}$ by the Heine-Cantor theorem. (Note in particular that for any $k$, this interval is contained in the set $[- (\Lambda + 1) \cdot \Delta f, (\Lambda + 2) \cdot \Delta f)$ that is independent of $k$, which justifies our use of the Heine-Cantor theorem.) For the selected $\xi > 0$, we can thus find $k' \in \mathbb{N}$ such that for all $k \geq k'$, we have
\begin{align*}
    c_i(\Delta f / k) - \underline{c}_i(\Delta f / k)
    \; &= \;
    \left( \frac{\Delta f}{k} \right)^{-1} \cdot \int_{x \in I_i(\Delta f / k)} c(x) \diff x - \inf_{x \in I_i(\Delta f / k)} c(x) \\
    &\leq \;
    \sup_{x \in I_i(\Delta f / k)} c(x) - \inf_{x \in I_i(\Delta f / k)} c(x)
    \; \leq \; \xi,
\end{align*}
uniformly across all $i \in [\pm (\Lambda \cdot k + k)]$. Here, the identity replaces $c_i(\Delta f / k)$ and $\underline{c}_i(\Delta f / k)$ with their respective definitions, the first inequality exploits that $c (x) \leq \sup_{x \in I_i(\Delta f / k)} c(x)$ for all $x \in I_i(\Delta f / k)$, and the last inequality makes use of the uniform continuity of $c$.

To bound $\text{\hyperref[{c-dual-lb-L-beta}]{$\mathrm{M}(\Lambda \cdot k, \Delta f / k) $}} -  \text{\hyperref[{dual-lb-L-beta}]{$\overline{\mathrm{D}}(\Lambda \cdot k,\Delta f / k)$}}$, take any optimal solution $p^\star$ to problem~\hyperref[{dual-lb-L-beta}]{$\overline{\mathrm{D}}(\Lambda \cdot k, \Delta f / k)$} and notice that $p^\star$ is feasible in \hyperref[{c-dual-lb-L-beta}]{$\mathrm{M}(\Lambda \cdot k, \Delta f / k)$} since both problems only differ in their objective coefficients. We thus have
\begin{equation*}
	\text{\hyperref[{c-dual-lb-L-beta}]{$\mathrm{M}(\Lambda \cdot k, \Delta f / k) $}} -  \text{\hyperref[{dual-lb-L-beta}]{$\overline{\mathrm{D}}(\Lambda \cdot k,\Delta f / k)$}}
    \; \leq \;
    \sum_{i \in [\pm (\Lambda \cdot k + k)]} \left[ c_i(\Delta f / k) - \underline{c}_i(\Delta f / k) \right] \cdot p^\star(i)
    \; \leq \;
    \xi,
\end{equation*}
where the first inequality holds since $p^\star$ is optimal in \hyperref[{dual-lb-L-beta}]{$\overline{\mathrm{D}}(\Lambda \cdot k, \Delta f / k)$} but feasible (and possibly not optimal) in \hyperref[{c-dual-lb-L-beta}]{$\mathrm{M}(\Lambda \cdot k, \Delta f / k)$}, and the second inequality is due to our earlier uniform bound on $c_i(\Delta f / k) - \underline{c}_i(\Delta f / k)$ and the fact that $p^\star$ is a probability distribution.
\qed \\[-4mm]

\noindent \textbf{Proof of Theorem~\ref{thm:strong_duality}.} $\;$
    Fix $\xi > 0$ as well as any $\xi_{1}, \xi_{2} > 0$ satisfying $\xi_{1} +  \xi_{2} = \xi$. Since
\begin{align*}
\mspace{-40mu}
   \text{\hyperref[{L-UB}]{$\mathrm{P}(\Lambda\cdot k, \Delta f / k)$}} - \text{\hyperref[{lb-L-beta}]{$\mathrm{D}(\Lambda\cdot k, \Delta f / k)$}}  = \left[\text{\hyperref[{L-UB}]{$\mathrm{P}(\Lambda\cdot k, \Delta f / k)$}} - \text{\hyperref[{c-dual-lb-L-beta}]{$\mathrm{M}(\Lambda\cdot k, \Delta f / k)$}} \right] + \left[\text{\hyperref[{c-dual-lb-L-beta}]{$\mathrm{M}(\Lambda\cdot k, \Delta f / k)$}} - \text{\hyperref[{lb-L-beta}]{$\mathrm{D}(\Lambda\cdot k, \Delta f / k)$}} \right],
\end{align*}
for any $\Lambda \in \mathbb{N}$ and $k \in \mathbb{N}$, it suffices to show that there is $\Lambda' \in \mathbb{N}$ and $k' \in \mathbb{N}$ such that
\begin{equation}\label{eq:thm1:two_ineqs}
    \text{\hyperref[{L-UB}]{$\mathrm{P}(\Lambda\cdot k, \Delta f / k)$}} - \text{\hyperref[{c-dual-lb-L-beta}]{$\mathrm{M}(\Lambda\cdot k, \Delta f / k)$}} \leq \xi_{1} \quad \text{and}\quad \text{\hyperref[{c-dual-lb-L-beta}]{$\mathrm{M}(\Lambda\cdot k, \Delta f / k)$}} -  \text{\hyperref[{lb-L-beta}]{$\mathrm{D}(\Lambda\cdot k, \Delta f / k)$}} \leq \xi_{2}
\end{equation}
simultaneously hold for all $\Lambda \geq \Lambda'$ and $k \geq k'$.

In view of the first inequality in~\eqref{eq:thm1:two_ineqs}, we invoke Lemma~\ref{lemma:convergence1} to select $L' \in \mathbb{N}$ such that
$\text{\hyperref[{L-UB}]{$\mathrm{P}(L, \Delta f / k)$}} - \text{\hyperref[{c-dual-lb-L-beta}]{$\mathrm{M}(L, \Delta f / k)$}} \leq \xi_{1}$ for all $k \in \mathbb{N}$ and all $L \geq L' \cdot k + k - 1$. Fix $\Lambda' = L' + 1$ and notice that $\Lambda' \cdot k = (L' + 1) \cdot k \geq L' \cdot k + k - 1$. For our choice of $\Lambda'$, we thus have $\text{\hyperref[{L-UB}]{$\mathrm{P}(\Lambda' \cdot k, \Delta f / k)$}} - \text{\hyperref[{c-dual-lb-L-beta}]{$\mathrm{M}(\Lambda' \cdot k, \Delta f / k)$}} \leq \xi_{1}$ for all $k \in \mathbb{N}$. As for the second inequality in~\eqref{eq:thm1:two_ineqs}, we invoke Lemma~\ref{lemma:convergence2} to select $k' \in \mathbb{N}$ such that $\text{\hyperref[{c-dual-lb-L-beta}]{$\mathrm{M}(\Lambda' \cdot k, \Delta f / k) $}} -  \text{\hyperref[{lb-L-beta}]{$\mathrm{D}(\Lambda' \cdot k,\Delta f / k)$}} \leq \xi_{2}$ for all $k \geq k'$. So far, we have shown that there exists $\Lambda' \in \mathbb{N}$ and $k' \in \mathbb{N}$ such that 
$\text{\hyperref[{L-UB}]{$\mathrm{P}(\Lambda'\cdot k, \Delta f / k)$}} - \text{\hyperref[{lb-L-beta}]{$\mathrm{D}(\Lambda'\cdot k, \Delta f / k)$}} \leq \xi$
holds for all $k \geq k'$. Since we have $\text{\hyperref[{L-UB}]{$\mathrm{P}(\Lambda'\cdot k, \Delta f / k)$}} \geq \text{\hyperref[{L-UB}]{$\mathrm{P}(\Lambda\cdot k, \Delta f / k)$}}$ and $\text{\hyperref[{lb-L-beta}]{$\mathrm{D}(\Lambda'\cdot k, \Delta f / k)$}} \leq \text{\hyperref[{lb-L-beta}]{$\mathrm{D}(\Lambda\cdot k, \Delta f / k)$}}$ for all $\Lambda \geq \Lambda'$, we can conclude the proof.
\qed

{\color{black}
\subsection{Proof of Proposition~\ref{prop:impossibility}}

In the following, we present two counterexamples that jointly prove the statement of Proposition~\ref{prop:impossibility}. Both examples rely on the following strong dual of \ref{lb-L-beta}:
\begin{align}\label{dual_form_to_refer}
    \begin{array}{cl}
        \underset{p}{\mathrm{minimize}} & \displaystyle \sum\limits_{i \in [\pm (L + \Delta f / \beta)]} \underline{c}_i(\beta) \cdot p(i) \\[1.5em]
        \mathrm{subject\; to} &\displaystyle 
        p:[\pm (L + \Delta f / \beta)] \mapsto \mathbb{R}_{+}, \ \sum\limits_{i \in [\pm (L + \Delta f / \beta)]} p(i) = 1 \\[1.5em]
        &\displaystyle \sum\limits_{i \in [\pm (L + \Delta f / \beta)]} \indicate{I_i(\beta) \subseteq A} \cdot p(i) \leq e^\varepsilon \cdot \sum\limits_{i \in [\pm (L + \Delta f / \beta)]} \indicate{I_i(\beta) + \varphi \subseteq A} \cdot p(i) + \delta \\
        & \pushright{\forall (\varphi, A) \in \mathcal{E}(L, \beta)}
    \end{array}
\end{align}
This dual has been derived earlier in the proof of Lemma~\ref{lemma:convergence2}; the version above is adapted slightly to match the notation of our counterexamples.

\subsubsection{Example 1: Violation of Assumption~\ref{assumptions_c}~\emph{(a)}}\label{counter_continuity}

\begin{figure}[tb]
\begin{center}
\begin{tikzpicture}
    \draw[thick] (0,1) circle [radius=0.1];
    \draw[fill,thick] (0,0) circle [radius=0.1];
    \draw[dashed] (0,0) -- (0,-0.5);
    \node at (0,-0.7) {$0$};
    
    \draw[thick] (0.1,1) -- (3.9,1); 
    \draw[thick] (-0.1,1) -- (-3.9,1); 
    \draw[dashed] (4,0.9) -- (4,-0.5);
    \node at (4,-0.7) {$2 \lceil 1 / (2\delta)\rceil \Delta f$}; 
    \draw[dashed] (-4,0.9) -- (-4,-0.5);
    \node at (-4,-0.7) {$-2 \lceil 1 / (2\delta)\rceil \Delta f$}; 

    \draw[fill,thick] (4,1) circle [radius=0.1];
    \draw[fill,thick] (-4,1) circle [radius=0.1];

    \draw[thick] (4,2) circle [radius=0.1];
    \draw[thick] (4.1,2.05) -- (6, 3); 

    \draw[thick] (-4,2) circle [radius=0.1];
    \draw[thick] (-4.1,2.05) -- (-6, 3); 

    \node[scale = 0.4] at (6.2,3.1){\textbullet};
    \node[scale = 0.4] at (6.4,3.2){\textbullet};
    \node[scale = 0.4] at (6.6,3.3){\textbullet};
    \node[scale = 0.4] at (-6.2,3.1){\textbullet};
    \node[scale = 0.4] at (-6.4,3.2){\textbullet};
    \node[scale = 0.4] at (-6.6,3.3){\textbullet};

    \node[scale = 1] at (1.0,0.0) {$c(x) = 0$};
    \node[scale = 1] at (5,1.0) {$c(x) = 1$};
    \node[scale = 1, rotate = 27] at (5,2.9) {$c(x) = x$};

    \node[scale = 1, rotate = 333] at (-5,2.9) {$c(x) = -x$};


\end{tikzpicture}
\end{center}
    \caption{\color{black} Visualization of the loss function $c$ in Section~\ref{counter_continuity}.}
    \label{figure:counter-ex-2}
\end{figure} 

Fix any $\Delta f, \varepsilon, \delta > 0$ and consider the following loss function $c$: 
\begin{align*}
c(x) \; = \;    
    \begin{cases}
    0 & \text{ if } x = 0,\\
    1 & \text{ if } x \neq 0 \text{ and } \lvert x \rvert \leq 2 \cdot \lceil 1 / (2\delta) \rceil \cdot \Delta f, \\
    1 + \lvert x \rvert & \text{ otherwise}.
    \end{cases}
\end{align*}
Figure~\ref{figure:counter-ex-2} shows that $c$ satisfies Assumption~\ref{assumptions_c}~\emph{(b)} since for any $r \in \mathbb{R}$ we have $c(x) \geq r$ for all $x$ satisfying $\lvert x \rvert \geq \max\{2 \lceil 1 / (2\delta) \rceil \Delta f, \, r\}$. However, $c$ violates the continuity condition of Assumption~\ref{assumptions_c} at $x = 0$ and $x = \pm 2 \lceil 1 / (2\delta) \rceil \Delta f$. In the following, we will show that for this loss function, \ref{L-UB} and~\ref{lb-L-beta} differ by at least $\delta / 2$ for all $L \in \mathbb{N}$ and $\beta > 0$. We do so in two steps: We first show that $\text{\ref{problem:restricted_main_finite}} \geq 1$, and we subsequently argue that $\text{\ref{lb-L-beta}} \leq 1 - \delta / 2$. The statement then follows since $\text{\ref{L-UB}} \geq \text{\ref{problem:restricted_main_finite}}$.

To see that $\text{\ref{problem:restricted_main_finite}} \geq 1$, we note that the objective coefficients in $\text{\ref{problem:restricted_main_finite}}$ satisfy
\begin{align*}
    c_i(\beta) = \beta^{-1} \cdot \int_{x \in I_i(\beta)} c(x) \diff x \geq 1 \quad \forall i \in \mathbb{Z},
\end{align*}
where the inequality holds since $c (x) \geq 1$ almost everywhere (except at $x = 0$). The claim then follows from the fact that the objective function in $\text{\ref{problem:restricted_main_finite}}$ constitutes a convex combination of those objective coefficients.

To show that $\text{\ref{lb-L-beta}} \leq 1 - \delta/2$, we can without loss of generality assume that $L \geq 2 \cdot \lceil 1 / (2\delta) \rceil \cdot k$ for $k\in \mathbb{N}$ satisfying $\beta = \Delta f / k$. Indeed, the result then follows for all smaller values of $L$ due to the monotonicity of $\text{\ref{lb-L-beta}}$ in $L$. To see that $\text{\ref{lb-L-beta}} \leq 1 - \delta/2$, we construct a feasible solution $p$ to the strong dual~\eqref{dual_form_to_refer} of~\ref{lb-L-beta} that attains the objective value $1 - \delta/2$. To this end, define $\mathcal{I}^k := \{ -2 \cdot \lceil 1 / (2\delta) \rceil \cdot k,\ldots, 2 \cdot \lceil 1 / (2\delta) \rceil \cdot k - 1 \}$ and set
\begin{align*}
    p(i) \; = \;
    \begin{cases}
        \delta / 2 & \text{ if } i = 0, \\
        \dfrac{1 - \delta/ 2}{4 \lceil 1 / (2\delta)\rceil k - 1} & \text{ if } i \in \mathcal{I}^k\setminus \{ 0\}, \\
        0 & \text{ if } i \in [\pm (L + \Delta f / \beta)] \setminus \mathcal{I}^k
    \end{cases}
    \qquad \forall i \in \mathbb{Z}.
\end{align*}
Note that $p$ indeed constitutes a probability distribution over $i \in \mathcal{Z}$, and that $p (i) = 0$ for all $i \notin [\pm L]$ since $L \geq 2 \cdot \lceil 1 / (2\delta) \rceil \cdot k$. Moreover, we have
\begin{align*}
    \sum\limits_{i \in [\pm (L + \Delta f / \beta)] } \underline{c}_i(\beta) \cdot p(i) = \sum\limits_{i \in \mathcal{I}^k} \underline{c}_i(\beta) \cdot p(i) =  \underline{c}_0(\beta)\cdot p(0) + \sum_{i \in \mathcal{I}^k \setminus \{0\} } p(i)  = 1 - \delta/2,
\end{align*}
where the last identity follows from the fact that $\underline{c}_0(\beta) = \inf \{ c(x) : x \in I_0(\beta) \} = 0$ since $0 \in I_0(\beta)$ as well as $\underline{c}_i(\beta) = 1$ for all $i \in \mathcal{I}^k\setminus \{ 0\}$ since $\lvert x \rvert \leq 2 \cdot \lceil 1 /(2\delta) \rceil \cdot \Delta f$ for all 
$x \in I_{i}(\beta)$ and $i \in \mathcal{I}^k\setminus\{0\}$. \mbox{To see that $p$ satisfies the constraints in~\eqref{dual_form_to_refer}, fix any $(\varphi, A) \in \mathcal{E}(L, \beta)$ and observe that}
\begin{align*}
    & \displaystyle \sum\limits_{i \in [\pm (L + \Delta f / \beta)]} \indicate{I_i(\beta) \subseteq A} \cdot p(i) - e^\varepsilon \cdot \sum\limits_{i \in [\pm (L + \Delta f / \beta)]} \indicate{I_i(\beta) + \varphi \subseteq A} \cdot p(i) \\
    \leq & \displaystyle \sum\limits_{i \in [\pm (L + \Delta f / \beta)]} (\indicate{I_i(\beta) \subseteq A} - \indicate{I_i(\beta) + \varphi \subseteq A} ) \cdot p(i) \\
    \leq & \displaystyle \sum\limits_{i \in [\pm L ]} (\indicate{I_i(\beta) \subseteq A} - \indicate{I_i(\beta) + \varphi \subseteq A}) \cdot p(i) \\
    \leq & \displaystyle \sum\limits_{i \in [\pm L ]} \indicate{I_i(\beta) \subseteq A \land I_i(\beta) + \varphi \not\subseteq A} \cdot p(i) \\
    \leq & \displaystyle \sum\limits_{i \in [\pm L ]} \indicate{I_i(\beta) + \varphi \not\subseteq A} \cdot p(i) \\
    \leq & \displaystyle \sum\limits_{i \in [\pm L ]} \indicate{I_i(\beta) + \varphi \not\subseteq \bigcup_{i' \in [\pm L]} I_{i'}(\beta)} \cdot p(i)  \; \leq \; \displaystyle \max_{i_1 < \ldots < i_k}\{p(i_1) + \ldots + p(i_k)\}  \;  \leq \; \delta.
\end{align*}
Here, the first inequality holds since $\varepsilon \geq 0$, and the second inequality uses the fact that $\indicate{I_i(\beta) \subseteq A} - \indicate{I_i(\beta) + \varphi \subseteq A} \leq 0$ for all $i \in [\pm (L + k)] \setminus [\pm L]$. The third inequality holds because $\indicate{I_i(\beta) \subseteq A} - \indicate{I_i(\beta) + \varphi \subseteq A} = 1$ whenever $I_i(\beta) \subseteq$ and $I_i(\beta) + \varphi \not\subseteq A$, whereas $\indicate{I_i(\beta) \subseteq A} - \indicate{I_i(\beta) + \varphi \subseteq A} \in \{ -1, 0 \}$ otherwise. The fourth inequality disregards one of operands in the conjunction inside the indicator terms. The fifth inequality exploits the fact that $A \subseteq \bigcup_{i' \in [\pm L]} I_{i'}(\beta)$ for any $A \in \mathcal{F}(L, \beta)$. The sixth inequality holds since there are at most $k$ indices $i \in [ \pm L ]$ satisfying $I_i(\beta) + \varphi \not\subseteq \bigcup_{i' \in [\pm L]} I_{i'}(\beta)$ since $\varphi \in \{- \Delta f, - \Delta f + \beta, \ldots, \Delta f \}$ for $\beta = \Delta f / k$. 
The last inequality, finally, holds by construction of $p$, whose non-zero values are $\delta / 2$ (at $i = 0$) and $(1 - \delta / 2) / (4 \lceil 1 / (2\delta)\rceil k - 1)$ (at $i \in \mathcal{I}^k \setminus \{ 0 \}$), and the fact that $\delta / 2 + (k - 1) \cdot  (1 - \delta / 2) / (4 \lceil 1 / (2\delta)\rceil k - 1) = \delta$. 


\subsubsection{Example 2: Violation of Assumption~\ref{assumptions_c}~\emph{(b)}}\label{counter_zigzag}

\begin{figure}[tb]
\begin{center}
\begin{tikzpicture}
    \draw[thick] (0,1) -- (2,1); 
    \draw[thick] (4,1) -- (6,1); 
    \draw[thick] (8,1) -- (10,1);

    \draw[thick] (2.5,0) -- (3.5,0); 
    \draw[thick] (6.5,0) -- (7.5,0); 

    \draw[thick] (2,1) -- (2.5,0); 
    \draw[thick] (3.5,0) -- (4,1); 
    \draw[thick] (6,1) -- (6.5,0); 
    \draw[thick] (7.5,0) -- (8,1); 


    \draw[dashed] (2,1) -- (2,-0.1);

    \draw[dashed] (4,1) -- (4,-0.1);

    \draw[dashed] (6,1) -- (6,-0.1);

    \draw[dashed] (8,1) -- (8,-0.1);
    

    \node[scale = 0.5] at (-0.5,0.25){\textbullet};
    \node[scale = 0.5] at (-1,0.25){\textbullet};
    \node[scale = 0.5] at (-1.5,0.25){\textbullet};
    \node[scale = 0.5] at (10.5,0.5){\textbullet};
    \node[scale = 0.5] at (11,0.5){\textbullet};
    \node[scale = 0.5] at (11.5,0.5){\textbullet};

    \node[scale = 1] at (-1,1) {$c(x) = 1$};
    \node[scale = 1] at (11,0) {$c(x) = 0$};

    \draw [black,decorate,decoration={brace,amplitude=10pt,mirror},xshift=0pt,yshift=0pt] (0.1,-0.2) -- (1.9,-0.2) node [black,midway, yshift = -0.8cm] {$C_0$};
    \draw [black,decorate,decoration={brace,amplitude=10pt,mirror},xshift=0pt,yshift=0pt] (2.1,-0.2) -- (3.9,-0.2) node [black,midway, yshift = -0.8cm] {$D_0$};
    \draw [black,decorate,decoration={brace,amplitude=10pt,mirror},xshift=0pt,yshift=0pt] (4.1,-0.2) -- (5.9,-0.2) node [black,midway, yshift = -0.8cm] {$C_1$};
    \draw [black,decorate,decoration={brace,amplitude=10pt,mirror},xshift=0pt,yshift=0pt] (6.1,-0.2) -- (7.9,-0.2) node [black,midway, yshift = -0.8cm] {$D_1$};
    \draw [black,decorate,decoration={brace,amplitude=10pt,mirror},xshift=0pt,yshift=0pt] (8.1,-0.2) -- (9.9,-0.2) node [black,midway, yshift = -0.8cm] {$C_2$};

    \draw [black,decorate,decoration={brace,amplitude=10pt},xshift=0pt,yshift=0pt] (2,1.2) -- (2.5,1.2) node [black,midway, yshift = 0.65cm] {\small $D^{\mathrm{L}}_0$};
    \draw [black,decorate,decoration={brace,amplitude=10pt},xshift=0pt,yshift=0pt] (2.55,1.2) -- (3.45,1.2) node [black,midway, yshift = 0.65cm] {\small $D^{\mathrm{M}}_0$};
    \draw [black,decorate,decoration={brace,amplitude=10pt},xshift=0pt,yshift=0pt] (3.5,1.2) -- (4.0,1.2) node [black,midway, yshift = 0.65cm] {\small $D^{\mathrm{R}}_0$};

    \draw[dashed] (2.5,1) -- (2.5,0.1);
    \draw[dashed] (3.5,1) -- (3.5,0.1);
    \draw[dashed] (4,1) -- (4,0.1);

\end{tikzpicture}
\end{center}
    \caption{\color{black} Visualization of the loss function $c$ in Section~\ref{counter_zigzag}.}
    \label{figure:counter-ex}
\end{figure} 

Fix any $\Delta f, \varepsilon, \delta > 0$ and consider the following loss function $c$:
\begin{align*}
    c(x) \; = \;    
    \begin{cases}
    1 & \text{ if } x \in C_i \text{ for some $i \in \mathbb{Z}$},\\
    5 + 8i - \dfrac{8x}{\Delta f} & \text{ if } x \in D^{\mathrm{L}}_i \text{ for some $i \in \mathbb{Z}$},\\
    0 & \text{ if } x \in D^{\mathrm{M}}_i \text{ for some $i \in \mathbb{Z}$},\\
    \dfrac{8x}{\Delta f} - 8i - 7 & \text{ if } x \in D^{\mathrm{R}}_i \text{ for some $i \in \mathbb{Z}$}.
    \end{cases}
\end{align*}
Here, $C_i := [i\cdot \Delta f, (i + \frac{1}{2}) \cdot \Delta f )$, and $D_i := [ (i + \frac{1}{2}) \cdot \Delta f , (i + 1) \cdot \Delta f), \ i \in \mathbb{Z}$, are intervals that partition $\mathbb{R}$, and
\begin{align*}
    \begin{array}{ll}
        \displaystyle
        D^{\mathrm{L}}_i = \Big[\big(i + \frac{1}{2}\big)\cdot \Delta f,\; \big(i + \frac{5}{8}\big)\cdot \Delta f \Big) & \displaystyle\text{ with } \lvert D^{\mathrm{L}}_i \rvert = \frac{\Delta f}{8}, \\[1em]
        \displaystyle
        D^\mathrm{M}_i = \Big[\big(i + \frac{5}{8}\big)\cdot \Delta f,\; \big(i + \frac{7}{8}\big)\cdot \Delta f \Big) & \displaystyle\text{ with } \lvert D^{\mathrm{M}}_i \rvert = \frac{\Delta f}{4},\\[1em]
        \displaystyle
        D^{\mathrm{R}}_i = \Big[\big(i + \frac{7}{8}\big)\cdot \Delta f,\; \big(i + 1\big)\cdot \Delta f \Big) & \displaystyle \text{ with } \lvert D^{\mathrm{R}}_i \rvert = \frac{\Delta f}{8}
    \end{array}
\end{align*}
further partition each interval $D_i$ into a left (superscript `L'), middle (superscript `M') and right (superscript `R') sub-interval, respectively. Figure~\ref{figure:counter-ex} shows that $c$ satisfies Assumption~\ref{assumptions_c}~\emph{(a)} since $c$ is piecewise linear without discontinuities. However, $c$ violates the unboundedness condition of Assumption~\ref{assumptions_c} since its image is contained in the interval $[0,1]$. In the following, we will show that for this loss function, \ref{L-UB} and~\ref{lb-L-beta} differ by at least $(1 - \delta) / (1 + e^\varepsilon)$ for all $L \in \mathbb{N}$ and $\beta > 0$. We do so in two steps: We first show that $\text{\ref{problem:integral_main}} \geq (1 - \delta) / (1 + e^\varepsilon)$, and we subsequently argue that $\text{\ref{lb-L-beta}} \leq 0$. The statement then follows since $\text{\ref{L-UB}} \geq \text{\ref{problem:integral_main}}$. 

To see that $\text{\ref{problem:integral_main}} \geq (1 - \delta) / (1 + e^\varepsilon)$, define $C = \bigcup_{i \in \mathbb{Z}} C_i$ and $D = \bigcup_{i \in \mathbb{Z}}D_i$ such that $D - (\Delta f / 2) = C$. The DP constraint of~\ref{problem:integral_main} indexed by $(\varphi, A) \in \mathcal{E}$ with $\varphi = \Delta f / 2$ and $A = D$ implies that
\begin{align*}
    \int_{x \in \mathbb{R}} \indicate{x \in D} \diff \gamma(x) \leq e^\varepsilon \cdot \int_{x \in \mathbb{R}} \indicate{x \in D - (\Delta f / 2)} \diff \gamma(x) + \delta = e^\varepsilon \cdot \int_{x \in \mathbb{R}} \indicate{x \in C} \diff \gamma(x) + \delta.
\end{align*}
Let us refer to $\int_{x \in \mathbb{R}} \indicate{x \in C} \diff \gamma(x)$ as $\gamma(C)$ and $\int_{x \in \mathbb{R}} \indicate{x \in D} \diff \gamma(x)$ as $\gamma(D)$, respectively, so that the above inequality reads as $\gamma(D) \leq e^\varepsilon \cdot \gamma(C)+\delta$. Since $C$ and $D$ partition $\mathbb{R}$ and $\gamma(\mathbb{R}) = 1$, we have $\gamma(D) = 1 - \gamma(C)$ and thus $\gamma(C) \geq (1 - \delta) /(1 + e^\varepsilon)$. Since $c(x) = 1$ for $x \in C$ and $c(x) \geq 0$ for all $x \in \mathbb{R}$, any feasible solution to~\ref{problem:integral_main} must satisfy
\begin{align*}
    \int_{x \in \mathbb{R}} c(x) \diff \gamma(x) = \int_{x \in C} c(x) \diff \gamma(x) + \int_{x \in D} c(x) \diff \gamma(x) \geq  \int_{x \in C} \diff \gamma(x) = \dfrac{1 - \delta}{1 + e^\varepsilon},
\end{align*}
which shows that indeed $\text{\ref{problem:integral_main}} \geq (1-  \delta) / (1 + e^\varepsilon)$.

To show that $\text{\ref{lb-L-beta}} \leq 0$ for any $L \in \mathbb{N}$ and $\beta > 0$, we construct a feasible solution $p$ to the strong dual~\eqref{dual_form_to_refer} of~\ref{lb-L-beta} that attains the objective value $0$. We do so in two steps: We first show that there is an index $i^\star \in \mathbb{Z}$ with $\underline{c}_{i^\star}(\beta) = 0$, and we subsequently argue that the solution $p : [\pm (L + \Delta f / \beta)] \mapsto \mathbb{R}_{+}$ with $p (i) = \indicate{i = i^\star}$ is feasible in~\eqref{dual_form_to_refer}. In view of the first step, note that by construction of $c$, any interval of length $\Delta f$ has a non-empty intersection with some sub-interval $D^{\mathrm{M}}_{i}$, $i \in \mathbb{Z}$. Since the union $\bigcup_{i = L+1}^{L+\Delta f / \beta} I_i(\beta)$ is an interval of length $\Delta f$, there must thus exist an index $i^\star \in \{L+1, \ldots,L+\Delta f / \beta\}$ such that $I_{i^\star}(\beta) \cap D^{\mathrm{M}}_{i} \neq \emptyset$ for some $i \in \mathbb{Z}$. We then have $\underline{c}_{i^{\star}}(\beta) = \inf_{x \in I_{i^\star}(\beta)} \{ c (x) \} = 0$ since $c(x) = 0$ for $x \in D^{\mathrm{M}}_{i}$. As for the second step, we note that for any $(\varphi, A) \in \mathcal{E}(L, \beta)$, the left-hand side of the constraint in~\eqref{dual_form_to_refer} satisfies
\begin{align*}
    \displaystyle \sum\limits_{i \in [\pm (L + \Delta f / \beta)]} \indicate{I_i(\beta) \subseteq A} \cdot p(i) = \indicate{I_{i^\star}(\beta) \subseteq A} \; = \; 0 ,
\end{align*}
where the first equality holds by construction of $p$ and the second equality uses the fact that $I_{i^\star}(\beta) \not \subseteq A$ since $i^\star \in \{ L + 1,\ldots, L +\Delta f / \beta \}$ whereas $A \subseteq \bigcup_{i \in [\pm L]} I_i (\beta)$ because $A \in \mathcal{F}(L, \beta)$. Since the right-hand side of the constraint in~\eqref{dual_form_to_refer} is non-negative by construction, we thus conclude that $p$ is feasible in~\eqref{dual_form_to_refer}.
}

\section{Proofs of Section~\ref{sec:nonlinear}}

{\color{black}
\subsection{Proof of Observation~\ref{observation_primal_nonlinear}}

Assumption~\ref{assumptions_f} allows us to replace the DP constraint in~\eqref{generalized_problem} with
\begin{align*}
    &\displaystyle \int_{x \in \mathbb{R}}\indicate{x \in A} \diff \gamma(x\mid f(D))  \leq e^\varepsilon \cdot \displaystyle \int_{x \in \mathbb{R}} \indicate{f(D') - f(D) + x \in A} \diff \gamma(x \mid f(D')) + \delta \\
    & \pushright{\forall (D, D') \in \mathcal{N}, \ \forall A \in \mathcal{F}} \\
    \iff&\displaystyle \int_{x \in \mathbb{R}}\indicate{x \in A} \diff \gamma(x \mid \phi)  \leq  e^\varepsilon \cdot \displaystyle \int_{x \in \mathbb{R}} \indicate{x + \varphi \in A} \diff \gamma(x\mid\phi+\varphi ) + \delta \\
    & \pushright{\forall \phi \in \Phi, \ \forall (\varphi, A) \in \mathcal{E}'(\phi).}
\end{align*}
Here, the first line holds since $\{A : A \in \mathcal{F}\} = \{ A + f(D) : A \in \mathcal{F} \}$ for any $D \in {\color{black}\mathcal{D}}$, whereas the second line is due to Assumption~\ref{assumptions_f}. Replacing the latter representation of the DP constraints with those in~\eqref{generalized_problem} gives~\ref{problem:integral_nonlinear} and thus concludes the observation. \qed
}

\subsection{Proof of Proposition~\ref{prop:general_bounded}}\label{app:prop:general_bounded}

The proof of Proposition~\ref{prop:general_bounded} relies on three auxiliary lemmas, each of which is devoted to the inclusion of one of the restrictions~\eqref{additional_measure}, \eqref{restriction_general} and~\eqref{eq:general_bound_sup} into problem~\ref{problem:integral_nonlinear}. We state and prove these auxiliary lemmas first.

\begin{lemmaA}\label{prop:temp_ub}
With the additional constraint~\eqref{additional_measure}, \ref{problem:integral_nonlinear} has the same optimal value as
\begin{align}\label{problem:temp_ub}
    \begin{array}{cl}
    \underset{\gamma}{\mathrm{minimize}}  & \displaystyle \beta \cdot \sum_{k \in [K]} \objweight_k(\beta) \cdot \int_{x \in \mathbb{R}} c(x) \diff \gamma_k(x) \\
    \mathrm{subject\; to}    &    \gamma_k \in \mathcal{P}_0, \ k \in [K] \\
    & \displaystyle \int_{x \in \mathbb{R}}\indicate{x \in A} \diff \gamma_k(x)  \leq  e^\varepsilon \cdot \displaystyle \int_{x \in \mathbb{R}} \indicate{x + \varphi \in A} \diff \gamma_m(x) + \delta \\
    & \pushright{\qquad \qquad \qquad \qquad \forall k,m \in [K], \ \forall (\varphi,A) \in \mathcal{E}''_{km}(\beta)},
    \end{array}
\end{align}
where $\mathcal{E}''_{km}(\beta) := [[-\Delta f , \Delta f] \cap (\Phi_m(\beta) - \Phi_k(\beta))] \times \mathcal{F}$ and $\objweight_k(\beta) := \beta^{-1} \cdot \int_{\phi \in\Phi_k(\beta)} \objweight(\phi) \diff \phi$.
\end{lemmaA}

\begin{proof}
We use restriction~\eqref{additional_measure} to replace the uncountable family of measures $\{ \gamma (\cdot | \phi) \}_{\phi \in \Phi}$ in problem \ref{problem:integral_nonlinear} with the finite set of measures $\{ \gamma_k \}_{k \in [K]}$. Note that in this case, the requirement $\gamma \in \Gamma$ simplifies to $\gamma_k \in \mathcal{P}_0$ for all $k \in [K]$.

We can now equivalently reformulate the objective function of problem~\ref{problem:integral_nonlinear} as
\begin{align*}
    \int_{\phi \in \Phi} \objweight(\phi) \cdot \left[\int_{x \in \mathbb{R}} c(x) \diff \gamma(x \mid \phi)\right]\diff \phi  &= 
    \sum_{k \in [K]} \int_{\phi \in \Phi_k(\beta)} \objweight(\phi) \cdot \Big[\int_{x \in \mathbb{R}} c(x) \diff \underbrace{\gamma(x \mid \phi)}_{= \gamma_k(x)}\Big]\diff \phi \\
    & = \sum_{k \in [K]} \Big[ \underbrace{\int_{\phi \in \Phi_k(\beta)} \objweight(\phi) \diff \phi}_{=\beta \cdot \objweight_k(\beta)} \Big]\cdot \Big[\int_{x \in \mathbb{R}} c(x) \diff \gamma_k(x)\Big],
\end{align*}
which coincides with the objective function of~\eqref{problem:temp_ub}.

Next, fix any DP constraint $(\phi, \varphi, A)$ in~\ref{problem:integral_nonlinear}, and note that $(\phi, \phi + \varphi) \in \Phi_k(\beta) \times \Phi_m(\beta)$ for a unique pair $(k, m) \in [K]$. In terms of our new measures $\{ \gamma_k \}_{k \in [K]}$, the constraint becomes
\begin{equation*}
    \int_{x \in \mathbb{R}}\indicate{x \in A} \diff \gamma_k(x)  \leq  e^\varepsilon \cdot \displaystyle \int_{x \in \mathbb{R}} \indicate{x + \varphi \in A} \diff \gamma_m(x) + \delta,
\end{equation*}
and this constraint is indeed included in~\eqref{problem:temp_ub} since $\varphi \in [-\Delta f, \Delta f] \cap (\Phi_m(\beta) - \Phi_k(\beta))$. Similarly, one readily verifies that all DP constraints in~\eqref{problem:temp_ub} have corresponding constraints in problem~\ref{problem:integral_nonlinear}. We thus conclude that~\ref{problem:integral_nonlinear} and~\eqref{problem:temp_ub} are indeed equivalent under restriction~\eqref{additional_measure}, as desired.
\end{proof}

In contrast to~\ref{problem:integral_nonlinear}, problem~\eqref{problem:temp_ub} comprises finitely many probability measures $\{ \gamma_k \}_{k \in [K]}$. The next result shows that the restriction~\eqref{restriction_general} allows us to equivalently represent each measure $\gamma_k$ by countably many decision variables.

\begin{lemmaA}\label{prop:discretize}
With the additional constraint~\eqref{restriction_general}, problem~\eqref{problem:temp_ub} has the same optimal value as
\begin{align}\label{problem:general_restricted_main_finite}\tag{$\mathrm{P'}(\beta)$}
    \begin{array}{cl}
        \underset{p}{\mathrm{minimize}} & \displaystyle \beta \cdot\sum\limits_{k \in [K]} \objweight_k(\beta) \cdot  \sum\limits_{i \in \mathbb{Z}} c(i) \cdot p_{k}(i) \\[1.5em]
        \mathrm{subject\; to} &\displaystyle p_{k} : \mathbb{Z} \mapsto \mathbb{R}_{+}, \  \sum\limits_{i \in \mathbb{Z}} p_{k}(i) = 1, \ k \in [K] \\[1.5em]
        &\displaystyle \sum\limits_{i \in \mathbb{Z}} \indicate{I_i(\beta) \subseteq A} \cdot p_{k}(i) \leq  e^\varepsilon \cdot \sum\limits_{i \in \mathbb{Z}} \indicate{I_i(\beta) + \varphi \subseteq A} \cdot p_{m}(i) + \delta \\
        & \pushright{\qquad \quad \forall k,m \in [K], \ \forall (\varphi, A) \in \mathcal{E}'_{km}(\beta),}
    \end{array}
\end{align}
where {\color{black}$\mathcal{E}'_{km}(\beta) := [\setvarphi(\beta) \cap \{(m-k-1)\cdot \beta,\ (m-k)\cdot \beta, \ (m-k+1)\cdot \beta \}]\times \mathcal{F}(\beta)$}.
\end{lemmaA}

\begin{proof}
Similar arguments as in the proof of Lemma~\ref{lemma_beginning} allow us to replace each measure $\gamma_k$ with a countable set of decision variables $p_k : \mathbb{Z} \mapsto \mathbb{R}_+$, $k \in [K]$. The resulting reformulation of problem~\eqref{problem:temp_ub} under restriction~\eqref{restriction_general} reads as follows.
\begin{align*}
    \begin{array}{cl}
        \underset{p}{\mathrm{minimize}} & \displaystyle \beta \cdot\sum\limits_{k \in [K]} \objweight_k(\beta) \cdot  \sum\limits_{i \in \mathbb{Z}} c(i) \cdot p_{k}(i) \\[1.5em]
        \mathrm{subject\; to} &\displaystyle \sum\limits_{i \in \mathbb{Z}} p_{k}(i) = 1, \ p_{k} : \mathbb{Z} \mapsto \mathbb{R}_{+}, \  k \in [K] \\[1.5em]
        &\displaystyle \sum_{i \in \mathbb{Z}} p_{k}(i) \cdot \dfrac{|A \cap I_i(\beta)|}{\beta} \leq e^\varepsilon \cdot \sum_{i \in \mathbb{Z}} p_{m}(i) \cdot \dfrac{|(A - \varphi) \cap I_i(\beta)|}{\beta} + \delta \\
        & \pushright{\qquad \quad \forall k,m \in [K], \ \forall (\varphi, A) \in \mathcal{E}''_{km}(\beta)}
    \end{array}
\end{align*}
In contrast to~\ref{problem:general_restricted_main_finite}, this problem still employs the larger constraint set $(\varphi, A) \in \mathcal{E}''_{km}(\beta)$. Applying similar arguments as in the proof of Lemma~\ref{lemma:worst-fixed} shows that the constraints $(\varphi, A)$ satisfying $A \in \mathcal{F} \setminus \mathcal{F} (\beta)$ are weakly dominated by the constraints $(\varphi, A')$ satisfying $A' \in \mathcal{F} (\beta)$, and arguments similar to those in the proof of Lemma~\ref{lemma:linear} show that the constraints $(\varphi, A)$ satisfying $A \in \mathcal{F} (\beta)$ and $\varphi \in [-\Delta f, \Delta f] \cap (\Phi_m(\beta) - \Phi_k(\beta))$ are weakly dominated by the constraints $(\varphi', A)$ satisfying $\varphi' \in \setvarphi(\beta) \cap \mathrm{cl}(\Phi_m(\beta) - \Phi_k(\beta))$ whenever $[-\Delta f, \Delta f] \cap (\Phi_m(\beta) - \Phi_k(\beta))$ is nonempty.
We can thus replace in the above optimization problem the constraints $(\varphi, A) \in \mathcal{E}''_{km}(\beta)$ with the smaller set of constraints $(\varphi, A) \in \mathcal{E}'_{km}(\beta)$. In that case, however, the identities
\begin{equation*}
    \frac{|A \cap I_i(\beta)|}{\beta} = \indicate{I_i(\beta) \subseteq A}
    \quad \text{and} \quad
    \frac{|(A - \varphi)\cap I_i(\beta)|}{\beta} = \indicate{I_i(\beta) + \varphi \subseteq A}
\end{equation*}
hold for all $i \in \mathbb{Z}$ (\emph{cf.}~the proof of Lemma~\ref{prop:finite}), which concludes the proof.
\end{proof}

Problem~\ref{problem:general_restricted_main_finite} still comprises a countably infinite number of decision variables and uncountably many constraints. The next result shows that the restriction~\eqref{eq:general_bound_sup} allows us to reformulate \ref{problem:general_restricted_main_finite} as a finite-dimensional linear program.

\begin{lemmaA}\label{final_lemma}
With the additional constraint~\eqref{eq:general_bound_sup},~\ref{problem:general_restricted_main_finite} has the same optimal value as the finite-dimensional linear program~\ref{general_L-UB}.
\end{lemmaA}

\begin{proof}
Under restriction~\eqref{eq:general_bound_sup}, we can reduce each countable set of decisions $p_k : \mathbb{Z} \mapsto \mathbb{R}_+$ to a finite set $p_k : [\pm L] \mapsto \mathbb{R}_+$, $k \in [K]$. Moreover, similar arguments as in the proof of Proposition~\ref{prop:bounded} allow us to show that in the resulting problem, each constraint $(\varphi, A) \in \mathcal{E}'_{km}(\beta) \setminus \mathcal{E}'_{km}(L, \beta)$ is weakly dominated by the a constraint $(\varphi, A_L) \in \mathcal{E}'_{km}(L, \beta)$. This concludes the proof.
\end{proof}

\noindent \textbf{Proof of Proposition~\ref{prop:general_bounded}.} $\;$
The proof directly follows from Lemmas~\ref{prop:temp_ub},~\ref{prop:discretize} and~\ref{final_lemma}.
\qed

\subsection{Proof of Proposition~\ref{prop:weakdual_general}}\label{app:prop:weakdual_general}

Fix any $\gamma$ feasible in \ref{problem:integral_nonlinear} and any $(\theta, \psi)$ feasible in \ref{problem:dual_integral_nonlinear}. We then have
\begin{align*}
    & \displaystyle \int_{\phi \in \Phi} \objweight(\phi) \cdot \left[ \int_{x \in \mathbb{R}} c(x)  \diff \gamma(x \mid \phi) \right] \diff \phi \\
=  &   \displaystyle \int_{\phi \in \Phi}  \int_{x \in \mathbb{R}} \Big[ \objweight(\phi) \cdot c(x) \Big] \diff \gamma(x \mid \phi)  \diff \phi\\
\geq  &   \displaystyle \int_{\phi \in \Phi}  \int_{x \in \mathbb{R}} \Bigg[ \theta(\phi) -  \int_{(\varphi, A)\in \mathcal{E}'(\phi)} \indicate{x \in A} \diff \psi( \varphi, A \mid \phi) + \\
& \hspace{2.2cm} e^\varepsilon \cdot \int_{(-\varphi, A)\in \mathcal{E}'(\phi)} \indicate{x+\varphi \in A}\diff \psi(\varphi, A \mid \phi - \varphi) \Bigg] \diff \gamma(x \mid \phi)  \diff \phi \\
= & \displaystyle \underbrace{\int_{\phi \in \Phi} \int_{x \in \mathbb{R}} \theta(\phi) \diff \gamma(x \mid \phi) \diff \phi}_{\text{\emph{(i)}}} - \underbrace{\int_{\phi \in \Phi} \int_{x \in \mathbb{R}} \int_{(\varphi, A)\in \mathcal{E}'(\phi)} \indicate{x \in A} \diff \psi( \varphi, A \mid \phi) \diff \gamma(x \mid \phi) \diff \phi}_{\text{\emph{(ii)}}}  \\
& \quad + e^\varepsilon \cdot \underbrace{\int_{\phi \in \Phi} \int_{x \in \mathbb{R}} \int_{(-\varphi, A)\in \mathcal{E}'(\phi)} \indicate{x+\varphi \in A}\diff \psi(\varphi, A \mid \phi - \varphi) \diff \gamma(x \mid \phi) \diff \phi}_{\text{\emph{(iii)}}},
\end{align*}
where the inequality follows from the constraints in \ref{problem:dual_integral_nonlinear} and the fact that $\gamma(\cdot|\phi)$ is a non-negative measure, and the final equality is due to the linearity of the integration operator.

The above term \emph{(i)} simplifies to
\begin{align*}
    \int_{\phi \in \Phi} \int_{x \in \mathbb{R}} \theta(\phi) \diff \gamma(x \mid \phi) \diff \phi = \int_{\phi \in \Phi}\theta(\phi) \left[\int_{x \in \mathbb{R}}  \diff \gamma(x \mid \phi)\right] \diff \phi =  \int_{\phi \in \Phi} \theta(\phi) \diff \phi,
\end{align*}
where we used the fact that $\gamma(\cdot | \phi)$ is a probability measure for every $\phi \in \Phi$.

The above term \emph{(ii)} can be reformulated as
\begin{align*}
    & \int_{\phi \in \Phi} \int_{x \in \mathbb{R}} \int_{(\varphi, A)\in \mathcal{E}'(\phi)} \indicate{x \in A} \diff \psi( \varphi, A \mid \phi) \diff \gamma(x \mid \phi) \diff \phi \\
=    & \int_{\phi \in \Phi} \int_{(\varphi, A)\in \mathcal{E}'(\phi)} \left[\int_{x \in \mathbb{R}} \indicate{x \in A} \diff \gamma(x \mid \phi) \right] \diff \psi( \varphi, A \mid \phi) \diff \phi,
\end{align*}
where we used Fubini's theorem (whose applicability follows from similar arguments as in the proof of Proposition~\ref{prop:weakdual}) to change the order of integration.

Finally, the above term \emph{(iii)} can be rewritten as
\begin{align*}
    & \int_{\phi \in \Phi} \int_{x \in \mathbb{R}} \int_{(-\varphi, A)\in \mathcal{E}'(\phi)} \indicate{x+\varphi \in A}\diff \psi(\varphi, A \mid \phi - \varphi) \diff \gamma(x \mid \phi) \diff \phi \\
=   & \int_{\phi \in \Phi} \int_{(-\varphi, A)\in \mathcal{E}'(\phi)} \left[\int_{x \in \mathbb{R}} \indicate{x+\varphi \in A}  \diff \gamma(x \mid \phi) \right] \diff \psi(\varphi, A \mid \phi - \varphi) \diff \phi \\
=   & \int_{\phi \in \Phi} \int_{(\varphi, A)\in \mathcal{E}'(\phi)} \left[\int_{x \in \mathbb{R}} \indicate{x+\varphi \in A}  \diff \gamma(x \mid \phi + \varphi)\right] \diff \psi(\varphi, A \mid \phi) \diff \phi.
\end{align*}
where the first equality follows from Fubini's theorem. The second equality is due to a change of variables. Specifically, we use the definition of $\mathcal{E}'(\phi)$ to rewrite the region of integration as
\begin{equation*}
    \left\{(\phi, \varphi, A) \ :\ \phi \in \Phi, (-\varphi, A)\in \mathcal{E}'(\phi) \right\} = \left\{(\phi, \varphi, A) \in \Phi \times [-\Delta f, \Delta f] \times \mathcal{F} \ : \phi - \varphi \in  \Phi \right\}.
\end{equation*}
Introducing a new variable $\phi' = \phi - \varphi$, we observe that $\varphi + \phi' = \phi \in \Phi$. Hence the region of integration region can be expressed as
\begin{equation*}
    \left\{(\phi', \varphi, A) \in \Phi \times [-\Delta f, \Delta f] \times \mathcal{F} \ : \varphi \in  \Phi - \phi' \right\} = \left\{(\phi', \varphi, A) \ :\ \phi' \in \Phi, (\varphi, A)\in \mathcal{E}'(\phi') \right\}
\end{equation*}
if we replace $\phi$ with $\phi' + \varphi$ in the integrals. The second equality now holds if we relabel $\phi'$ as $\phi$.

Replacing the terms \emph{(i)}--\emph{(iii)} with their equivalent expressions derived above, we obtain
\begin{align*}
& \displaystyle \int_{\phi \in \Phi} \objweight(\phi) \cdot \left[ \int_{x \in \mathbb{R}} c(x)  \diff \gamma(x \mid \phi) \right] \diff \phi \\
\geq & \int_{\phi \in \Phi} \theta(\phi) \diff \phi \ -  \int_{\phi \in \Phi} \int_{(\varphi, A)\in \mathcal{E}'(\phi)} \Bigg[\int_{x \in \mathbb{R}} \indicate{x \in A} \diff \gamma(x \mid \phi) -  \\
 & \hspace{6cm} e^\varepsilon \cdot \int_{x \in \mathbb{R}} \indicate{x+\varphi \in A} \diff \gamma(x \mid \phi + \varphi) \Bigg] \diff \psi( \varphi, A \mid \phi) \diff \phi \\
 \geq & \int_{\phi \in \Phi} \left[\theta(\phi) -  \delta \cdot \int_{(\varphi, A) \in \mathcal{E}'(\phi)}\diff \psi (\varphi, A \mid \phi) \right] \diff \phi,
\end{align*}
where the final inequality is due to the constraints in \ref{problem:integral_nonlinear} and the fact that $\psi(\cdot| \phi)$ is a non-negative measure. The last expression coincides with the objective function of \ref{problem:dual_integral_nonlinear}, as desired.
\qed

\subsection{Proof of Proposition~\ref{prop:generalized_finite_dual}}\label{app:prop:generalized_finite_dual}

Our proof proceeds in two steps. We first use restriction~\eqref{general_dual_23} to reduce the number of decision variables in \ref{problem:dual_integral_nonlinear}, and we subsequently use restriction~\eqref{general_dual_1} to remove the integrals as well as reduce the number of constraints in \ref{problem:dual_integral_nonlinear}. This will yield the formulation in the statement of Proposition~\ref{prop:generalized_finite_dual}.

In view of the first step, we use restriction~\eqref{general_dual_23} to replace the measure $\theta : \Phi \mapsto \mathbb{R}$ with a vector $\bm{\theta} \in \mathbb{R}^K$ and $\psi$ with a finite family of unconditional measures $\psi_k \in \mathcal{M}_{+}(\mathcal{E}'(\phi))$, $k \in [K]$. Under those substitutions, problem \ref{problem:dual_integral_nonlinear} simplifies to the following formulation.
\begin{align*}
    \begin{array}{cl}
    \underset{\theta, \psi}{\text{maximize}}  & \displaystyle \beta \cdot \sum_{k \in [K]} \theta_k - \delta \cdot \sum_{k \in [K]} \int_{\phi \in \Phi_k(\beta)} \int_{(\varphi, A) \in \mathcal{E}'(\phi)} \diff \psi_k(\varphi, A) \diff \phi \\
    \text{subject to}    & \bm{\theta} \in \mathbb{R}^{K}, \ \psi_k \in \mathcal{M}_{+}(\mathcal{E}'(\phi)), \ k \in [K] \\
    & \displaystyle \theta_k \leq \int_{(\varphi, A)\in \mathcal{E}'(\phi)} \indicate{x \in A} \diff \psi_k( \varphi, A) - \\
    & \displaystyle \mspace{40mu} e^\varepsilon \cdot \sum_{m \in [K]} \int_{(-\varphi, A)\in \mathcal{E}'(\phi)} \indicate{x+\varphi \in A} \cdot \indicate{\phi - \varphi \in \Phi_m(\beta)} \diff \psi_m (\varphi, A) \\
    & \pushright{ + c(x) \cdot \objweight(\phi)  \quad \forall k \in [K], \ \forall \phi \in \Phi_k(\beta), \ \forall x \in \mathbb{R}}
    \end{array}
\end{align*}
Here, our reformulation uses the fact that $\{ \Phi_k(\beta) \}_{k \in [K]}$ partitions $\Phi$, and the reformulated first term of the objective function additionally exploits that $| \Phi_k(\beta) | = \beta$ for all $k \in [K]$.

As for the second step, we note that under restriction~\eqref{general_dual_1}, the second expression in the objective function simplifies to
\begin{align*}
    \displaystyle \sum_{k \in [K]} \int_{\phi \in \Phi_k(\beta)} \int_{(\varphi, A) \in \mathcal{E}'(\phi)} \diff \psi_k(\varphi, A) \diff \phi
    & = \sum_{k \in [K]} \int_{\phi \in \Phi_k(\beta)} \int_{(\varphi, A) \in \mathcal{E}_{k}'(L,\beta)} \diff \psi_k(\varphi, A) \diff \phi \\
    & = \beta \cdot \sum_{k \in [K]} \int_{(\varphi, A) \in \mathcal{E}_{k}'(L,\beta)} \diff \psi_k(\varphi, A),
\end{align*}
where the first equality is due to restriction~\eqref{general_dual_1} and the fact that $\mathcal{E}'(\phi) \cap \mathcal{E} (L, \beta) = \mathcal{E}'_{k}(L,\beta)$, and the second equality follows from taking the inner integral outside (as it is not parameterized by $\phi$) and from $|\Phi_k(\beta)| = \beta$. The resulting objective function coincides with that of problem~\ref{general_lb-L-beta}. For any fixed $k \in [K]$, $\phi \in \Phi_k(\beta)$ and $x \in \mathbb{R}$, the first integral in the constraints simplifies to
\begin{equation*}
    \int_{(\varphi, A)\in \mathcal{E}'(\phi)} \indicate{x \in A} \diff \psi_k( \varphi, A) = \int_{(\varphi, A)\in \mathcal{E}'_k (L, \beta)} \indicate{x \in A} \diff \psi_k( \varphi, A).
\end{equation*}
Likewise, the second integral simplifies to
\begin{align*}
    & \sum_{m \in [K]} \int_{(-\varphi, A)\in \mathcal{E}'(\phi)} \indicate{x+\varphi \in A} \cdot \indicate{\phi - \varphi \in \Phi_m(\beta)} \diff \psi_m (\varphi, A) \\
    = & \sum_{m \in [K]} \int_{(-\varphi, A)\in \mathcal{E}'_{k}(L, \beta)} \indicate{x+\varphi \in A} \cdot \indicate{\phi - \varphi \in \Phi_m(\beta)} \diff \psi_m (\varphi, A) \\
    = & \int_{(-\varphi, A)\in \mathcal{E}'_k(L, \beta)} \indicate{x+\varphi \in A} \diff \psi_{k - \varphi / \beta} (\varphi, A),
\end{align*}
where the first equality is due to restriction~\eqref{general_dual_1} and the second equality exploits the fact that for $\phi \in \Phi_k (\beta)$ and $-\varphi \in \setvarphi (\beta)$, we have $\phi - \varphi \in \Phi_m(\beta)$ if and only if $m = k - \varphi / \beta$. In summary, the constraints simplify to
\begin{align*}
    & \displaystyle \theta_k \leq \int_{(\varphi, A) \in \mathcal{E}'_{k}(L,\beta)} \indicate{x \in A} \diff \psi_k( \varphi, A) -e^\varepsilon \cdot \int_{(-\varphi, A)\in \mathcal{E}'_{k}(L,\beta)} \indicate{x + \varphi \in A}\diff \psi_{k - \varphi/\beta}(\varphi, A) \\ 
    & \pushright{+ c(x) \cdot \objweight(\phi)\quad \forall k \in [K], \ \forall \phi \in \Phi_k(\beta), \ \forall x \in \mathbb{R}.}
\end{align*}
Note that the index $\phi \in \Phi_k(\beta)$ only affects the last term in this constraint, and the constraint is thus equivalent to
\begin{align*}
    & \displaystyle \theta_k \leq \int_{(\varphi, A) \in \mathcal{E}'_{k}(L,\beta)} \indicate{x \in A} \diff \psi_k( \varphi, A) -e^\varepsilon \cdot \int_{(-\varphi, A)\in \mathcal{E}'_{k}(L,\beta)} \indicate{x + \varphi \in A}\diff \psi_{k - \varphi/\beta}(\varphi, A) \\ 
    & \pushright{+ c(x) \cdot \underline{\objweight}_k(\beta)\quad \forall k \in [K], \ \forall x \in \mathbb{R}.}
\end{align*}
We can equivalently express the index $x \in \mathbb{R}$ in this constraint with the double index $(i, x) \in \mathbb{Z} \times I_i(\beta)$, and arguments similar to those in the proof of Lemma~\ref{prop:lb-beta} show that the constraint subsequently simplifies to
\begin{align*}
    & \displaystyle \theta_k \leq \int_{(\varphi, A) \in \mathcal{E}'_{k}(L,\beta)} \indicate{I_i(\beta) \subseteq A} \diff \psi_k( \varphi, A) - e^\varepsilon \cdot \int_{(-\varphi, A)\in \mathcal{E}'_{k}(L,\beta)} \indicate{I_i(\beta) + \varphi \subseteq A}\diff \psi_{k - \varphi/\beta}(\varphi, A)  \\ 
    & \pushright{\displaystyle + \underline{c}_i(\beta) \cdot \underline{\objweight}_k(\beta) \quad \forall k \in [K], \ \forall i \in \mathbb{Z}.}
\end{align*}
Similar arguments as in the proof of Proposition~\ref{prop:lb-L-beta} allow us to further restrict the index $i \in \mathbb{Z}$ in the above constraint to $i \in [\pm (L + \Delta f / \beta)]$. The restriction~\eqref{general_dual_1} also allows us to replace the measures $\psi_k \in \mathcal{M}_{+}(\mathcal{E}'(\phi))$ with discrete measures $\psi_k: \mathcal{E}'_k(L,\beta) \mapsto \mathbb{R}_{+}$, $k \in [K]$, and replace the integrals in the objective function and the constraints with sums. This results in the formulation \ref{general_lb-L-beta} and thus concludes the proof.
\qed

\subsection{Proof of Theorem~\ref{thm:generalized_strong_dual}}\label{prop:thm:generalized_strong_dual}
We employ the same strategy as in the proof of Theorem~\ref{thm:strong_duality}. We define the auxiliary problem
\begin{align}\label{mid_L-UB}\tag{$\mathrm{M'}(\Lambda \cdot \ell,\Delta f / \ell)$}
    \begin{array}{cl}
        \underset{ p }{\mathrm{minimize}} & \displaystyle (\Delta f / \ell) \cdot  \sum\limits_{k \in [K]} \objweight_k(\Delta f / \ell) \cdot \Big[ \sum\limits_{i \in [\pm (\Lambda \cdot \ell + \ell)]} c_i(\Delta f / \ell) \cdot p_{k}(i) \Big] \\[1.5em]
         \mathrm{subject\; to} &\displaystyle  p_{k}: [\pm (\Lambda \cdot \ell + \ell)] \mapsto \mathbb{R}_{+}, \ \sum\limits_{i \in [\pm (\Lambda \cdot \ell + \ell)]} p_{k}(i) = 1, \ k \in [K] \\[1.5em]
        &\displaystyle \sum\limits_{i \in [\pm (\Lambda \cdot \ell + \ell)]} \indicate{I_i(\Delta f / \ell) \subseteq A} \cdot p_{k}(i) \leq e^\varepsilon \cdot \sum\limits_{i \in [\pm (\Lambda \cdot \ell + \ell)]} \indicate{I_i(\Delta f / \ell) + \varphi \subseteq A} \cdot p_{m}(i) + \delta \\
        & \pushright{\forall k,m \in [K], \ \forall (\varphi, A) \in \mathcal{E}'_{km}(\Lambda \cdot \ell,\Delta f / \ell)},
    \end{array}
\end{align}
and we show that the optimal values of \hyperref[{general_L-UB}]{$\mathrm{P'}(\Lambda \cdot \ell, \Delta f / \ell)$} and \hyperref[{general_lb-L-beta}]{$\mathrm{D'}(\Lambda \cdot \ell, \Delta f / \ell)$} converge to that of \ref{mid_L-UB} when $\ell$ increases (which, in return, refines the granularity $\Delta f / \ell$) and $\Lambda$ increases (which, in return, increases the support $[- \Lambda \cdot \Delta f,\; (\Lambda + 1/\ell)\cdot \Delta f)$). Note that the number $K$ of intervals in $\Phi$ depends on $\ell$ due to the partitioning $\Phi = \bigcup_{k \in [K]} \Phi_k(\Delta f / \ell)$.

To see that the optimal value of \hyperref[{general_L-UB}]{$\mathrm{P'}(\Lambda \cdot \ell, \Delta f / \ell)$} converges to that of \ref{mid_L-UB}, we first note that \ref{mid_L-UB} differs from \hyperref[{general_L-UB}]{$\mathrm{P'}(\Lambda \cdot \ell, \Delta f / \ell)$} only in the existence of the additional decision variables $p_k (i)$, $i \in [\pm (\Lambda \cdot \ell + \ell)] \setminus [\pm \Lambda \cdot \ell]$, which also implies that  $\text{\hyperref[{general_L-UB}]{$\mathrm{P'}(\Lambda \cdot \ell, \Delta f / \ell)$}} \geq \text{\ref{mid_L-UB}}$. Using similar arguments as in the proof of Lemma~\ref{lemma:bound_tail}, we can show that for any $\varepsilon > 0$, $\delta > 0$ and $\tau > 0$, there is $\Lambda' \in \mathbb{N}$ such that any optimal solution $p^\star$ to~\ref{mid_L-UB} satisfies $\sum_{i \in [\pm (\Lambda \cdot \ell + \ell)] \setminus [\pm \Lambda \cdot \ell]} p_k^\star(i) < \tau$ for all $k\in[K]$, $\ell \in \mathbb{N}$ and $\Lambda \geq \Lambda'$. Similar arguments as in the proof of Lemma~\ref{lemma:convergence1} then allow us to show that there is $\Lambda' \in \mathbb{N}$ such that $\text{\hyperref[{general_L-UB}]{$\mathrm{P'}(\Lambda' \cdot \ell, \Delta f / \ell)$}}- \text{\hyperref[{mid_L-UB}]{$\mathrm{M'}(\Lambda' \cdot \ell, \Delta f / \ell)$}} \leq \xi$ for all $\ell \in \mathbb{N}$. For the remainder of the proof, we fix such a value of $\Lambda'$.

To see that the optimal value of \hyperref[{general_lb-L-beta}]{$\mathrm{D'}(\Lambda' \cdot \ell, \Delta f / \ell)$} converges to that of \text{\hyperref[{mid_L-UB}]{$\mathrm{M'}(\Lambda' \cdot \ell, \Delta f / \ell)$}} , on the other hand, we note that~\text{\hyperref[{mid_L-UB}]{$\mathrm{M'}(\Lambda' \cdot \ell, \Delta f / \ell)$}} differs from the strong dual of \hyperref[{general_lb-L-beta}]{$\mathrm{D'}(\Lambda' \cdot \ell, \Delta f / \ell)$} essentially only in the objective coefficients, which change from $\underline{\objweight}_k(\Delta f / \ell)$ and $\underline{c}_i(\Delta f / \ell)$ in the strong dual of \hyperref[{general_lb-L-beta}]{$\mathrm{D'}(\Lambda' \cdot \ell, \Delta f / \ell)$} to $\objweight_k(\Delta f / \ell)$ and $c_i(\Delta f / \ell)$ in \text{\hyperref[{mid_L-UB}]{$\mathrm{M'}(\Lambda' \cdot \ell, \Delta f / \ell)$}}, respectively. Similar arguments as in the proof of Lemma~\ref{lemma:convergence2} show that for any $\varepsilon > 0$, $\delta > 0$ and $\xi > 0$, there exists $\ell' \in \mathbb{N}$ such that $\hyperref[{mid_L-UB}]{\mathrm{M'}(\Lambda' \cdot \ell,\Delta f / \ell)} - \hyperref[{general_lb-L-beta}]{\mathrm{D'}(\Lambda' \cdot \ell,\Delta f / \ell)} \leq \xi$ for all $\ell \geq \ell'$. Here, the uniform continuity of $c$ ensures the convergence of the terms $c_i(\Delta f / \ell)$ and $\underline{c}_{i}(\Delta f / \ell)$, while our earlier assumption that $\objweight$ is a continuous probability density function ensures the convergence of the terms $\objweight_k(\Delta f / \ell)$ and $\underline{\objweight}_k(\Delta f / \ell)$. Moreover, since $\{ \objweight_k(\Delta f / \ell) \}_\ell$ and $\{ \underline{\objweight}_k(\Delta f / \ell) \}_\ell$ are non-negative and sum up to at most $1$, the overall objective functions of both problem converge despite the growing number $K$ of subsets of $\Phi$.

So far, we have shown that there exists $\Lambda' \in \mathbb{N}$ and $\ell' \in \mathbb{N}$ such that $\text{\hyperref[{general_L-UB}]{$\mathrm{P'}(\Lambda' \cdot \ell, \Delta f / \ell)$}} - \text{\hyperref[{general_lb-L-beta}]{$\mathrm{D'}(\Lambda' \cdot \ell, \Delta f / \ell)$}} \leq \xi$ holds for all $\ell \geq \ell'$. Since we have $\hyperref[{general_L-UB}]{\mathrm{P'}(\Lambda' \cdot \ell, \Delta f / \ell)} \geq \hyperref[{general_L-UB}]{\mathrm{P'}(\Lambda \cdot \ell, \Delta f / \ell)}$ and $\hyperref[{general_lb-L-beta}]{\mathrm{D'}(\Lambda' \cdot \ell, \Delta f / \ell)} \leq \hyperref[{general_lb-L-beta}]{\mathrm{D'}(\Lambda \cdot \ell, \Delta f / \ell)}$ for all $\Lambda \geq \Lambda'$, we can conclude the proof. Further details of this proof are relegated to the GitHub supplement.

\section{Proofs of Section~\ref{sec:algo}}\label{app:sec:algo}

\subsection{Proof of Corollary~\ref{corr:nonuniform}}

\ref{new_L-UB} and \ref{new_LB} sandwich \ref{problem:integral_main} and \ref{problem:integral_dual} from above and below since \ref{new_L-UB} and \ref{new_LB} constitute restrictions of the earlier problems \ref{L-UB} and \ref{lb-L-beta} that satisfy the same inequality (\emph{cf.}~Lemmas~\ref{prop:finite} and~\ref{prop:lb-beta} as well as Propositions~\ref{prop:bounded} and~\ref{prop:lb-L-beta}).

In view of the second part of the statement, recall from Theorem~\ref{thm:strong_duality} that there is $\Lambda' \in \mathbb{N}$ and $k' \in \mathbb{N}$ such that $\text{\hyperref[{L-UB}]{$\mathrm{P}(\Lambda \cdot k, \Delta f / k)$}} - \text{\hyperref[{lb-L-beta}]{$\mathrm{D}(\Lambda \cdot k, \Delta f / k)$}} \leq \xi$ holds for all $\Lambda \geq \Lambda'$ and $k \geq k'$. Moreover, we have $\text{\ref{new_L-UB}} \leq \text{\hyperref[{L-UB}]{$\mathrm{P}(\Lambda \cdot k, \Delta f / k)$}}$ if $\{\Pi_j(\beta)\}_{j \in [N]}$ is a refinement of $\{I_i(\Delta f / k)\}_{i \in [\pm \Lambda \cdot k]}$. Indeed, $\text{\hyperref[{L-UB}]{$\mathrm{P}(\Lambda \cdot k, \Delta f / k)$}}$ is equivalent to a variant of \ref{new_L-UB} that includes the additional constraints $p(j) = p(j')$ for all $j, j' \in [N]$ satisfying $\Pi_j(\beta), \Pi_{j'}(\beta) \subseteq I_i(\Delta f / k)$ for some $i \in [\pm \Lambda \cdot k]$. A similar argument shows that $\text{\ref{new_LB}} \geq \text{\hyperref[{lb-L-beta}]{$\mathrm{D}(\Lambda \cdot k, \Delta f / k)$}}$, which concludes the proof.

\subsection{Proof of Proposition~\ref{alg:correct}}

We prove Proposition~\ref{alg:correct} via three auxiliary results. Lemma~\ref{lemma:algo_p1} proves that each inner loop over $j$ in Algorithm~\ref{alg:worst-event} determines a worst-case event $A \in \arg \max \{ V (\varphi, A) \, : \, A \in \mathcal{F} (L, \beta) \}$ for the query output difference $\varphi \in \setvarphi(\beta)$ fixed by the outer loop. Subsequently, Lemma~\ref{lemma:algo_p2} proves that each outer loop over $\varphi$ determines a maximally violated constraint $(\varphi, A)$, which concludes the correctness of Algorithm~\ref{alg:worst-event}. Finally, Lemma~\ref{lemma:algo_p3} shows that Algorithm~\ref{alg:worst-event} can be implemented such that it determines a maximally violated constraint $(\varphi, A)$ in time $\mathcal{O}(N^3)$.

\begin{lemmaA}\label{lemma:algo_p1}
    For any $\varphi \in \setvarphi(\beta)$ fixed by the outer loop of Algorithm~\ref{alg:worst-event}, the event $A$ constructed in the inner loop satisfies $A \in \arg \max \{ V (\varphi, A) \, : \, A \in \mathcal{F} (L, \beta) \}$.
\end{lemmaA}

\begin{proof}
%
Fix an arbitrary $\varphi \in \setvarphi(\beta)$ and recall that for any $A \in \mathcal{F}(L,\beta)$, the privacy shortfall $V (\varphi, A)$ can be expressed as
\begin{align*}
    & \displaystyle \sum\limits_{j \in [N]} p(j) \cdot \dfrac{|A \cap \Pi_j(\beta)|}{|\Pi_j(\beta)|} - e^\varepsilon \cdot \sum\limits_{j \in [N]} p(j)\cdot \dfrac{|A \cap (\Pi_j(\beta) + \varphi)|}{|\Pi_j(\beta)|} \\
    \mspace{-3mu} = & \displaystyle \sum_{i \in [\pm L]} \indicate{I_{i}(\beta) \subseteq A} \cdot \left[ 
    \sum\limits_{j \in [N]} p(j) \cdot \dfrac{|I_{i}(\beta) \cap \Pi_{j}(\beta)|}{|\Pi_{j}(\beta)|} - e^\varepsilon \cdot \sum\limits_{j \in [N]} p(j)\cdot \dfrac{|(I_{i}(\beta) - \varphi ) \cap \Pi_{j}(\beta)|}{|\Pi_{j}(\beta)|} 
    \right] \\
    \mspace{-3mu}= & \displaystyle \sum_{i \in [\pm L]} \indicate{I_{i}(\beta) \subseteq A} \cdot \beta \cdot \left[ 
    \sum\limits_{j \in [N]} p(j) \cdot \dfrac{\indicate{I_{i}(\beta) \subseteq \Pi_{j}(\beta)}}{|\Pi_{j}(\beta)|} - e^\varepsilon \cdot \sum\limits_{j \in [N]} p(j)\cdot \dfrac{\indicate{I_{i}(\beta) \subseteq (\Pi_{j}(\beta) + \varphi)}}{|\Pi_{j}(\beta)|}
    \right],
\end{align*}
where we disregard the constant $-\delta$ since it does not affect the relative order of privacy shortfalls across the constraints $(\varphi, A)$. Here, the first identity exploits that $A = \bigcup_{i \in \mathcal{L}} I_i(\beta)$ for some $\mathcal{L} \subseteq [\pm L]$. The second identity holds since $|I_i(\beta)| = \beta$ and $I_i(\beta)$ is either entirely contained in or intersection free with $\Pi_{j}(\beta)$ and $\Pi_{j}(\beta) + \varphi$, $i \in [\pm L]$ and $j \in [N]$. Thus, $I_i(\beta)$ must be contained in the worst-case event $A$ whenever
\begin{align}\label{eq:quant_positive}
    \sum\limits_{j \in [N]} p(j) \cdot \dfrac{\indicate{I_{i}(\beta) \subseteq \Pi_{j}(\beta)}}{|\Pi_{j}(\beta)|} - e^\varepsilon \cdot \sum\limits_{j' \in [N]} p(j')\cdot \dfrac{\indicate{I_{i}(\beta) \subseteq \Pi_{j'}(\beta) + \varphi}}{|\Pi_{j'}(\beta)|}
\end{align}
is strictly positive; $I_i(\beta)$ can be (but does not have to be) included in $A$ if the above quantity is zero; and it must not be contained in $A$ if the above quantity is negative. In the remainder, we fix $i \in [\pm L]$ and distinguish between two cases: \emph{(i)} there is no $j' \in [N]$ satisfying $I_i(\beta) \subseteq \Pi_{j'}(\beta) + \varphi$; and \emph{(ii)} there is $j' \in [N]$ satisfying $I_i(\beta) \subseteq \Pi_{j'}(\beta) + \varphi$.

In case \emph{(i)}, we can include $I_i(\beta)$ in the worst-case event $A$ since the second term in~\eqref{eq:quant_positive} vanishes, whereas the first term is always non-zero by construction. Note that the events $A_j$, $j \in [N]$, constructed in the first part of the inner loop of Algorithm~\ref{alg:worst-event} comprise precisely all intervals $I_i(\beta)$ falling under case \emph{(i)}.
 
In case \emph{(ii)}, the existence of $j' \in [N]$ satisfying $I_i(\beta) \subseteq \Pi_{j'}(\beta) + \varphi$ implies that  \eqref{eq:quant_positive} equals to $p(j) / |\Pi_{j}(\beta)| - e^\varepsilon \cdot p(j') / |\Pi_{j'}(\beta)|$ for some $j \in [N]$, and $I_i(\beta)$ should be included in the worst-case event if this quantity is positive. Note that the events $A_{jj'}$, $j, j' \in [N]$, constructed in the second part of the inner loop of Algorithm~\ref{alg:worst-event} comprise precisely all intervals $I_i(\beta)$ falling under case \emph{(ii)} that satisfy $p(j) / |\Pi_{j}(\beta)| - e^\varepsilon \cdot p(j') / |\Pi_{j'}(\beta)| > 0$.
\end{proof}

\begin{lemmaA}\label{lemma:algo_p2}
    Algorithm~\ref{alg:worst-event} returns a constraint $(\varphi, A)$ with maximum privacy shortfall.
\end{lemmaA}

\begin{proof}
    Recall that the DP constraints of~\ref{new_L-UB} are indexed by $(\varphi, A) \in \mathcal{E}(L, \beta) = \setvarphi(\beta) \times \mathcal{F}(L, \beta)$ and that Lemma~\ref{lemma:algo_p1} proved that for any fixed $\varphi \in \setvarphi(\beta)$, the inner loop of Algorithm~\ref{alg:worst-event} constructs a worst-case event $A$ associated with $\varphi$. We show in this proof that it is sufficient to consider the values $\varphi \in \setvarphi(\beta, \bm{\pi}) \subseteq \setvarphi(\beta)$, where $$\setvarphi(\beta, \bm{\pi}) := \{ (\pi_{j} - \pi_{j'})\cdot \beta \ : \ (\pi_{j} - \pi_{j'})\cdot \beta \in [- \Delta f, \Delta f] \text{ and } j,j' \in [N] \} \cup \{-\Delta f, \Delta f \},$$ which is what the outer loop of Algorithm~\ref{alg:worst-event} does.

    Our earlier arguments of this section have shown that for any $\varphi \in \setvarphi(\beta)$, the maximum privacy shortfall $\max \{ V (\varphi, A) \, : \, A \in \mathcal{F}(L, \beta) \}$ satisfies
    \begin{align*}
        \sum_{j \in [N]} |A_j(\varphi)| \cdot \dfrac{p(j)}{|\Pi_{j}(\beta)|} + \sum_{j, j' \in [N]} |A_{jj'}(\varphi)| \cdot \left[ \dfrac{p(j)}{|\Pi_{j}(\beta)|} - e^\varepsilon \cdot \dfrac{p(j')}{|\Pi_{j'}(\beta)|} \right]^{+} - \delta,
    \end{align*}
    where $[x]^{+} = \max\{0, \, x\}$ and the only quantities varying with $\varphi$ are
    \begin{align*}
        \mspace{-5mu}
        A_j(\varphi) = \Pi_j(\beta) \setminus [-L\cdot \beta + \varphi, (L+1)\cdot\beta + \varphi)
        \quad \text{and} \quad
        A_{jj'}(\varphi) = \Pi_{j}(\beta) \cap (\Pi_{j'}(\beta) + \varphi),
        \;\; j, j' \in [N].
    \end{align*}
    One readily verifies that both $| A_j(\varphi) |$ and $| A_{jj'}(\varphi) |$, $j, j' \in [N]$, are affine between any two consecutive points in $\setvarphi(\beta, \bm{\pi})$. In other words, the maximum privacy shortfall is piecewise affine with breakpoints $\setvarphi(\beta, \bm{\pi})$ or a subset thereof, which implies that its maximum must be attained at one of the points $\varphi \in \setvarphi(\beta, \bm{\pi})$. This concludes the proof.
\end{proof}

\begin{lemmaA}\label{lemma:algo_p3}
Algorithm~\ref{alg:worst-event} can be implemented such that it terminates in time $\mathcal{O}(N^3)$.
\end{lemmaA}

\begin{proof}
Since all individual steps in Algorithm~\ref{alg:worst-event} take constant time, the runtime of the algorithm is determined by the numbers of iterations in the outer and inner loops. In the na\"ive implementation of the main text, both loops comprise $\mathcal{O}(N^2)$ iterations, and thus the overall complexity of that implementation is $\mathcal{O}(N^4)$. We show in this proof that the inner loop can be implemented such that it comprises $\mathcal{O}(N)$ iterations only, which implies the statement of the lemma.

Note that the inner loop in Algorithm~\ref{alg:worst-event} constructs all events $A_j$, $j \in [N]$, in time $\mathcal{O}(N)$, and thus we only need to consider the construction of the events $A_{jj'}$, $j, j' \in [N]$. Instead of the na\"ive implementation from the main text, which probes all pairs of subsets $(j, j') \in [N]^2$, we consider the following variant of the Bentley-Ottmann algorithm used to identify crossings in a set of line segments \citep[\S 2]{BKOS00:comp_geometry}: We merge the two lists of tuples $\{ (\pi_j \cdot \beta, 1) : j \in [N + 1] \}$ and $\{ (\pi_{j'} \cdot \beta + \varphi, 2) : j' \in [N + 1] \}$ in order of non-decreasing first component; since each list is already sorted, this can be achieved in time $\mathcal{O} (N)$. We initialize the two index counters $j_1 = j_2 = 0$ and loop through the entire merged list of tuples once in sorted order. Whenever we encounter an element $(\pi_j, 1)$, we update $j_1 \leftarrow \pi_j$; whenever we encounter an element $(\pi_{j'}, 2)$, we update $j_2 \leftarrow \pi_{j'}$. After each update, we consider the intersection $A_{jj'} = \Pi_{j_1} (\beta) \cap (\Pi_{j_2} (\beta) + \varphi)$ for possible inclusion in the worst-case event A whenever $(j_1, j_2) \neq (0, 0)$. Since the merged list of tuples has length $2N + 2$, the overall algorithm evidently runs in time $\mathcal{O} (N)$.
\end{proof}

\noindent \textbf{Proof of Proposition~\ref{alg:correct}} $\;$
The proof immediately follows from Lemmas~\ref{lemma:algo_p1}, \ref{lemma:algo_p2} and \ref{lemma:algo_p3}.
\qed

\subsection{Bounding~\ref{problem:integral_nonlinear} and~\ref{problem:dual_integral_nonlinear} with Non-Uniform Partitions}\label{app:dependent_X}

We next extend the non-uniform upper and lower bounding problems \ref{new_L-UB} and \ref{new_LB} of Section~\ref{sec:non-identical} to the data dependent case. We obtain the following upper bound on problem~\ref{problem:integral_nonlinear},
\begin{align}\label{nonlinear_new_L-UB}\tag{$\mathrm{P'}(\bm{\pi}, \beta)$}
\mspace{-30mu}
    \begin{array}{cll}
        \underset{p}{\mathrm{minimize}} & \displaystyle \beta\cdot \sum_{k \in [K]} \omega_k(\beta)
        \cdot \Big[\sum\limits_{j \in [N]} c_j(\bm{\pi}, \beta) \cdot p_k(j) \Big]  \\
        \mathrm{subject\; to} &\displaystyle p_k : [N] \mapsto \mathbb{R}_+, \ \sum_{j \in [N]} p_k(j) = 1, \ k \in [K] \\[1.5em]
        &\displaystyle \sum\limits_{j \in [N]} p_k(j) \cdot \dfrac{|A \cap \Pi_{j}(\beta)|}{|\Pi_{j}(\beta)|} \leq e^\varepsilon \cdot \sum\limits_{j \in [N]} p_m(j)\cdot \dfrac{|(A - \varphi ) \cap \Pi_{j}(\beta)|}{|\Pi_{j}(\beta)|} + \delta  \\
        &  \pushright{\forall k,m \in [K], \ \forall (\varphi, A) \in \mathcal{E}'_{km}(L, \beta)},
    \end{array}
\end{align}
as well as the following lower bound on problem~\ref{problem:dual_integral_nonlinear},
\begin{align}\label{nonlinear_new_LB}\tag{$\mathrm{D'}(\bm{\pi}, \beta)$}
\mspace{-30mu}
    \begin{array}{cll}
        \underset{p}{\mathrm{minimize}} & \displaystyle \beta\cdot \sum_{k \in [K]} \omega_k(\beta)
        \cdot \Big[\sum\limits_{j \in \mathfrak{N}} \underline{c}_j(\bm{\pi}, \beta) \cdot p_k(j) \Big]  \\
        \mathrm{subject\; to} &\displaystyle p_k : \mathfrak{N} \mapsto \mathbb{R}_+, \ \sum_{j \in \mathfrak{N}} p_k(j) = 1, \ k \in [K] \\[1.5em]
        &\displaystyle \sum\limits_{j \in \mathfrak{N}} p_k(j) \cdot \dfrac{|A \cap \Pi_{j}(\beta)|}{|\Pi_{j}(\beta)|} \leq e^\varepsilon \cdot \sum\limits_{j \in \mathfrak{N}} p_m(j)\cdot \dfrac{|(A - \varphi ) \cap \Pi_{j}(\beta)|}{|\Pi_{j}(\beta)|} + \delta  \\
        &  \pushright{\forall k,m \in [K], \ \forall (\varphi, A) \in \mathcal{E}'_{km}(L, \beta)}.
    \end{array}
\end{align}

Algorithm~\ref{alg:worst-event_dep} extends the ideas of Algorithm~\ref{alg:worst-event} to the data dependent setting; as before, extending the domain of the decisions $p$ from $[N]$ to $\mathfrak{N}$ allows us to employ the same algorithm to solve the lower bounding problem \ref{nonlinear_new_LB} as well. Similar arguments as in the previous section show that Algorithm~\ref{alg:worst-event_dep} terminates in time $\mathcal{O}(K^2 \cdot N)$. We explain in the GitHub supplement of this paper how the bounding problems \ref{nonlinear_new_L-UB} and \ref{nonlinear_new_LB}, as well as Algorithm~\ref{alg:worst-event_dep}, can be generalized to account for non-uniform partitions of the set of possible query outputs $\Phi$ as well; this reduces computation times when the granularity parameter $\beta$ is small.

\begin{algorithm}[tb]
\SetAlgoLined
\SetKwInOut{Input}{input}
\SetKwInOut{Output}{output}
\Input{$\bm{\pi}, \ \beta, \ p, \ \Delta f$}
\Output{constraint $(\varphi^\star, A^\star)$ with maximum privacy shortfall $V(\varphi^\star, A^\star)$}
Initialize $V^\star = 0$\;
\For {$k, m \in [K]$}
{
\For {$\varphi \in \{ (m - k - 1)\cdot \beta, \ (m-k)\cdot \beta,\  (m-k+1)\cdot \beta \} \cap [- \Delta f, \Delta f]$}
{
Initialize $A = \emptyset$ and $V = 0$\;
\For{$j = 1,\ldots, N$}
{
Let $A_j = \Pi_j(\beta) \setminus [-L\cdot \beta + \varphi, (L+1)\cdot\beta + \varphi)$ and update
$$
A = A \cup A_j, \quad
    V = V + |A_j|\cdot  \dfrac{p_k(j)}{|\Pi_{j}(\beta)|}.
$$
    \For{$j' = 1,\ldots, N$}
    {
     \If{$p_k(j)/|\Pi_{j}(\beta)| > e^\varepsilon \cdot p_m(j')/|\Pi_{j'}(\beta)|$}{
     \vspace{1mm}
        Let $A_{jj'} = \Pi_j(\beta) \cap (\Pi_{j'}(\beta) + \varphi)$ and update
    $$ 
    A = A \cup A_{jj'}, \quad V = V + |A_{jj'}| \cdot \left[\dfrac{p_k(j)}{|\Pi_{j}(\beta)|} - e^\varepsilon \cdot \dfrac{p_m(j')}{|\Pi_{j'}(\beta)|}\right].
    $$
    }
    }
}
\If{$V > V^\star$}{
Update $\varphi^\star = \varphi$,  $A^\star = A$ and $V^\star = V$.
}
}
}

\Return $(\varphi^\star, A^\star)$ and $V^\star(\varphi, A) = V^\star - \delta$.
\caption{\textit{Identification of a constraint in~\ref{nonlinear_new_L-UB} with maximum privacy shortfall}}
\label{alg:worst-event_dep}
\end{algorithm} 

\clearpage

\section{Additional Numerical Experiments}\label{app_numerical}

Table~\ref{tab:tradefoff_main} summarized the suboptimality of the best performing benchmark mechanisms by summing the upper and lower bound gaps. Tables~\ref{tab:tradefoff_ub} and~\ref{tab:tradefoff_lb} report these gaps separately. 
\begin{table}[!h]
        \small
        \setlength{\arrayrulewidth}{0.2mm}
        \begin{center}
        \resizebox{1.0\columnwidth}{!}{%
    \hskip-0.8cm
\renewcommand{\arraystretch}{2}
\begin{tabular}{c |  r | c  c c c c c c c c c|}
\multicolumn{2}{c}{} & \multicolumn{10}{c}{\Large $\delta$} 
 \\
\cline{3-12}
\multicolumn{2}{c|}{} & \textbf{0.005} & \textbf{0.010} &\textbf{0.020} & \textbf{0.050} & \textbf{0.100} & \textbf{0.200} & \textbf{0.250} &\textbf{0.300} &\textbf{0.500} &\textbf{0.750}\\
\cline{2-12}
\multirow{10}{*}{\Large $\varepsilon$} & {\textbf{0.005}} & {\cellcolor[gray]{0.98} 0.17\%}& {\cellcolor[gray]{0.985} 0.12\%}& {\cellcolor[gray]{0.995} 0.03\%}& {\cellcolor[gray]{0.975} 0.26\%}& {\cellcolor[gray]{0.97} 0.26\%}& {\cellcolor[gray]{0.84} 0.65\%}& {\cellcolor[gray]{0.845} 0.64\%}& {\cellcolor[gray]{0.85} 0.63\%}& {\cellcolor[gray]{0.92} 0.49\%}& {\cellcolor[gray]{0.615} 8.37\%}
\\ \hhline{~~|}
\cline{3-12}
& {\textbf{0.010}} & {\cellcolor[gray]{0.94} 0.41\%}& {\cellcolor[gray]{0.955} 0.33\%}& {\cellcolor[gray]{0.87} 0.58\%}& {\cellcolor[gray]{0.9} 0.54\%}& {\cellcolor[gray]{0.91} 0.54\%}& {\cellcolor[gray]{0.835} 0.68\%}& {\cellcolor[gray]{0.895} 0.55\%}& {\cellcolor[gray]{0.89} 0.55\%}& {\cellcolor[gray]{0.885} 0.55\%}& {\cellcolor[gray]{0.655} 5.26\%}
\\
\cline{3-12}
& {\textbf{0.020}} & {\cellcolor[gray]{0.96} 0.28\%}& {\cellcolor[gray]{0.935} 0.41\%}& {\cellcolor[gray]{0.99} 0.04\%}& {\cellcolor[gray]{0.865} 0.58\%}& {\cellcolor[gray]{0.875} 0.58\%}& {\cellcolor[gray]{0.825} 0.73\%}& {\cellcolor[gray]{0.855} 0.61\%}& {\cellcolor[gray]{0.765} 1.24\%}& {\cellcolor[gray]{0.86} 0.61\%}& {\cellcolor[gray]{0.605} 8.53\%}
\\
\cline{3-12}
& {\textbf{0.050}} & {\cellcolor[gray]{0.915} 0.50\%}& {\cellcolor[gray]{0.965} 0.27\%}& {\cellcolor[gray]{0.945} 0.38\%}& {\cellcolor[gray]{0.93} 0.45\%}& {\cellcolor[gray]{0.83} 0.73\%}& {\cellcolor[gray]{0.8} 0.89\%}& {\cellcolor[gray]{0.815} 0.80\%}& {\cellcolor[gray]{0.755} 1.48\%}& {\cellcolor[gray]{0.81} 0.87\%}& {\cellcolor[gray]{0.595} 8.73\%}
\\
\cline{3-12}
& {\textbf{0.100}} & {\cellcolor[gray]{0.95} 0.38\%}& {\cellcolor[gray]{0.925} 0.45\%}& {\cellcolor[gray]{0.88} 0.58\%}& {\cellcolor[gray]{0.905} 0.54\%}& {\cellcolor[gray]{0.795} 0.90\%}& {\cellcolor[gray]{0.775} 1.14\%}& {\cellcolor[gray]{0.77} 1.15\%}& {\cellcolor[gray]{0.735} 1.88\%}& {\cellcolor[gray]{0.76} 1.25\%}& {\cellcolor[gray]{0.59} 9.12\%}
\\
\cline{3-12}
&{\textbf{0.200}} & {\cellcolor[gray]{0.82} 0.78\%}& {\cellcolor[gray]{0.805} 0.87\%}& {\cellcolor[gray]{0.78} 1.02\%}& {\cellcolor[gray]{0.79} 0.92\%}& {\cellcolor[gray]{0.785} 0.96\%}& {\cellcolor[gray]{0.75} 1.66\%}& {\cellcolor[gray]{0.725} 1.92\%}& {\cellcolor[gray]{0.71} 2.61\%}& {\cellcolor[gray]{0.745} 1.80\%}& {\cellcolor[gray]{0.575} 10.24\%}
\\
\cline{3-12}
& {\textbf{0.500}} & {\cellcolor[gray]{0.74} 1.81\%}& {\cellcolor[gray]{0.73} 1.91\%}& {\cellcolor[gray]{0.72} 2.09\%}& {\cellcolor[gray]{0.715} 2.58\%}& {\cellcolor[gray]{0.705} 2.71\%}& {\cellcolor[gray]{0.7} 3.35\%}& {\cellcolor[gray]{0.695} 3.92\%}& {\cellcolor[gray]{0.69} 4.08\%}& {\cellcolor[gray]{0.685} 4.33\%}& {\cellcolor[gray]{0.54} 12.91\%}
\\
\cline{3-12}
& {\textbf{1.000}} & {\cellcolor[gray]{0.675} 4.70\%}& {\cellcolor[gray]{0.67} 4.72\%}& {\cellcolor[gray]{0.68} 4.66\%}& {\cellcolor[gray]{0.665} 4.84\%}& {\cellcolor[gray]{0.66} 5.10\%}& {\cellcolor[gray]{0.65} 5.82\%}& {\cellcolor[gray]{0.645} 6.03\%}& {\cellcolor[gray]{0.64} 6.35\%}& {\cellcolor[gray]{0.62} 8.09\%}& {\cellcolor[gray]{0.51} 16.38\%}
\\
\cline{3-12}
&{\textbf{2.000}} & {\cellcolor[gray]{0.635} 7.99\%}& {\cellcolor[gray]{0.63} 8.02\%}& {\cellcolor[gray]{0.625} 8.09\%}& {\cellcolor[gray]{0.61} 8.39\%}& {\cellcolor[gray]{0.6} 8.70\%}& {\cellcolor[gray]{0.585} 9.60\%}& {\cellcolor[gray]{0.58} 10.18\%}& {\cellcolor[gray]{0.57} 10.86\%}& {\cellcolor[gray]{0.535} 13.05\%}& {\cellcolor[gray]{0.5} 19.74\%}
\\
\cline{3-12}
& {\textbf{5.000}} & {\cellcolor[gray]{0.565} 12.25\%}& {\cellcolor[gray]{0.56} 12.28\%}& {\cellcolor[gray]{0.555} 12.33\%}& {\cellcolor[gray]{0.55} 12.50\%}& {\cellcolor[gray]{0.545} 12.78\%}& {\cellcolor[gray]{0.53} 13.37\%}& {\cellcolor[gray]{0.525} 13.67\%}& {\cellcolor[gray]{0.52} 13.97\%}& {\cellcolor[gray]{0.515} 14.86\%}& {\cellcolor[gray]{0.505} 16.74\%}
\\
\cline{2-12}
\end{tabular}}
\end{center}
~\\[-6mm]
\caption{\textit{Upper bound suboptimality of the best performing benchmark mechanisms on synthetic data independent instances with $\Delta f = 1$, $\ell_1$-loss and various combinations of $\varepsilon$ and $\delta$.}}
\label{tab:tradefoff_ub}
\end{table}
\noindent\vspace*{-0.2cm}
\begin{table}[!h]
        \small
        \setlength{\arrayrulewidth}{0.2mm}
        \begin{center}
        \resizebox{1.0\columnwidth}{!}{%
    \hskip-0.8cm
\renewcommand{\arraystretch}{2}
\begin{tabular}{c |  r | c  c c c c c c c c c|}
\multicolumn{2}{c}{} & \multicolumn{10}{c}{\Large $\delta$} 
 \\
\cline{3-12}
\multicolumn{2}{c|}{} & \textbf{0.005} & \textbf{0.010} &\textbf{0.020} & \textbf{0.050} & \textbf{0.100} & \textbf{0.200} & \textbf{0.250} &\textbf{0.300} &\textbf{0.500} &\textbf{0.750}\\
\cline{2-12}
\multirow{10}{*}{\Large $\varepsilon$} &  {\textbf{0.005}} &{\cellcolor[gray]{0.97} 1.71\%}& {\cellcolor[gray]{0.99} 1.09\%}& {\cellcolor[gray]{0.98} 1.18\%}& {\cellcolor[gray]{0.885} 5.84\%}& {\cellcolor[gray]{0.765} 18.27\%}& {\cellcolor[gray]{0.925} 2.90\%}& {\cellcolor[gray]{0.535} 48.94\%}& {\cellcolor[gray]{0.7} 22.55\%}& {\cellcolor[gray]{0.52} 49.47\%}& {\cellcolor[gray]{0.655} 24.97\%}
\\ \hhline{~~|}
\cline{3-12}
& {\textbf{0.010}} & {\cellcolor[gray]{0.935} 2.48\%}& {\cellcolor[gray]{0.985} 1.09\%}& {\cellcolor[gray]{0.855} 7.42\%}& {\cellcolor[gray]{0.965} 1.77\%}& {\cellcolor[gray]{0.79} 16.55\%}& {\cellcolor[gray]{0.945} 2.40\%}& {\cellcolor[gray]{0.545} 48.63\%}& {\cellcolor[gray]{0.705} 22.48\%}& {\cellcolor[gray]{0.525} 49.36\%}& {\cellcolor[gray]{0.625} 28.09\%}
\\
\cline{3-12}
& {\textbf{0.020}} & {\cellcolor[gray]{0.93} 2.56\%}& {\cellcolor[gray]{0.94} 2.42\%}& {\cellcolor[gray]{0.995} 0.52\%}& {\cellcolor[gray]{0.81} 14.50\%}& {\cellcolor[gray]{0.825} 13.61\%}& {\cellcolor[gray]{0.5} 66.71\%}& {\cellcolor[gray]{0.555} 47.78\%}& {\cellcolor[gray]{0.71} 21.50\%}& {\cellcolor[gray]{0.53} 49.23\%}& {\cellcolor[gray]{0.66} 24.84\%}
\\
\cline{3-12}
& {\textbf{0.050}} & {\cellcolor[gray]{0.975} 1.35\%}& {\cellcolor[gray]{0.96} 1.84\%}& {\cellcolor[gray]{0.955} 2.18\%}& {\cellcolor[gray]{0.77} 17.92\%}& {\cellcolor[gray]{0.895} 5.50\%}& {\cellcolor[gray]{0.505} 65.14\%}& {\cellcolor[gray]{0.565} 45.35\%}& {\cellcolor[gray]{0.735} 20.47\%}& {\cellcolor[gray]{0.54} 48.72\%}& {\cellcolor[gray]{0.665} 24.69\%}
\\
\cline{3-12}
& {\textbf{0.100}} & {\cellcolor[gray]{0.91} 4.40\%}& {\cellcolor[gray]{0.92} 3.73\%}& {\cellcolor[gray]{0.845} 8.10\%}& {\cellcolor[gray]{0.76} 18.83\%}& {\cellcolor[gray]{0.62} 32.42\%}& {\cellcolor[gray]{0.51} 62.71\%}& {\cellcolor[gray]{0.57} 41.72\%}& {\cellcolor[gray]{0.75} 18.97\%}& {\cellcolor[gray]{0.55} 47.92\%}& {\cellcolor[gray]{0.67} 24.38\%}
\\
\cline{3-12}
&{\textbf{0.200}} & {\cellcolor[gray]{0.835} 9.65\%}& {\cellcolor[gray]{0.84} 8.73\%}& {\cellcolor[gray]{0.85} 7.88\%}& {\cellcolor[gray]{0.795} 16.37\%}& {\cellcolor[gray]{0.755} 18.86\%}& {\cellcolor[gray]{0.515} 58.51\%}& {\cellcolor[gray]{0.59} 35.79\%}& {\cellcolor[gray]{0.785} 16.74\%}& {\cellcolor[gray]{0.56} 46.53\%}& {\cellcolor[gray]{0.685} 23.39\%}
\\
\cline{3-12}
& {\textbf{0.500}} & {\cellcolor[gray]{0.72} 21.29\%}& {\cellcolor[gray]{0.74} 19.50\%}& {\cellcolor[gray]{0.69} 23.36\%}& {\cellcolor[gray]{0.82} 13.78\%}& {\cellcolor[gray]{0.585} 35.99\%}& {\cellcolor[gray]{0.58} 39.20\%}& {\cellcolor[gray]{0.645} 25.60\%}& {\cellcolor[gray]{0.815} 14.41\%}& {\cellcolor[gray]{0.575} 41.52\%}& {\cellcolor[gray]{0.725} 20.88\%}
\\
\cline{3-12}
& {\textbf{1.000}} & {\cellcolor[gray]{0.605} 35.41\%}& {\cellcolor[gray]{0.61} 35.01\%}& {\cellcolor[gray]{0.595} 35.57\%}& {\cellcolor[gray]{0.6} 35.56\%}& {\cellcolor[gray]{0.68} 23.52\%}& {\cellcolor[gray]{0.63} 26.71\%}& {\cellcolor[gray]{0.73} 20.59\%}& {\cellcolor[gray]{0.8} 15.08\%}& {\cellcolor[gray]{0.615} 33.71\%}& {\cellcolor[gray]{0.775} 16.98\%}
\\
\cline{3-12}
&{\textbf{2.000}} & {\cellcolor[gray]{0.64} 25.97\%}& {\cellcolor[gray]{0.635} 26.16\%}& {\cellcolor[gray]{0.65} 25.36\%}& {\cellcolor[gray]{0.695} 23.24\%}& {\cellcolor[gray]{0.675} 23.76\%}& {\cellcolor[gray]{0.745} 19.10\%}& {\cellcolor[gray]{0.78} 16.94\%}& {\cellcolor[gray]{0.805} 14.81\%}& {\cellcolor[gray]{0.715} 21.29\%}& {\cellcolor[gray]{0.83} 10.77\%}
\\
\cline{3-12}
& {\textbf{5.000}} & {\cellcolor[gray]{0.86} 7.07\%}& {\cellcolor[gray]{0.865} 7.03\%}& {\cellcolor[gray]{0.87} 6.96\%}& {\cellcolor[gray]{0.875} 6.74\%}& {\cellcolor[gray]{0.88} 6.37\%}& {\cellcolor[gray]{0.89} 5.64\%}& {\cellcolor[gray]{0.9} 5.27\%}& {\cellcolor[gray]{0.905} 4.91\%}& {\cellcolor[gray]{0.915} 3.79\%}& {\cellcolor[gray]{0.95} 2.33\%}
\\
\cline{2-12}
\end{tabular}}
\end{center}
~\\[-6mm]
\caption{\textit{Lower bound suboptimality of the best performing benchmark mechanisms on synthetic data independent instances with $\Delta f = 1$, $\ell_1$-loss and various combinations of $\varepsilon$ and $\delta$.}}
\label{tab:tradefoff_lb}
\end{table} 
\newpage
{\color{black}
Table~\ref{tab:tradefoff_main_ell2} reports results similar to those of Table~\ref{tab:tradefoff_main} for the $\ell_2$-loss. Moreover, our GitHub repository includes results analogous to those of Tables~\ref{tab:tradefoff_ub} and~\ref{tab:tradefoff_lb}. We note that if we replace $\max \{ O, 1 \}$ with $O$ in the denominator of the optimality gap formula, then the gaps in Table~\ref{tab:tradefoff_main_ell2} increase to more than $850\%$ for $(\varepsilon, \delta) = (5, 0.75)$.}
\begin{table}[!h]
        \small
        \setlength{\arrayrulewidth}{0.2mm}
        \begin{center}
        \resizebox{1.0\columnwidth}{!}{%
    \hskip-0.8cm
\renewcommand{\arraystretch}{1.82} 
\begin{tabular}{c |  r | c  c c c c c c c c c|}
\multicolumn{2}{c}{} & \multicolumn{10}{c}{\Large $\delta$} 
 \\
\cline{3-12}
\multicolumn{2}{c|}{} & \textbf{0.005} & \textbf{0.010} &\textbf{0.020} & \textbf{0.050} & \textbf{0.100} & \textbf{0.200} & \textbf{0.250} &\textbf{0.300} &\textbf{0.500} &\textbf{0.750}\\
\cline{2-12}
\multirow{10}{*}{\Large $\varepsilon$} &  {\textbf{0.005}} & {\cellcolor[gray]{0.972} 2.76\%}& {\cellcolor[gray]{0.984} 1.83\%}& {\cellcolor[gray]{0.988} 1.71\%}& {\cellcolor[gray]{0.896} 9.12\%}& {\cellcolor[gray]{0.764} 26.58\%}& {\cellcolor[gray]{0.6} 83.29\%}& {\cellcolor[gray]{0.624} 66.77\%}& {\cellcolor[gray]{0.708} 34.72\%}& {\cellcolor[gray]{0.7} 36.59\%}& {\cellcolor[gray]{0.832} 18.06\%}
\\ \hhline{~~|}
\cline{3-12}
& {\textbf{0.010}} & {\cellcolor[gray]{0.952} 4.29\%}& {\cellcolor[gray]{0.976} 2.33\%}& {\cellcolor[gray]{0.876} 12.41\%}& {\cellcolor[gray]{0.96} 3.82\%}& {\cellcolor[gray]{0.78} 24.47\%}& {\cellcolor[gray]{0.608} 83.05\%}& {\cellcolor[gray]{0.628} 66.55\%}& {\cellcolor[gray]{0.712} 34.34\%}& {\cellcolor[gray]{0.672} 39.28\%}& {\cellcolor[gray]{0.864} 15.73\%}
\\
\cline{3-12}
& {\textbf{0.020}} & {\cellcolor[gray]{0.956} 4.04\%}& {\cellcolor[gray]{0.948} 4.39\%}& {\cellcolor[gray]{0.996} 0.32\%}& {\cellcolor[gray]{0.796} 23.37\%}& {\cellcolor[gray]{0.8} 22.76\%}& {\cellcolor[gray]{0.604} 83.06\%}& {\cellcolor[gray]{0.632} 66.06\%}& {\cellcolor[gray]{0.72} 33.60\%}& {\cellcolor[gray]{0.676} 39.25\%}& {\cellcolor[gray]{0.816} 19.53\%}
\\
\cline{3-12}
& {\textbf{0.050}} & {\cellcolor[gray]{0.992} 1.66\%}& {\cellcolor[gray]{0.98} 2.03\%}& {\cellcolor[gray]{0.968} 2.86\%}& {\cellcolor[gray]{0.744} 27.96\%}& {\cellcolor[gray]{0.888} 10.67\%}& {\cellcolor[gray]{0.612} 82.06\%}& {\cellcolor[gray]{0.636} 62.81\%}& {\cellcolor[gray]{0.724} 31.56\%}& {\cellcolor[gray]{0.68} 39.05\%}& {\cellcolor[gray]{0.836} 18.03\%}
\\
\cline{3-12}
& {\textbf{0.100}} & {\cellcolor[gray]{0.944} 4.79\%}& {\cellcolor[gray]{0.964} 3.81\%}& {\cellcolor[gray]{0.88} 11.48\%}& {\cellcolor[gray]{0.732} 29.04\%}& {\cellcolor[gray]{0.656} 47.70\%}& {\cellcolor[gray]{0.616} 78.46\%}& {\cellcolor[gray]{0.64} 60.47\%}& {\cellcolor[gray]{0.748} 27.93\%}& {\cellcolor[gray]{0.688} 38.74\%}& {\cellcolor[gray]{0.828} 18.13\%}
\\
\cline{3-12}
&{\textbf{0.200}} & {\cellcolor[gray]{0.884} 11.32\%}& {\cellcolor[gray]{0.892} 9.61\%}& {\cellcolor[gray]{0.9} 8.40\%}& {\cellcolor[gray]{0.784} 24.36\%}& {\cellcolor[gray]{0.772} 25.68\%}& {\cellcolor[gray]{0.62} 73.75\%}& {\cellcolor[gray]{0.648} 55.80\%}& {\cellcolor[gray]{0.788} 24.00\%}& {\cellcolor[gray]{0.692} 38.23\%}& {\cellcolor[gray]{0.82} 19.51\%}
\\
\cline{3-12}
& {\textbf{0.500}} & {\cellcolor[gray]{0.804} 21.38\%}& {\cellcolor[gray]{0.84} 17.35\%}& {\cellcolor[gray]{0.768} 25.81\%}& {\cellcolor[gray]{0.668} 41.55\%}& {\cellcolor[gray]{0.66} 46.75\%}& {\cellcolor[gray]{0.644} 58.92\%}& {\cellcolor[gray]{0.716} 34.29\%}& {\cellcolor[gray]{0.852} 15.86\%}& {\cellcolor[gray]{0.704} 34.77\%}& {\cellcolor[gray]{0.824} 18.67\%}
\\
\cline{3-12}
& {\textbf{1.000}} & {\cellcolor[gray]{0.696} 37.40\%}& {\cellcolor[gray]{0.684} 38.99\%}& {\cellcolor[gray]{0.664} 43.91\%}& {\cellcolor[gray]{0.652} 55.20\%}& {\cellcolor[gray]{0.756} 27.47\%}& {\cellcolor[gray]{0.728} 31.26\%}& {\cellcolor[gray]{0.808} 20.88\%}& {\cellcolor[gray]{0.872} 12.62\%}& {\cellcolor[gray]{0.752} 27.77\%}& {\cellcolor[gray]{0.844} 17.07\%}
\\
\cline{3-12}
&{\textbf{2.000}} & {\cellcolor[gray]{0.74} 28.11\%}& {\cellcolor[gray]{0.736} 28.60\%}& {\cellcolor[gray]{0.76} 26.62\%}& {\cellcolor[gray]{0.792} 23.70\%}& {\cellcolor[gray]{0.776} 24.67\%}& {\cellcolor[gray]{0.848} 16.97\%}& {\cellcolor[gray]{0.856} 15.78\%}& {\cellcolor[gray]{0.868} 12.87\%}& {\cellcolor[gray]{0.812} 20.87\%}& {\cellcolor[gray]{0.86} 15.76\%}
\\
\cline{3-12}
& {\textbf{5.000}} & {\cellcolor[gray]{0.908} 7.85\%}& {\cellcolor[gray]{0.912} 7.83\%}& {\cellcolor[gray]{0.916} 7.78\%}& {\cellcolor[gray]{0.904} 7.90\%}& {\cellcolor[gray]{0.92} 7.74\%}& {\cellcolor[gray]{0.924} 7.47\%}& {\cellcolor[gray]{0.928} 7.35\%}& {\cellcolor[gray]{0.932} 7.23\%}& {\cellcolor[gray]{0.936} 7.19\%}& {\cellcolor[gray]{0.94} 7.14\%}
\\
\cline{2-12}
\end{tabular}}
\end{center}
~\\[-6mm]
\caption{\textit{\color{black} Suboptimality of the best performing benchmark mechanisms on synthetic data independent instances with $\Delta f = 1$, $\ell_2$-loss and various combinations of $\varepsilon$ and $\delta$.}}
\label{tab:tradefoff_main_ell2}
\end{table}

\end{document}